\documentclass[12pt, a4paaper, oneside]{article}
\usepackage{lineno,hyperref}

\usepackage[margin=1in]{geometry}  
\usepackage{graphicx}              
\usepackage{epsfig}
\usepackage{amsmath}               
\usepackage{amssymb}
\usepackage{epstopdf}
\usepackage{verbatim}
\usepackage{mathptmx}
\usepackage{mathalpha}
\usepackage{mathtools}
\usepackage{graphicx, times}
\usepackage{amsfonts}              
\usepackage{amsthm}                
\usepackage{multicol}
\usepackage{algorithm}
\usepackage{algorithmic}
\usepackage{varwidth}
\usepackage{parskip}
\usepackage{hyperref}
\usepackage{rotating}
\usepackage{multirow}
\usepackage{pdflscape}
\usepackage[numbers]{natbib}
\usepackage{caption}
\usepackage{xcolor}
\usepackage{color}
\usepackage{comment}
\usepackage{subcaption}
\newcommand*{\rom}[1]{\expandafter\@\romannumeral #1}

\newcommand{\bea}{\begin{eqnarray}}
	\newcommand{\eea}{\end{eqnarray}}
\newcommand{\bee}{\begin{eqnarray*}}
	\newcommand{\eee}{\end{eqnarray*}}
\newcommand{\nn}{\nonumber}
\newcommand{\lb}{\label}

\newtheorem{thm}{Theorem}[section]

\begin{document}
\author{ Ashish Verma$^{1}$\footnote{ashishhaiverma@gmail.com}, Sourav Pradhan$^{1}$\footnote{spiitkgp11@gmail.com}
\vspace{.3cm}\\
${}^{1}$ Department of Mathematics,\\ Visvesvaraya National Institute of Technology, Nagpur, India.}
\date{}
\title{Queue occupancy and server size distribution of a queue length dependent vacation queue with an optional service}

\maketitle

\begin{abstract}
The discrete time queueing system is highly applicable to modern telecommunication systems, where it provides adaptive packet handling, congestion controlled security/inspection, energy efficient operation, and supports bursty traffic common in 5G, Internet of Things (IoT), and edge computing environments. In this article, we analyze an infinite-buffer discrete-time batch-arrival queue with single and multiple vacation policy where customers are served in batches, in two phases, namely first essential service (FES) and second optional service (SOS). In such systems, the FES corresponds to basic data processing or packet routing, while SOS represents secondary tasks such as encryption, error checking, data compression, or deep packet inspection that may not be necessary for every packet. Here, we derive the bivariate probability generating functions for the joint distribution of the number of packets waiting for transmission and the number are being processed immediately after the completion of both the FES and SOS. Furthermore, the complete joint distribution at arbitrary time slots, including vacation completion states, is established. Numerical illustrations demonstrate the applicability of the proposed framework, including an example with discrete phase type service time distribution. Finally, the sensitivity analysis of the key parameters on marginal system's probabilities and different performance measures have been investigated through several graphical representations.
\end{abstract}
\section{Introduction}\label{sec:1}
\hspace*{0.3 cm} Queueing models play a fundamental role in the performance evaluation of modern communication networks, wireless systems, and service platforms where traffic is random, resources are limited, and processing decisions depend on system conditions. In wireless telecommunication systems such as 5G/6G access points, Internet of Things (IoT) gateways, and edge computing servers, traffic often arrives in bursts, processing requirements vary across packets, and energy efficient operation is essential due to battery constraints and sustainability requirements. These practical considerations motivate the development of analytical frameworks that incorporate batch arrivals, adaptive service phases, and dynamic energy saving mechanisms.\\
\hspace*{0.3cm}
In many wireless applications, packets undergo essential processing operations such as routing, decoding, or minimal security verification. In addition, packets may require optional tasks such as deep packet inspection, strong encryption, additional error checking, or content filtering. Performing these optional tasks for every packet can significantly increase processing overhead and delay, especially during congestion. Hence, real world systems often execute such optional tasks selectively based on network load. This motivates the study of two phase service structures, where the First Essential Service (FES) is mandatory, while the Second Optional Service (SOS) is executed only under favorable queue conditions.\\
\hspace*{0.3cm}
Similarly, wireless access points, base stations, and network nodes commonly employ sleep or low power modes to reduce energy consumption during idle periods. These mechanisms can be modeled using queue length dependent vacations, where the server (or network device) enters a sleep mode when the system becomes empty and resumes operation only when the queue length crosses a specific threshold. Such vacation models help to capture the energy delay trade off inherent in green wireless communication.\\
\hspace*{0.3cm}
The present work introduces and analyzes a discrete time single server group arrival queueing system with FES as well as SOS, batch size dependent service time and queue length dependent vacation time, which is denoted by $Geo^{X}/G_{n,k}^{(a,b)}/1$. Arrivals follow a compound Bernoulli process, representing the bursty nature of wireless traffic. In FES, service times follow general distribution and depend on the number of customers in a batch taken by the server for a particular cycle of processing, enabling congestion-sensitive service control policy. The single server starts processing of SOS, after the completion of FES of a particular batch, by allowing a part of those customers (already got service in FES) upto a maximum pre defined threshold. When the system becomes empty, the server enters a vacation period and resumes the service only when the number reaches to a specified threshold, capturing energy efficient sleep mode behavior.
\subsection{Literature review}
Discrete and continuous-time batch-service queues have been extensively studied due to their wide-ranging applications in telecommunications, manufacturing, wireless sensor and IoT networks, healthcare, and laboratory systems etc.
A number of related works on finite and infinite buffer batch service queues, with or without batch size dependent service times, are summarized below. In a series of articles, Banerjee et al. (2015) \cite{banerjee2015analysis}; Gupta et al. (2020) \cite{gupta2020finite}; Tamrakar and Banerjee (2025) \cite{tamrakar2025state} examined both discrete and continuous-time batch-service models with finite waiting capacities and presented joint distributions of queue length and server content, along with several performance measures of practical significance. For systems with infinite waiting space, Pradhan and Gupta (2019) \cite{pradhan2019infinite}; Nandy and Pradhan (2023) \cite{nandy2023stationary}; Karan and Pradhan (2025)  \cite{karan2025queue}; derived bivariate probability generating functions for the joint queue and server size, from which complete joint distributions and various performance measures were subsequently obtained.\\
\hspace*{0.3cm}
In another line of work, Gupta et al. (2021) \cite{gupta2021complete} have investigated a discrete-time batch-service queue with a discrete Markovian arrival process (MAP) and have derived vector generating functions and detailed extraction procedures for the associated probability vectors. Bulk-service queues have also been explored in several other studies: for instance, Tadj et al. (2026) \cite{tadj2006quorum}; Tadj and Ke (2008) \cite{tadj2008hysteretic} analyzed the systems with setup time, N-policy, Bernoulli vacations, and hysteretic bulk quorum service with optional re-service; Claeys et al. (2013) \cite{claeys2013versatile} examined discrete-time batch service queues and obtained joint probability generating functions (pgfs) without explicitly extracted probabilities; Yu et al. (2015) \cite{yu2015algorithm}  proposed an algorithm to compute queue length distributions and Vadivu and Arumuganathan (2015) \cite{vadivu2015cost} performed a cost analysis for a bulk-service queue with MAP input and close-down time.
The detailed modeling, methodologies and numerical techniques are comprehensively discussed in the monograph by Chaudhry (2000) \cite{chaudhry2000numerical}, while an extensive treatment of bulk-service queues appears in the book by Chaudhry and Templeton (1983) \cite{chaudhry1983first}. Additional related literature may be found in the references of these research articles.\\
\hspace*{0.3cm}In Gupta and Sikdar (2004) \cite{gupta2004finite}, Sikdar (2024) \cite{sikdar2024geom}, Sikdar and Gupta (2008) \cite{sikdar2008batch},                            several vacation queueing models were investigated, and queue length distributions at various epochs were obtained using either embedded Markov chain techniques (EMCT) or the supplementary variable technique (SVT). In another line of research, Samanta et al. (2007) \cite{samanta2007dbmap}; Samanta and Zhang (2012) \cite{samanta2012gidmsp} analyzed discrete time vacation queues under single and multiple vacation policies derived queue length distributions at different time instances. Furthermore, Jain and Bhagat (2016) \cite{jain2016mx} studied an 
$M^X/G/1$ vacation system incorporating  multi-optional services and reneging. A queue length dependent vacation model with group arrivals has been analyzed by Pradhan et al. \cite{pradhan2024versatile}, where power consumption characteristics of the system is highlighted through numerical illustrations. Nevertheless, none of these above works considered batch size dependent service times, queue length dependent vacation, joint queue and server content distribution, FES and SOS simultaneously.\\
\hspace*{0.3cm}The pioneer contribution on SOS is due to Madan (2000) \cite{madan2000mg1}, who analyzed an $M/G/1$ queue under a two-stage service policy using SVT, assuming a general service time distribution for the FES and an exponential service time for SOS. Subsequent studies extended this framework in various directions. For instance, Medhi (2002) \cite{medhi2002single} considered an $M/G/1$ system with exponentially distributed SOS and obtained the pgf of the system size; Wang (2004) \cite{wang2004mg1} incorporated server breakdown in an $M/G/1$ queue with SOS; and Choudhury and Paul (2006) \cite{choudhury2006batch}; Choudhury et al. (2009) \cite{choudhury2009npolicy} examined group arrival systems with additional optional service under N-policy and delayed repair. Using SVT and the maximum entropy approach, Singh et al. (2011) \cite{singh2011state} derived queue length distributions for a system with optional service. Systems with unreliable servers, balking, or Bernoulli vacations were investigated in Jain and Kaur (2021) \cite{jain2021bernoulli}, whereas customer impatience and working vacations were incorporated by Vijaya Laxmi and Jyothsna (2022) \cite{vijaya2022cost}, Vijaya Laxmi et al. (2023) \cite{vijaya2023retention}. However, none of these studies incorporated batch-size-dependent service mechanisms, queue length dependent vacation which play an important role in reducing high delays in communication networks.\\
\hspace*{0.3cm} Few researchers have examined batch service mechanisms for both the FES and SOS. For instance, Ayyappan and Supraja (2018) \cite{ayyappan2018batch} and Ayyappan and Nirmala (2020) \cite{ayyappan2020mx} studied an optional service queue with server unreliability, close down, and setup under N-policy. They derived the pgf of the queue length distribution at arbitrary epoch. In another direction, Deepa and Azhagappan (2018) \cite{deepa2018analysis}; Deepa and Azhagappan (2022)  \cite{deepa2022state} have analyzed a group arrival batch service queue with SOS and state dependent arrival rates and have obtained the pgf of the queue size along with expressions for the mean queue length, mean waiting time, and expected busy/idle periods. Further, Vijaya Laxmi and George (2020)  \cite{vijaya2020transient}; Vijaya Laxmi and George (2022) \cite{vijaya2022steady} have provided transient results for batch-service queues with SOS and reneging. In all these studies, the served group in FES was assumed either to proceed entirely to the SOS stage or to leave the system with a certain probability. However, practical situations may arise in which only a portion of the served batch continues to SOS, a realistic feature missing in the above works. In addition, server content distribution, which plays a crucial role in evaluating productivity in batch service systems, was not considered in any of these articles.\\
\hspace*{0.3cm} Recently, Banerjee and Lata (2024) \cite{banerjee2024complete} have obtained the complete joint distribution of a continuous-time finite-buffer bulk-service queue with SOS, but without any vacation policies. Their methodology involves construction of a TPM for the joint queue–server content distribution at FES and SOS completion instants. Each element of this TPM is itself a matrix, making the construction computationally intensive. They have further used the SVT separately to determine the joint distribution at arbitrary epochs. In contrast, our present work investigates an infinite buffer discrete time analogue of this model, augmented with single and multiple vacation policies, batch arrival processes, queue length dependent vacation mechanisms, and characterized by significantly different modeling and numerical aspects. Incorporation of such enriched vacation strategies is particularly beneficial in shared communication environments, where adaptive server inactivation based on the queue length can substantially improve energy efficiency and resource utilization. Moreover, a discrete time framework is highly applicable to satellite communication, wireless networks, and IEEE 802.11n WLANs, where data transmission occurs in uniformly spaced time slots due to packet switching protocols and where bursty, batch-type traffic arrivals naturally arise. A notable advantage of our approach is that the use of SVT alone eliminates the need for a separate TPM construction and enables direct evaluation of state probabilities at arbitrary slots in a simple and efficient manner, even under batch arrival dynamics and queue-dependent vacation behavior.
\subsection{Motivation and major contribution}
In modern service systems such as cloud computing servers, smart manufacturing units, health-care centers, and telecommunication nodes, the demand generally arrives in bulk (group or batch form) instead of individual units. For example, in a smart factory, multiple items are collected for inspection and transmitted simultaneously to a central processing unit for analysis, or several patients may reach a hospital registration counter simultaneously. To model such practical scenarios, queue with bulk arrivals 
have gained increasing attention.
Moreover, a customer may require primary essential processing first, followed by additional optional services like data validation, quality inspection, or packaging. In many real systems, more than one optional service exists, and the secondary optional service may be offered only to a fraction of customers based on the need or choice. This makes the system more realistic and flexible when compared to single optional service frameworks available in the literature.\\\hspace*{0.3cm}
Additionally, modern service centers adopt energy saving policies. During low demand, servers either turn off or enter a reduce power mode rather than staying idle. Many organizations implement queue length dependent vacations, where the server decides to take a vacation depending on the current queue size. Such a dynamic vacation model not only reduces idle energy wastage but also helps maintain an optimal balance between service cost, system responsiveness and energy consumption. \\
\hspace*{0.3cm}
Keeping these aspects in mind, we develop the model in a discrete time framework with late arrival delayed system (LAS-DA) mechanism, with the assumption of infinite waiting space. Firstly, we construct the steady state governing equations of the system using the SVT, which allows us to avoid the complicated formulation of the TPM. An additional advantage of using SVT is that it naturally leads to the joint distribution of the system at an arbitrary slot.	Secondly, we obtain the bivariate pgf at both the FES and SOS completion epochs. This result represents one of the key contributions of the present work. In addition, we derive the pgf of the queue length distribution at vacation completion slots. From these pgf, the corresponding probability distributions are explicitly extracted, and a variety of significant performance measures such as energy saving factor, mean energy consumption, mean waiting delay, and system throughput are evaluated. Moreover, we present several numerical experiments, including cases involving discrete phase-type service distributions, to demonstrate the behavioral characteristics of the proposed model.
\subsection{Organization of the paper}
\noindent This paper is organized as follows. The next section provides a comprehensive description of the proposed model. Section \ref{PA} highlights its potential real-world applications. The governing system equations are formulated in Section \ref{SE}, and the bivariate pgf, joint distribution, and stability condition are derived in Section \ref{SC}. Special cases and their validation with existing literature are discussed in Section \ref{PC}. Section \ref{RS} presents the arbitrary-slot state probabilities along with the corresponding performance measures in Section \ref{PI}. Numerical experiments are reported in Section \ref{NI}. Finally, concluding remarks and future research directions are summarized at the end of the paper.
\section{Description of proposed mathematical model}\lb{MM}
\begin{itemize}
	
	\item \emph{\textbf{Discrete-time set-up:}} The time axis is partitioned into equal intervals, known as slots. Without loss of generality, each slot is taken to be of unit length and the time scale is indexed by $0,1,2,\ldots,k,\ldots$. The model is examined under the late arrival delayed access (LAS-DA) framework, where arrivals can occur during the interval $(k-,k)$ and departures may take place during $(k,k+)$. A distinguishing feature of discrete-time systems is that an arrival and a departure may simultaneously occur at the boundary of a slot. The different time epochs of the LAS-DA mechanism are illustrated in Figure 1.
	\begin{figure}[h!]
		\begin{center}
			\includegraphics[scale=0.4]{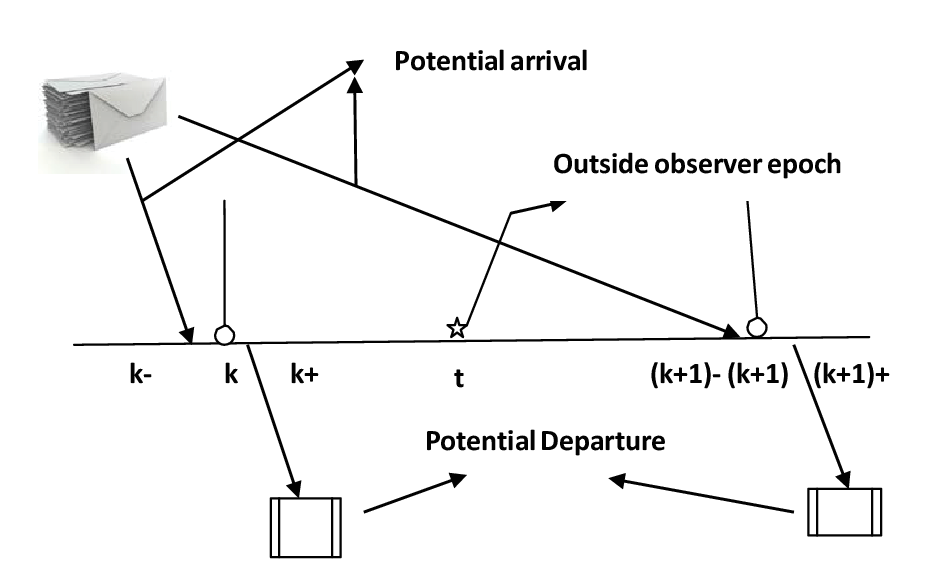}
		\end{center}
		\caption{Different epochs of LAS-DA system}\lb{figg1}
	\end{figure}
	\item \emph{\textbf{Arrival process:}} In communication network traffic modeling, it is essential to choose an arrival mechanism that is easy to analyze and free from fractal behavior such as self-similarity. One such suitable choice is the compound Bernoulli arrival process, widely used in statistical multiplexing of packet switched networks. On account of this, we consider a Bernoulli batch arrival process in which a bulk arrival occurs with
	rate $\lambda$, (no arrival with probability $\bar{\lambda}=1-\lambda$) which leads to geometric inter-arrival times. The probability mass function (pmf) of the group size distribution is characterized by the random variable $X$ and given by $\varphi_{n}=\lambda\bar{\lambda}^{n-1},~0\le \lambda\le 1,~n\ge1$. The size of arriving group of customers during a time slot consists of identically distributed random variables with $Pr(X=k)=g_{k}$, where $k={1,2,3..\dots}$ with corresponding probability generating function (pgf) $G(z)=\sum_{k=1}^{\infty}g_{k}z^k$ and hence, mean arrival group size $\bar{g}=\sum_{k=1}^{\infty}kg_{k}$. However, in most of the practical scenarios, group size is finite. Keeping this in mind, during numerical illustration, we consider the group size to be finite.
	
	\item \emph{\textbf{Service process and rule:}} In this model, a single server offers FES, which is compulsory for all customers, followed by an SOS that may be availed by only a subset of customers. Both FES and SOS are rendered in batches, and both service times are assumed to follow general distributions that depend on the batch size, thereby allowing the model to capture a wide range of practical scenarios.\\
	(i) \emph{\textbf{Service policy at FES:}} 
	In the FES stage, the single server operates under the general bulk-service $(a,b)$ rule. Specifically, when the queue length is fewer than `$a$',  the server remains idle, and service begins only when at least `$a$' customers are available, with a maximum of `$b$' customers served simultaneously. A batch of `$r$' customers, where $(a\leq r \leq b)$ is served according to a general service-time distribution with pmf $S_r(n),~n=1,2,\cdots$,  corresponding $z$-transform $S^*_r(z)=\sum_{\ell=1}^{\infty}S_{r}(\ell)z^{\ell}$ and the mean service time $\frac{1}{\mu_r}=\bar S_r=S_r^{*(1)}(1)$, $a\leq r \leq b$, where $S_r^{*(1)}(1)$ denotes the derivative of $S^*_r(z)$ evaluated at $z=1$.
	~\\
	(ii) \emph{\textbf{Service policy at SOS:}} After the completion of FES, a subset or all of the customers from the served batch may choose to undertake the SOS, which is also provided by the same server. From a batch of size $r$ $(a \le r \le b)$ that has just completed FES, a group of $y$ customers $(1 \le y \le r)$ enter the SOS stage, where the selection is assumed to follow a binomial distribution with the pmf given by:
	\begin{eqnarray}
		\chi_{r,y}=\binom{r}{y}p^y(1-p)^{r-y},~0\leq y \leq r, ~a\leq r \leq b \label{eq1}
	\end{eqnarray}
	where $p$ $(0<p<1)$ denotes the probability that an individual customer opts for the SOS, and $\sum_{y=0}^{r}\chi_{r,y}=1$. The service time during SOS is assumed to follow a general distribution with pmf $S_y^{O}(n)$, $n=1,2,\ldots$, corresponding $z$-transform $S_y^{*O}(z)=\sum_{\ell=1}^{\infty}S_{y}^{O}(\ell)z^{\ell}$, and mean service duration $\frac{1}{\mu_y^{O}} = \bar{S}_y^{O} = S_y^{*O(1)}(1)$ for $1 \le y \le b$, where $S_y^{*O(1)}(1)$ denotes the derivative of $S_y^{*O}(z)$ evaluated at $z=1$.
	
	\item \emph{\textbf{Vacation policy:}} Several vacation policies, including single vacation, multiple vacations, differentiated vacation, working vacation, $N$-limited vacation schemes etc., have been extensively investigated by various authors. Among these, single and multiple sleep/vacation strategies are particularly popular in improving energy efficiency in 5G cellular networks. In the performance evaluation of computer and manufacturing systems, server vacations deliver as an effective tool to enhance operational flexibility. Incorporating single or multiple vacations allows the queueing model to better reflect practical situations.
	Motivated by this, the present study analyzes the system under two vacation policies, namely single and multiple vacations. It is important to note that the server operates exclusively under one of these policies throughout the service process, and switching between them is not allowed once the policy is chosen. For convenience, the steady state equations and corresponding performance measures for both vacation policies are presented in a unified manner, and each can be separately derived using the indicator function defined as:
	\begin{align*}
		\delta_{p} =\left\{\begin{array}{r@{\mskip\thickmuskip}l}
			0, ~~& \text{if the system is working with single vacation,} \\
			1, ~~& \text{if the system is working with multiple vacations}.
		\end{array}\right.
	\end{align*}
	All the derived equations and performance measures can be obtained by setting $\delta_p=0$ for the single vacation case and $\delta_p=1$ for the multiple vacation case.\\
	(i) \emph{\textbf{Single vacation:}} After completing the service of a batch, the server takes a vacation if the number of customers in the queue is below the minimum service threshold $`a$'. At the end of the random vacation period, the server returns and resumes service by serving all $`r$', $(a \leq r \leq b)$ customers waiting in the queue (if any). If the queue still contains fewer than $`a$' customers, the server remains idle until at least $`a$' customers are accumulated.\\
	(ii) \emph{\textbf{Multiple vacations:}} After completing a batch service, the server examines the queue. If the number of waiting customers is fewer than the minimum required threshold $`a$', it immediately enters a vacation state; otherwise, it continues the service based on the $(a,b)$ policy. At the end of each vacation, if the queue size is still below $`a$', the server initiates another vacation and repeats this process until at least $`a$' customers are accumulated.\\
	We consider $V^{[k]},~(0 \le k\le a-1)$, as the random variable represents the vacation time when the server is on $k$'th type vacation. The pmf of $k$'th type vacation, which is of random length, is defined as $v^{[k]}(n)$, where $0 \le k \le a-1,~n=1, 2, 3,\dots$ and the corresponding $z$-transformation is given by $V^{[k]*}(z) = \sum_{n=1}^{\infty} v^{[k]}(n)z^{n},$ and the mean is expressed as $E(V^{[k]}) = \bar{V}^{[k]}=\frac{1}{\eta_{k}}=V^{[k]*(1)}(1)$.
	\end{itemize}
	\section{Applications}\lb{PA}
	\subsection{Application of the model in Wireless Networks / IoT / Cloud computing/ Manufacturing}
	This model presents a flexible analytical framework for a variety of modern engineering systems in which traffic arrives in bursts, service complexity varies with load, and idle periods must be managed efficiently. In wireless communication networks, the batch arrival structure naturally reflects packet bursts generated by users, multimedia services, and uplink traffic in 5G/6G systems, while the two-phase service mechanism corresponds to essential packet handling FES followed by optional tasks SOS such as enhanced error correction, adaptive modulation adjustments, encryption updates, or deep packet inspection that are selectively applied when the queue length lies within a moderate congestion region $(a,b)$. The incorporated vacation discipline mirrors energy-saving features such as Discontinuous Reception (DRX) in mobile devices and sleep modes in base stations, enabling reduced power consumption during idle periods. In IoT networks, where devices exhibit sporadic transmission, limited battery capacity, and variable processing needs, the model captures the interplay between mandatory sensing/communication tasks and optional analytic or validation steps, while vacations reflect low-power sleep states crucial for extending device lifetime. Similarly, in cloud and edge computing environments, the model accurately represents multi-stage service pipelines in which essential processing is performed for all tasks but computationally expensive functions such as compression, indexing, anomaly detection, or data sanitization are triggered only under moderate workload to maintain latency guarantees and avoid resource bottlenecks. During idle periods, virtual machines or edge servers may be suspended or migrated, effectively modeled through queue length dependent vacation periods. In manufacturing and industrial automation, the FES corresponds to primary processing or assembly, such as optional quality inspection, calibration, or finishing tasks are executed conditionally based on buffer occupancy, ensuring efficient throughput while maintaining product quality. Machine downtime, maintenance, or deliberate energy-saving shutdowns are naturally represented through the vacation component of the model. Overall, this queueing structure provides a unified representation of diverse systems that require adaptive service control, congestion-aware processing, and energy-efficient operational strategies.
	\subsection{Other areas of possible application of the model}
	\begin{itemize}
		\item Medical Facilities:	The proposed model is equally applicable to medical and healthcare facilities, where patient flow is highly variable and services often consist of multiple stages. The FES phase reflects essential diagnostic or treatment procedures that must be provided to all patients, while the SOS phase captures optional yet beneficial activities such as counseling, follow-up checking, secondary testing, or electronic record updates, executed only when patient load is manageable. Queue-length dependent vacations policy specially applicable when healthcare staff attend emergencies, shift between departments, or engage in administrative tasks, thereby influencing service availability.
		\item Banking and financial service sector: It captures the behaviour through its batch arrival structure and congestion-dependent service phases. The FES corresponds to the mandatory customer interactions such as cash withdrawal, deposits, cheque processing, and account inquiries, which must be completed for every individual. The SOS represents additional operations such as KYC verification, investment advisory, loan clarification, or document validation which are performed conditionally based on customers choice. This ensures that optional yet time-consuming tasks are completed only when customer load is moderate, thus reducing waiting time during busy periods. The queue-length dependent vacation policy appropriately models counter closures during staff breaks, low-traffic hours, or temporary unavailability due to administrative duties. This enables banks to optimize staffing levels, enhance service efficiency, and maintain acceptable customer waiting times even under varying traffic conditions.
	\end{itemize}
	\section{State of the system and difference equations}\lb{SE}
	To achieve the analysis of the model, we begin this section by formulating the governing equations in the steady state. Consequently, it is essential to define several important states of the system at time $t-$ as:
	\begin{itemize}
		
		\item $X_q(t-)$ $\equiv$ The measure of packets or customers accumulated in the queue awaiting service,
		
		\item $X_s(t-)$ $\equiv$ The quantity of packets or customers being served at the FES  by the server,
		
		\item $\widetilde{X}_s(t-)$ $\equiv$ The quantity of packets or customers presently receiving service from the server in SOS,
		
		\item $J_1(t-)$ $\equiv$ The pending service slot in FES, if any,
		
		\item $J_2(t-)$ $\equiv$ The pending service slot in SOS, if any,
		
		\item $\xi^{[k]}(t-)$ $\equiv$ The remaining vacation time of the server, excluding the current slot, corresponding to the $k^{\text{th}}$ vacation,
		
		\item $\eta(t-)$ $\equiv$ The server's state is specified as:
		\begin{align*}
			\hspace{1.4cm}\eta(t-) = \left\{\begin{array}{r@{\mskip\thickmuskip}l}
				&0\hspace{0.8cm}  \mbox{for ~dormant~state,}\\
				&1\hspace{0.8cm}  \mbox{for~ vacation~state,}\\
				& 2\hspace{0.8cm} \mbox{for~busy~state.}\\
			\end{array}\right.
		\end{align*}	
	\end{itemize}
	The joint probabilities are defined as:
	\begin{eqnarray}
		\vartheta_{n,0}(t-)&=&\mbox{Pr}\{ X_q(t-)=n,~ X_s(t-)=0, ~\eta(t-)=0  \}, \quad 0\leq n \leq a-1\nn\\
		\alpha_{n,r}(\ell,t-)&=&\mbox{Pr}\{ X_q(t-)=n, ~X_s(t-)=r,~ J_1(t-)=\ell,~ \eta(t-)=2 \}, \quad \ell\geq 1, \nn \\
		&&~~n\geq 0, ~~ a\leq r\leq b.\nn\\
		\beta_{n,j}(\ell,t-)&=&\mbox{Pr}\{ X_q(t-)=n, ~\widetilde{X}_s(t-)=j,~ J_2(t-)=\ell,~ \eta(t-)=2 \}, \quad \ell\geq 1, \nn \\
		&&~~n\geq 0 , ~~ 1\leq j\leq b.\nn\\
		\gamma^{[k]}_{n}(\ell,t-)&=&\mbox{Pr}\{ X_q(t-)=n, ~\xi^{[k]}(t-)=\ell,~ \eta(t-)=1 \}, \quad \ell\geq 1,~n\geq 0, ~0\le k \le a-1.\nn
	\end{eqnarray}
	Hence the limiting probabilities are:
	\begin{eqnarray}
		\vartheta_{n,0}&=&\lim_{t-\rightarrow \infty}\vartheta_{n,0}(t-), ~~ o\le n \le a-1.\nn\\
		\alpha_{n,r}(\ell)&=&\lim_{t-\rightarrow \infty}\alpha_{n,r}(\ell,t-), ~~ \ell\geq 1,~n\geq 0,~a\leq r\leq b.\nn \\
		\beta_{n,j}(\ell)&=&\lim_{t-\rightarrow \infty}\beta_{n,j}(\ell,t-), ~~ \ell\geq 1,~n\geq 0, ~1\le j \le b.\nn \\
		\gamma^{[k]}_{n}(\ell)&=&\lim_{t-\rightarrow \infty}\gamma^{[k]}_{n}(\ell,t-), ~~ \ell\geq 1,~n\geq 0,~0\le k\le a-1.\nn
	\end{eqnarray}
	The steady state difference differential equations listed below have been obtained through the use of SVT, which connects the system's states at the times $t-$ and $(t+1)-$:
	\begin{eqnarray}
		\vartheta_{0,0}&=&(1-\delta_{p})\bigg(\bar{\lambda}\vartheta_{0,0}+\bar{\lambda}\gamma_{0}^{[0]}(1)\bigg) ~~~\label{ch1.eq2}\\
		\vartheta_{n,0}&=&(1-\delta_{p})\bigg(\bar{\lambda}\vartheta_{n,0}+{\lambda}\sum_{i=1}^{n}g_{i}\vartheta_{n-i,0}+\bar{\lambda}\sum_{k=0}^{n}\gamma_{n}^{[k]}(1)+\lambda\sum_{i=1}^{n}\sum_{k=0}^{n-i}g_{i}\gamma_{n-i}^{[k]}(1)\bigg), \nonumber \\
		&&~ 1\le n \le a-1 \label{ch1.eq3} \\
		\alpha_{0,r}(u)&=&\bar{\lambda}\alpha_{0,r}(u+1)+\bigg[\bar{\lambda}\sum_{m=a}^{b}\alpha_{r,m}(1)\chi_{m,0}+\lambda\sum_{i=1}^{r}\sum_{m=a}^{b}g_{i}\alpha_{r-i,m}(1)\chi_{m,0}
		\nonumber \\
		&&~+\bar{\lambda}\sum_{j=1}^{b}\beta_{r,j}(1) +\lambda\sum_{i=1}^{r}\sum_{j=1}^{b}g_{i}\beta_{r-i,j}(1)+\bar{\lambda}\sum_{k=0}^{a-1}\gamma_{r}^{[k]}(1)+\lambda\sum_{i=1}^{r}\sum_{k=0}^{a-1}g_{i}\gamma_{r-i}^{[k]}(1) \nonumber \\
		&&~+(1-\delta_{p})\lambda\sum_{i=0}^{a-1}g_{r-i}\vartheta_{i,0}\bigg]S_{r}(u), \ \ \ \ a\le r \le b ~~~\label{ch1.eq4} \\
		\alpha_{n,r}(u)&=&\bar{\lambda}\alpha_{n,r}(u+1)+{\lambda}\sum_{i=1}^{n}g_{i}\alpha_{n-i,r}(u+1), \ \ \ \ n\ge 1, \ \ a\le r \le b-1 ~~~\label{ch1.eq5}\\ 
		\alpha_{n,b}(u)&=&\bar{\lambda}\alpha_{n,b}(u+1)+{\lambda}\sum_{i=1}^{n}g_{i}\alpha_{n-i,b}(u+1)+\bigg[\bar{\lambda}\sum_{m=a}^{b}\alpha_{n+b,m}(1)\chi_{m,0}+\bar{\lambda}\sum_{j=1}^{b}\beta_{n+b,j}(1) \nonumber \\
		&&~+\lambda\sum_{i=1}^{n+b}\sum_{m=a}^{b}g_{i}\alpha_{n+b-i,m}(1)\chi_{m,0}+\lambda\sum_{i=1}^{n+b}\sum_{j=1}^{b}g_{i}\beta_{n+b-i,j}(1)+\bar{\lambda}\sum_{k=0}^{a-1}\gamma_{n+b}^{[k]}(1) \nonumber \\
		&&~+\lambda\sum_{i=1}^{n+b}\sum_{k=0}^{a-1}g_{i}\gamma_{n+b-i}^{[k]}(1)+(1-\delta_{p})\lambda\sum_{i=0}^{a-1}g_{n+b-i}\vartheta_{i,0}\bigg]S_{b}(u), \ \ \ n\ge 1 ~~~\label{ch1.eq6} \\
		\beta_{0,j}(u)&=&\bar{\lambda}\beta_{0,j}(u+1)+\bar{\lambda}\sum_{i=max(a,j)}^{b}\alpha_{0,i}(1)\chi_{i,j}S_{j}^{O}(u),\ \ \ \ 1\le j \le b
		~~~\label{ch1.eq7} \\
		\beta_{n,j}(u)&=&\bar{\lambda}\beta_{n,j}(u+1)+{\lambda}\sum_{i=1}^{n}g_{i}\beta_{n-i,j}(u+1)+\bigg[\bar{\lambda}\sum_{i=max(a,j)}^{b}\alpha_{n,i}(1)\chi_{i,j}\nonumber \\
		&&~+\lambda\sum_{m=1}^{n}\sum_{i=max(a,j)}^{b}g_{m}\alpha_{n-m,i}(1)\chi_{i,j}\bigg]S_{j}^{O}(u),\ \ \ \ n\ge 1, \ \ 1\le j \le b ~~~\label{ch1.eq8}\\
		\gamma_{0}^{[0]}(u)&=&\bar{\lambda}\gamma_{0}^{[0]}(u+1)+\bigg[\bar{\lambda}\sum_{m=a}^{b}\alpha_{0,m}(1)\chi_{m,0}+\bar{\lambda}\sum_{j=1}^{b}\beta_{0,j}(1)+\delta_{p}\bar{\lambda}\gamma^{[0]}_{0}(1)\bigg]v^{[0]}(u) ~~~\label{ch1.eq9}\\
		\gamma_{k}^{[k]}(u)&=&\bar{\lambda}\gamma_{k}^{[k]}(u+1)+\bigg[\bar{\lambda}\sum_{m=a}^{b}\alpha_{k,m}(1)\chi_{m,0}+\lambda\sum_{i=1}^{k}\sum_{m=a}^{b}g_{i}\alpha_{k-i,m}(1)\chi_{m,0}+\bar{\lambda}\sum_{j=1}^{b}\beta_{k,j}(1)\nn\\
		&&+\lambda\sum_{i=1}^{k}\sum_{j=1}^{b}g_{i}\beta_{k-i,j}(1)+\delta_{p}\big(\bar{\lambda}\gamma^{[k]}_{n}(1)+\lambda\sum_{i=1}^{k}g_{i}\gamma^{[k]}_{n-i}(1)\big)\bigg]v^{[k]}(u),\nn\\
		&& ~~~~ 0 \le k \le a-1 ~~~\label{ch1.eq10} \\
		\gamma_{n}^{[k]}(u)&=&\bar{\lambda}\gamma_{n}^{[k]}(u+1)+\lambda\sum_{i=1}^{n}g_{i}\gamma_{n-i}^{[k]}(u+1), \ \ \ \ n \ge k+1, \ \ 0\le k \le a-1 ~~~\label{ch1.eq11}
	\end{eqnarray}
	In order to achieve our intended goals of obtaining the joint distribution from the equations mentioned above, we establish the following $z$-transforms.
	\begin{eqnarray*}
		\alpha_{n,r}^{*}(z)&=&\sum_{\ell=1}^{\infty} \alpha_{n,r}(\ell)z^{\ell}, \ \ n \ge 0, \ \ \  a\le r \le b \\
		\beta_{n,j}^{*}(z)&=&\sum_{\ell=1}^{\infty} \beta_{n,j}(\ell)z^{\ell}, \ \ n \ge 0, \ \ \ 1 \le j \le b \\
		\gamma_{n}^{[k]*}(z)&=&\sum_{\ell=1}^{\infty}\gamma_{n}^{[k]}(\ell)z^{\ell}, ~~ n \ge 0,\  \ \ 0 \le k \le a-1  
	\end{eqnarray*}
	Accordingly, we derive the following:
	\begin{eqnarray*}
		\alpha_{n,r} \equiv \alpha_{n,r}^{*}(1)&=&\sum_{\ell=1}^{\infty} \alpha_{n,r}(\ell), \ \ n \ge 0, \ \ \  a\le r \le b \\
		\beta_{n,j} \equiv \beta_{n,j}^{*}(1)&=&\sum_{\ell=1}^{\infty} \beta_{n,j}(\ell), \ \ n \ge 0, \ \ \ 1 \le j \le b \\
		\gamma_{n}^{[k]} \equiv \gamma_{n}^{[k]*}(1)&=&\sum_{\ell=1}^{\infty} \gamma_{n}^{[k]}(\ell), \ \ n \ge 0, \  \ \ 0 \le k \le a-1 
	\end{eqnarray*}
	Consequently, we multiply equations (\ref{ch1.eq4}) - (\ref{ch1.eq11}) by $z^{\ell}$ and summing over $\ell$ from $1$ to $\infty$, we get
	\begin{eqnarray}
		\left(\frac{z-\bar{\lambda}}{z}\right)\alpha_{0,r}^{*}(z)&=&\bigg[\bar{\lambda}\sum_{m=a}^{b}\alpha_{r,m}(1)\chi_{m,0}+\bar{\lambda}\sum_{j=1}^{b}\beta_{r,j}(1)+\lambda\sum_{i=1}^{r}\sum_{m=a}^{b}g_{i}\alpha_{r-i,m}(1)\chi_{m,0}
		\nonumber \\
		&&~+\lambda\sum_{i=1}^{r}\sum_{j=1}^{b}g_{i}\beta_{r-i,j}(1)+\bar{\lambda}\sum_{k=0}^{a-1}\gamma_{r}^{[k]}(1)+\lambda\sum_{i=1}^{r}\sum_{k=0}^{a-1}g_{i}\gamma^{[k]}_{r-i}(1) \nonumber \\
		&&~+(1-\delta_{p})\lambda\sum_{i=0}^{a-1}g_{r-i}\vartheta_{i,0}\bigg]S_{r}^{*}(z)-\bar{\lambda}\alpha_{0,r}(1),\ \ a \le r \le b ~~~\label{ch1.eq12} \\
		\left(\frac{z-\bar{\lambda}}{z}\right)\alpha_{n,r}^{*}(z)&=&\frac{\lambda}{z}\sum_{i=1}^{n}g_{i}\alpha_{n-i,r}^{*}(z)-\lambda\sum_{i=1}^{n}g_{i} \alpha_{n-i,r}(1)-\bar{\lambda}\alpha_{n,r}(1), \nn \\
		&&~ n\ge 1, ~~ a \le r \le b-1 ~~~\label{ch1.eq13} 
	\end{eqnarray}
	\begin{eqnarray}
		\left(\frac{z-\bar{\lambda}}{z}\right)\alpha_{n,b}^{*}(z)&=&\frac{\lambda}{z}\sum_{i=1}^{n}g_{i}\alpha_{n-i,b}^{*}(z)+\bigg[\bar{\lambda}\sum_{m=a}^{b}\alpha_{n+b,m}(1)\chi_{m,0}
		+\bar{\lambda}\sum_{j=1}^{b}\beta_{n+b,j}(1) \nonumber \\
		&&~+\lambda\sum_{i=1}^{n+b}\sum_{m=a}^{b}g_{i}\alpha_{n+b-i,m}(1)\chi_{m,0}+\lambda\sum_{i=1}^{n+b}\sum_{j=1}^{b}g_{i}\beta_{n+b-i,j}(1) \nonumber \\	                            
		&&~+\bar{\lambda}\sum_{k=0}^{a-1}\gamma_{n+b}^{[k]}(1)+\lambda\sum_{i=1}^{n+b}\sum_{k=0}^{a-1}g_{i}\gamma_{n+b-i}^{[k]}(1)+(1-\delta_{p})\lambda\sum_{i=1}^{a-1}g_{n+b-i}\vartheta_{i,0}\bigg] \nonumber \\
		&&~S_{b}^{*}(z)-\bar{\lambda}\alpha_{n,b}(1)-\lambda\sum_{i=1}^{n}g_{i}\alpha_{n-i,b}(1), \  n \ge 1 ~~~\label{ch1.eq14} \\
		\left(\frac{z-\bar{\lambda}}{z}\right)\beta_{0,j}^{*}(z)&=&\bar{\lambda}\sum_{m=max(a,j)}^{b}\alpha_{0,m}(1)\chi_{m,j}S_{j}^{*O}(z) -\bar{\lambda}\beta_{0,j}(1), ~~ 1\le j \le b ~~~\label{ch1.eq15} \\
		\left(\frac{z-\bar{\lambda}}{z}\right)\beta_{n,j}^{*}(z)&=&\frac{\lambda}{z}\sum_{i=1}^{n}g_{i}\beta_{n-1,j}^{*}(z)+\bigg[\lambda\sum_{m=1}^{n}\sum_{i=max(a,j)}^{b}g_{m}\alpha_{n-m,i}(1)\chi_{i,j}
		\nonumber \\
		&&~+\bar{\lambda}\sum_{i=max(a,j)}^{b}\alpha_{n,i}(1)\chi_{i,j}\bigg]S_{j}^{*O}(z)-\bar{\lambda}\beta_{n,j}(1)-\lambda \sum_{i=1}^{n} g_{i}\beta_{n-i,j}(1), \nn \\
		&&~ n\ge 1, ~ 1\le j \le b \label{ch1.eq16} \\
		\left(\frac{z-\bar{\lambda}}{z}\right)\gamma_{0}^{[0]*}(z)&=&\bigg[\bar{\lambda}\sum_{m=a}^{b}\alpha_{0,m}(1)\chi_{m,0}+\bar{\lambda}\sum_{j=1}^{b}\beta_{0,j}(1)+\delta_{p}\bar{\lambda}\gamma^{[0]}_{0}(1)\bigg]V^{[0]*}(z)\nn \\
		&&~-\bar{\lambda}\gamma^{[0]}_{0}(1) ~~~\label{ch1.eq17} \\
		\left(\frac{z-\bar{\lambda}}{z}\right)\gamma_{k}^{[k]*}(z)&=&\bigg[\bar{\lambda}\sum_{m=a}^{b}\alpha_{k,m}(1)\chi_{m,0}
		+\lambda\sum_{i=1}^{k}\sum_{m=a}^{b}g_{i}\alpha_{k-i,m}(1)\chi_{m,0}+\bar{\lambda}\sum_{j=1}^{b}\beta_{n,j}(1) \nonumber \\ &&~+\lambda\sum_{i=1}^{k}\sum_{j=1}^{b}g_{i}\beta_{k-i,j}(1)+\delta_{p}\bar{\lambda}\sum_{m=0}^{k}\gamma_{n}^{[m]}(1)+\delta_{p}\lambda\sum_{i=1}^{k}\sum_{m=0}^{k}g_{i}\gamma_{n-i}^{[k]}(1)\bigg]\nn\\
		&&V^{[k]*}(z)-\bar{\lambda}\gamma^{[k]}_{k}(1),
		~1\le k \le a-1\label{ch1.eq18} \\
		\left(\frac{z-\bar{\lambda}}{z}\right)\gamma_{n}^{[k]*}(z)&=&\frac{\lambda}{z}\sum_{i=1}^{n}g_{i}\gamma_{n-i}^{[k]*}(z)-\lambda\sum_{i=1}^{n}g_{i}\gamma_{n-i}^{[k]}(1)-\bar{\lambda}\gamma_{n}^{[k]}(1), \nn \\
		&&~ 0\le k \le a-1, ~~~ n \ge k+1. \label{ch1.eq19}
	\end{eqnarray}
	\section{Stationary joint distribution at service/vacation completion epoch}\lb{SC}
	Deriving the joint distribution at the completion epochs of FES and SOS requires the use of embedded points defined at these instants. Thus, the corresponding joint probabilities are expressed as follows:
	\begin{itemize}
		\item $\alpha_{n,r}^+$ = $Pr\{ X_q(t+)=n,~X_s(t+)=r ~\mbox{and FES has just finished}\}$,\quad $n\ge 0$,  $a \le r \le b$.
		\item $\alpha_{n}^+ \left(=\sum_{r=a}^{b}\alpha_{n,r}^+\right)$= $Pr\{ X_q(t+)=n ~\mbox{and FES has just ended}\}$,\quad $n\ge 0$.
		\item  $\beta_{n,j}^+ $= $Pr\{ X_q(t+)=n,~\widetilde{X}_s(t+)=j~ \mbox{and SOS has just finished}\}$,\quad $n \ge 0, \ 1\le j\le b.$
		\item  $\beta_{n}^+ \left(=\sum_{j=1}^{b}\beta_{n,j}^+\right)$= $Pr\{ X_q(t+)=n~ \mbox{and SOS has just ended}\}$,\quad $n\ge 0$.
		\item  $\gamma_{n}^{[k]+} \left(=\sum_{k=0}^{min(n,a-1)}\gamma_{n}^{[k]+}\right) ~\mbox{be the probability that $n$ customers in the queue}\\
		\mbox{during vacation completion epoch of the server}\},\quad n\ge 0, ~ 0\le k \le a-1$.
	\end{itemize}
	Based on equations (\ref{ch1.eq2}), (\ref{ch1.eq3}), and (\ref{ch1.eq12})–(\ref{ch1.eq19}), we proceed to establish two results which are presented as lemmas and will serve an important role in the continuation of our study. \\
	\textbf{Lemma 5.1}
	The probabilities $\left( \alpha_{n,r}^+, \alpha_{n,r}(1)\right),\left( \beta_{n,j}^+, \beta_{n,j}(1)\right) $ and  $(\gamma_{n}^{[k]+},\gamma_{n}^{[k]}(1)) $ are related by the following expression
	\begin{eqnarray}
		\alpha_{0,r}^+ &=& \tau^{-1} \{\bar{\lambda}\alpha_{0,r} (1)\},  \ \ \ \ a \le r \le b ~~~\label{ch1.eq20} \\
		\alpha_{n,r}^+ &=& \tau^{-1} \{\bar{\lambda}\alpha_{n,r} (1)+ \lambda \sum_{i=1}^{n}g_{i}\alpha_{n-i,r} (1) \},  \ \ \ \ n \ge 1, \ \ \ \ a \le r \le b ~~~\label{ch1.eq21}\\
		\beta_{0,j}^+ &=& \tau^{-1} \{\bar{\lambda}\beta_{0,j} (1)\},  \ \ \ \ 1 \le j \le b ~~~\label{ch1.eq22} \\
		\beta_{n,j}^+ &=& \tau^{-1} \{\bar{\lambda}\beta_{n,j} (1)+ \lambda \sum_{i=1}^{n}g_{i}\beta_{n-i,j} (1) \},  \ \ \ \ n \ge 1, \ \ \ \ 1 \le j \le b ~~~\label{ch1.eq23} \\
		\gamma_{k}^{[k]+} &=& \tau^{-1} \{\bar{\lambda}\gamma_{k} ^{[k]}(1)\}, \ \ \ \ 0\le k \le a-1 ~~~\label{ch1.eq24}\\
		\gamma_{n}^{[k]+} &=& \tau^{-1} \{\bar{\lambda}\gamma_{n}^{[k]} (1)+ \lambda \sum_{i=1}^{n}g_{i} \gamma_{n-i}^{[k]} (1) \},  \ \ \ \ n\ge k+1, \ \ \ \ 0\le k\le a-1 ~~~\label{ch1.eq25}
	\end{eqnarray}
	where $\tau=\sum_{n=0}^{\infty}\sum_{r=a}^{b} \alpha_{n,r}(1)+\sum_{n=0}^{\infty}\sum_{j=1}^{b}\beta_{n,j}(1)+\sum_{n=0}^{\infty} \sum_{k=0}^{min(n,a-1)}\gamma_{n}^{[k]}(1).$ \\
	
	\noindent   	\textbf{Lemma 5.2}
	The value of $\tau$ is given by
	\begin{eqnarray}
		\tau= \sum_{n=0}^{\infty}\sum_{r=a}^{b} \alpha_{n,r}(1)+\sum_{n=0}^{\infty}\sum_{j=1}^{b}\beta_{n,j}(1)+\sum_{n=0}^{\infty} \sum_{k=0}^{min(n,a-1)}\gamma_{n}^{[k]}(1)=\frac{1-(1-\delta_{p})\sum_{n=0}^{a-1}\vartheta_{n,0}}{\Lambda} \nonumber \\
		\label{ch1.eq26}
	\end{eqnarray}
	where 
	\begin{eqnarray}
		\Lambda&=&\sum_{n=0}^{a-1}\bigg[(1-\delta_{p})\sum_{k=0}^{n}\gamma_n^{[k]+}\bigg(\sum_{j=n}^{a-1}e_{j,n}\big(\sum_{\ell=a}^{b}g_{\ell-j}\bar{S_\ell}+\sum_{\ell=b+1}^{\infty}g_{\ell-j}\bar{S_b}\big)\bigg)+\sum_{\ell=a}^{b}\alpha_{n,\ell}^{+}\chi_{\ell,0}\bar{V}^{[n]} \nn \\
		&&~+\sum_{j=1}^{b}\beta_{n,j}^{+}\bar{V}^{[n]}+\sum_{\ell=a}^{b}\sum_{m=1}^{\ell}\alpha_{n,\ell}^{+}\chi_{\ell,m}\bar S_{m}^O+\delta_{p}\sum_{k=0}^{a-1}\gamma_{n}^{[k]+}\bar{V}^{[n]}\bigg] \nn\\
		&&~+\sum_{n=a}^{b}\bigg[\sum_{k=0}^{a-1}\gamma_n^{[k]+}\bar S_{n}+\sum_{r=a}^{b}\bigg(\chi_{r,0}\bar S_{n}+\sum_{m=1}^{r}\chi_{r,m}\bar S_{m}^O\bigg)\alpha_{n,r}^{+}
		+\sum_{j=1}^{b}\beta_{n,j}^{+}\bar S_{n}\bigg]  \nn\\
		&&~+\sum_{n=b+1}^{\infty}\bigg[\sum_{k=0}^{a-1}\gamma_n^{[k]+}\bar S_{b}+\sum_{r=a}^{b}\bigg(\chi_{r,0}\bar S_{b}+\sum_{m=1}^{r}\chi_{r,m}\bar S_{m}^O\bigg)\alpha_{n,r}^{+}
		+\sum_{j=1}^{b}\beta_{n,j}^{+}\bar S_{b}\bigg] ~~~\label{ch1.eq27}
	\end{eqnarray}
	and
	\begin{eqnarray}
		e_{n,i}=\sum_{j=i+1}^{n-1}e_{n,j}g_{j-i}+g_{n-i},~i=0,1,\dots,n-2 ~\mbox{with}~ e_{n,n-1}=g_1 ~\mbox{and}~ e_{n,n}=1.~~~\label{ch1.eq28}
	\end{eqnarray}
	\begin{proof}
		For single vacation, substituting $\delta_{p}=0$ in (\ref{ch1.eq2}) and (\ref{ch1.eq3}), we obtain
		\begin{eqnarray}
			\lambda \vartheta_{n,0}=\bar{\lambda}e_{n,0}\gamma_{0}^{[0]}(1)+\sum_{m=1}^{n}\sum_{k=0}^{m}e_{n,m}\bigg(\bar{\lambda}\gamma_{m}^{[k]}(1)+\lambda\sum_{j=1}^{m} g_{n}\gamma_{m-j}^{[k]}(1)\bigg),~~~0\leq n \leq a-1.\label{ch1.eq29}
		\end{eqnarray}
		Using (\ref{ch1.eq29}) in (\ref{ch1.eq12}) and then summing over $r$ from $a$ to $b$ and $n$ from $0$ to $\infty$ in the equations (\ref{ch1.eq12}) - (\ref{ch1.eq19}), after some simplifications, we get
		\begin{small}
			\begin{eqnarray*}
				&&~\bigg(\frac{z-1}{z}\bigg)\bigg\{\sum_{n=0}^{\infty}\sum_{r=a}^{b}\alpha_{n,r}^{*}(z)+\sum_{n=0}^{\infty}\sum_{j=1}^{b}\beta_{n,j}^{*}(z)+\sum_{k=0}^{a-1}\sum_{n=k}^{\infty}\gamma_{n}^{[k]*}(z) \bigg\} \\
				&=&\sum_{n=0}^{a-1}\bigg[\biggl\{\bar{\lambda}e_{n,0}\gamma_{0}^{[0]}(1)+\sum_{m=1}^{n}\sum_{k=0}^{m}e_{n,m}\bigg(\bar{\lambda}\gamma_{m}^{[k]}(1)+\lambda\sum_{j=1}^{m} g_{n}\gamma_{m-j}^{[k]}(1)\bigg)\biggr\} \{(1-\delta_{p})\bigg(\sum_{r=a}^{b}g_{r-n}S^{*}_r(z) \\
				&&~+\sum_{r=b+1}^{\infty}g_{r-n}S^{*}_b(z)\bigg)+\delta_{p}V^{[n]*}(z)\}+\bigg(\sum_{r=a}^{b}\bar{\lambda}\alpha_{0,r}(1)\chi_{r,0}+\sum_{j=1}^{b}\bar{\lambda}\beta_{0,r}(1)+\sum_{n=1}^{a-1}(\bar{\lambda}\alpha_{n,r}(1)\chi_{r,0}\\
				&&~+\lambda \alpha_{n-1,r}(1)\chi_{r,0}) +\sum_{j=1}^{b}(\bar{\lambda}\beta_{n,j}(1)+\sum_{i=1}^{n}\lambda \beta_{n-i,j}(1))\bigg)V^{[n]*}(z)\bigg]\\
				&&~+\sum_{n=a}^{b}\bigg[\sum_{r=a}^{b}\bigg(\bar{\lambda}\alpha_{n,r}(1)\chi_{r,0}+\lambda\sum_{i=1}^{n}g_{i} \alpha_{n-i,r}(1)\chi_{r,0}\bigg)+\sum_{j=1}^{b}\bigg(\bar{\lambda}\beta_{n,j}(1)+\lambda\sum_{i=1}^{n}g_{i} \beta_{n-i,j}(1)\bigg) \\
				&&~+\sum_{k=0}^{a-1}\bigg(\bar{\lambda}\gamma_{n}^{[k]}(1)+\sum_{i=1}^{n}\lambda g_{i} \gamma_{n-i}^{[k]}(1)\bigg)\bigg]S_{n}^{*}(z)+\sum_{n=b+1}^{\infty}\bigg[\sum_{r=a}^{b}\bigg(\bar{\lambda}\alpha_{n,r}(1)+\sum_{i=1}^{n}\lambda g_{i} \alpha_{n-i,r}(1)\bigg)\chi_{r,0} \\
				&&~+\sum_{j=1}^{b}(\bar{\lambda}\beta_{n,j}(1)+\sum_{i=1}^{n}\lambda g_{i} \beta_{n-i,j}(1))+\sum_{k=0}^{a-1}(\bar{\lambda}\gamma_{n}^{[k]}(1)+\sum_{i=1}^{n}\lambda g_{i}\gamma_{n-i}^{[k]}(1))\bigg]S_{b}^{*}(z)\\
				&&~+\sum_{n=0}^{\infty}\sum_{m=1}^{b}\sum_{r=\mbox{max}(a,m)}^{b}\bigg(\bar{\lambda}\alpha_{n,r}(1)+\sum_{i=1}^{n}\lambda g_{i} \alpha_{n-i,r}(1)\bigg)\chi_{r,m}S_{m}^{*O}(z) \\
				&&~-\sum_{n=0}^{\infty}\sum_{r=a}^{b}\alpha_{n,r}^{+}-\sum_{n=0}^{\infty}\sum_{m=1}^{b}\beta_{n,m}^{+}-\sum_{k=0}^{a-1}\sum_{n=k}^{\infty}\gamma_{n}^{[k]+}
			\end{eqnarray*}
		\end{small}
		Now taking limit $z\rightarrow1$ in both sides of the above expression, using L'Hospital's rule and the normalizing condition $(1-\delta_{p})\sum_{n=0}^{a-1}\vartheta_{n,0}+\sum_{n=0}^{\infty}\sum_{r=a}^{b}\alpha_{n,r}^{+}+\sum_{n=0}^{\infty}\sum_{j=1}^{b}\beta_{n,j}^{+}+\sum_{n=0}^{\infty} \sum_{k=0}^{min(n,a-1)}\gamma_{n}^{[k]+}=1$, we obtain the desired result.
	\end{proof}
	\subsection{Bivariate pgfs}
	In this subsection, we derive the bivariate pgfs for the joint distribution observed at the completion of each service batch. We also endeavor to present these distributions in a clear and accessible format. We begin by defining the following pgfs:
	\begin{eqnarray}
		\Pi(z,z_1,z_2)&=&\sum_{n=0}^{\infty}\sum_{r=a}^{b}\alpha_{n,r}^{*}(z){z_1}^{n}{z_2}^{r}, \ \ |z_1|\le 1,  \ |z_2|\le 1  ~~~\label{ch1.eq30}\\
		\Pi^{+}(z_1,z_2)&=&\sum_{n=0}^{\infty}\sum_{r=a}^{b}\alpha^{+}_{n,r}{z_1}^{n}{z_2}^{r}, \ \ |z_1|\le 1,  \ |z_2|\le 1  ~~~\label{ch1.eq31}\\
		\Pi^{+}(z_1,1)&=&\sum_{n=0}^{\infty}\sum_{r=a}^{b}\alpha^{+}_{n,r}{z_1}^{n}=\sum_{n=0}^{\infty}\alpha^{+}_{n}{z_1}^{n}=\Pi^{+}(z_1), \ \ |z_1|\le 1 ~~~\label{ch1.eq32}\\
		\Delta^{+}(z_1,z_2)&=&\sum_{n=0}^{\infty}\sum_{y=1}^{b}\beta^{+}_{n,y}{z_1}^{n}{z_2}^{y}, \ \ |z_1|\le 1,  \ |z_2|\le 1  ~~~\label{ch1.eq33}\\
		\Delta^{+}(z_1,1)&=&\sum_{n=0}^{\infty}\sum_{y=1}^{b}\beta^{+}_{n,y}{z_1}^{n}=\sum_{n=0}^{\infty}\beta^{+}_{n}{z_1}^{n} ~~~\label{ch1.eq34}\\
		\Upsilon^{+}(z_1)&=&\sum_{k=0}^{a-1}\sum_{n=0}^{\infty}\gamma^{[k]+}_{n}{z_1}^{n} \ \ |z_1|\le 1.~~~\label{ch1.eq35}
	\end{eqnarray}
	As our primary aim is to derive the bivariate pgf, we multiply (\ref{ch1.eq12})–(\ref{ch1.eq14}) by the appropriate powers of 
	$z_{1}$ and  $z_{2}$, perform the summation over $n \ge 0$ and 
	$a\le r \le b,$ and make use of (\ref{ch1.eq30}) and (\ref{ch1.eq31}). Consequently, we obtain: 
	\begin{small}
		\begin{eqnarray}
			&&\bigg\{\frac{z-\bar{\lambda}-\lambda G(z_1)}{z}\bigg\}\Pi(z,z_1,z_2)\nn\\
			&=&\sum_{n=a}^{b} \bigg[\sum_{m=a}^{b}({\bar{\lambda}\alpha_{n,m}(1)+\lambda\sum_{i=1}^{n}g_{i}\alpha_{n-i,m}(1)})\chi_{m,0}+\sum_{j=1}^{b}(\bar{\lambda}\beta_{n,j}(1)+\lambda\sum_{i=1}^{n}g_{i}\beta_{n-i,j}(1)) \nonumber \\
			&&+\sum_{k=0}^{a-1}({\bar{\lambda}\gamma_{n}^{[k]}(1)+\lambda\sum_{i=1}^{n}g_{i}\gamma_{n-i}^{[k]}(1)})\bigg]S^{*}_n(z){z_2}^n-\sum_{n=1}^{\infty}\sum_{m=a}^{b}\bigg({\bar{\lambda}\alpha_{n,m}(1)+\lambda\sum_{i=1}^{n}g_{i}\alpha_{n-i,m}(1)}\bigg){z_1}^n{z_2}^m \nonumber \\
			&&-\bar{\lambda}\sum_{m=a}^{b}\alpha_{0,m}(1){z_2}^m+(1-\delta_{p})\sum_{n=0}^{a-1}\bigg(\sum_{r=a}^{b}g_{r-n}S^{*}_r(z)z_2^{r}+\sum_{r=b+1}^{\infty}g_{r-n}S^{*}_b(z){z_2}^{b}{z_1}^{r-b}\bigg)  \nonumber \\
			&&\bigg[\bar{\lambda}e_{i,0}\gamma_{0}^{[0]}(1)+\sum_{m=1}^{n}\sum_{k=0}^{m}e_{n,m}\bigg(\bar{\lambda}\gamma_{m}^{[k]}(1)+\lambda \sum_{j=1}^{m} g_{j}\gamma_{m-j}^{[k]}(1)\bigg)\bigg]+\sum_{n=b+1}^{\infty}\bigg[\sum_{k=0}^{a-1}\bigg({\bar{\lambda}\gamma_{n}^{[k]}(1)+\lambda\sum_{i=1}^{n}g_{i}\gamma_{n-i}^{[k]}(1)}\bigg)\nonumber \\
			&&+\sum_{m=a}^{b}\bigg({\bar{\lambda}\alpha_{n,m}(1)+\lambda\sum_{i=1}^{n}g_i\alpha_{n-i,m}(1)}\bigg)\chi_{m,0}+\sum_{j=1}^{b}\bigg(\bar{\lambda}\beta_{n,j}(1)+\lambda\sum_{i=1}^{n}g_{i}\beta_{n-i,j}(1)\bigg)\bigg]{z_1}^{n-b}{z_2}^{b}S^{*}_b(z) ~~~\label{ch1.eq36}
		\end{eqnarray}
			\end{small}
	Now substituting $z=\bar{\lambda}+\lambda G(z_1)$ and using (\ref{ch1.eq20}) - (\ref{ch1.eq25}), and (\ref{ch1.eq31}), we obtain
	\begin{small}
		\begin{eqnarray}	 \Pi^{+}(z_1,z_2)&=&(1-\delta_{p})\sum_{n=0}^{a-1}\sum_{k=0}^{n}\gamma_{n}^{[k]+}\bigg[\sum_{j=n}^{a-1}e_{j,n}\bigg(\sum_{i=a}^{b}g_{i-j}K^{i}(z_1)z_2^i+\sum_{i=b+1-j}^{\infty}g_iK^{b}(z_1){z_1}^{i+j-b}z_2^b\bigg)\bigg]  \nonumber \\
			&&+\sum_{n=a}^{b}\bigg[\sum_{r=a}^{b}\alpha_{n,r}^{+}\chi_{r,0}+\sum_{j=1}^{b}\beta_{n,j}^{+}+\sum_{k=0}^{a-1}\gamma_{n}^{[k]+}\bigg]{z_2}^n K^{n}(z_1) \nonumber \\
			&&+\sum_{n=b+1}^{\infty}\bigg[\sum_{r=a}^{b}\alpha_{n,r}^{+}\chi_{r,0}+\sum_{j=1}^{b}\beta_{n,j}^{+}+\sum_{k=0}^{a-1}\gamma_{n}^{[k]+}\bigg]{z_2}^b{z_1}^{n-b} K^{b}(z_1) ~~~\label{ch1.eq37}
		\end{eqnarray}
	\end{small}
	where $K^{m}(z_1)=\sum_{i=0}^{\infty}k_{i}^{m}z_1^{i}=S_{m}^{*}(\bar{\lambda}+\lambda G(z_1)), ~~ a \le m \le b$ with $k_{i}^{m}=Pr\{i$ customers arrive during the FES time while serving $m$ customers$\}$.\\
	Now, multiplying (\ref{ch1.eq15}) - (\ref{ch1.eq16}) by appropriate powers of ${z_1}$ and ${z_2}$, summing over $n$ from $0$ to $\infty$ and $j$ from $1$ to $b$, we get
	\begin{eqnarray}
		&&\bigg\{\frac{z-\bar{\lambda}-\lambda G(z_1)}{z}\bigg\}\sum_{n=0}^{\infty}\sum_{j=1}^{b}\beta_{n,j}^{*}(z){z_1}^{n}{z_2}^{j}\nn\\
		&&=\sum_{j=1}^{b}\sum_{m=max{(a,j)}}^{b}\bigg({\bar{\lambda}\alpha_{0,m}(1)+\sum_{n=1}^{\infty}\bigg(\bar{\lambda}\alpha_{n,m}(1)+\lambda\sum_{i=1}^{n}g_{i}\alpha_{n-i,m}(1)\bigg){z_1}^{n}}\bigg){z_2}^{j}\chi_{m,j}S^{*O}_j(z)\nn \\
		&&-\sum_{j=1}^{b}\bigg({\bar{\lambda}\beta_{0,j}(1)+\sum_{n=1}^{\infty}\bigg(\bar{\lambda}\beta_{n,j}(1)+\lambda\sum_{i=1}^{n}g_{i}\beta_{n-i,j}(1)\bigg)}{z_1}^{n}{z_2}^{r}\bigg) \label{ch1.eq38}
	\end{eqnarray}
	Now substituting $z=\bar{\lambda}+\lambda G{(z_1)}$ and using (\ref{ch1.eq20}) - (\ref{ch1.eq25}), and (\ref{ch1.eq33}), we obtain
	\begin{eqnarray}
		\Delta^{+}(z_1,z_2)=\sum_{n=0}^{\infty}\sum_{m=a}^{b}\sum_{j=1}^{m}\alpha_{n,m}^{+}{z_1}^{n}{z_2}^{j}\chi_{m,j}T^{j}(z_1) \label{ch1.eq39}
	\end{eqnarray}
	where $T^{y}(z_1)=\sum_{i=0}^{\infty}t^{y}_{i}z_1^{i}=S_{y}^{*O}(\bar{\lambda}+\lambda G(z_1)), ~~ 1 \le y \le b$ with $t_{i}^{y}=Pr\{i$ customers arrive over the duration of SOS time while serving $y$ customers$\}$.\\
	Now, multiplying (\ref{ch1.eq17}) - (\ref{ch1.eq19}) by appropriate powers of ${z_1}$ and ${z_2}$, summing over $n$ from $k$ to $\infty$ and $k$ from $1$ to $a-1$, we get\\
	\begin{small}
		\begin{eqnarray}
			&&\bigg\{\frac{z-\bar{\lambda}-\lambda G(z_1)}{z}\bigg\}\sum_{k=0}^{a-1}\sum_{n=k}^{\infty}\gamma_{n}^{[k]*}(z){z_1}^{n}{z_2}^{k} \nonumber \\ &&=\sum_{k=1}^{a-1}\sum_{n=k}^{k}\bigg({\sum_{m=a}^{b}\big(\bar{\lambda}\alpha_{n,m}(1)
				+\lambda\sum_{i=1}^{n}g_{i}\alpha_{n-i,m}(1)\big)\chi_{m,0}}+\sum_{j=1}^{b}{\big(\bar{\lambda}\beta_{n,j}(1)
				+\lambda\sum_{i=1}^{n}g_{i}\beta_{n-i,j}(1)\big)}\bigg) \nonumber \\
			&&{z_1}^{n}{z_2}^{k}V^{[n]*}(z)+\bar{\lambda}\bigg(\sum_{m=a}^{b}\alpha_{0,m}(1)\chi_{m,0}+\sum_{j=1}^{b}\beta_{0,j}(1)\bigg)V^{[0]*}(z) \nonumber \\
			&&-\sum_{k=0}^{a-1}\bigg(\bar{\lambda}\gamma_{k}^{[k]}(1)+\sum_{n=k+1}^{\infty}\big(
			\bar{\lambda}\gamma_{n}^{[k]}(1)+\lambda\sum_{i=1}^{n}g_{i}\gamma_{n-i}^{[k]}(1)\big)\bigg){z_1}^{n}{z_2}^{k} \label{ch1.eq40}
		\end{eqnarray}
	\end{small}
	Now substituting $z=\bar{\lambda}+\lambda G{(z_1)}$ and using (\ref{ch1.eq20}) - (\ref{ch1.eq25}), we obtain
	\begin{eqnarray}
		\Upsilon^{+}(z_1,z_2)&=&\sum_{n=0}^{a-1}\bigg(\sum_{\ell=a}^{b}\alpha_{n,\ell}^{+}\chi_{\ell,0}+\sum_{j=1}^{b}\beta_{n,j}^{+}+\delta_{p}\sum_{k=0}^{n}\gamma_n^{[k]+}\bigg)z_1^{n}z_2^{n}H_n(z_1) \label{ch1.eq41}
	\end{eqnarray}
	where $H_{n}(z_1)=\sum_{i=0}^{\infty}h_{i}^{(k)}z^{i}_{1}=V^{[k]*}(\bar{\lambda}+\lambda G(z_1))$, $0\le k \le a-1$ with $h_{i}^{(k)}=Pr\{i$ customers come during the $k$'th type vacation of a server$\}$.\\
	Now substituting $z_2=1$ in (\ref{ch1.eq37}), we obtain the pgf of only queue content distribution at service completion, which is given by
	\begin{small}
		\begin{eqnarray}
			\Pi^{+}(z_1)&=&\frac{\begin{aligned}
					\sum_{n=0}^{a-1}\bigg[(1-\delta_{p})\sum_{k=0}^{n}\gamma_n^{[k]+}\bigg(\sum_{j=n}^{a-1}e_{j,n}\bigg(\sum_{i=a}^{b}g_{i-j}K^{i}{(z_1)}z_1^b+\sum_{i=b+1-j}^{\infty}g_{i}z_{1}^{i+j}K^{b}(z_1)\\
					-\sum_{i=a}^{b-1}g_{i-j}K^{i}(z_1)(\chi_{b,0}+F_{b-a+1}(z_1,1))K^b(z_1)\bigg)\bigg)+\bigg(\sum_{\ell=a}^{b}\alpha_{n,\ell}^{+}\chi_{\ell,0}+\sum_{j=1}^{b}\beta_{n,j}^{+}+\delta_{p}\sum_{k=0}^{a-1}\gamma_{n}^{[k]+}\bigg)z_1^{n}\\
					H_{n}(z_1)K^{b}(z_1)-\bigg(\sum_{\ell=a}^{b}\alpha_{n,\ell}^{+}\chi_{\ell,0}+\sum_{j=1}^{b}\beta_{n,j}^{+}+\sum_{k=0}^{a-1}\gamma_n^{[k]+}\bigg)K^{b}(z_1)z_1^{n}\bigg]
					\\+K^{b}(z_1)\sum_{n=a}^{b-1}\bigg[\bigg(\sum_{\ell=a}^{b}\alpha_{n,\ell}^{+}\chi_{\ell,0}+\sum_{j=1}^{b}\beta_{n,j}^{+}+\sum_{k=0}^{a-1}\gamma_n^{[k]+}\bigg)K^n(z_1)\bigg(z_1^{b}-\big(\chi_{b,0}+F_{b-a+1}(z_1,1)\big)\bigg)\\
					+K^{n}(z_1)\bigg(\sum_{\ell=a}^{b}\alpha_{n,\ell}^{+}\chi_{\ell,0}+\sum_{j=1}^{b}\beta_{n,j}^{+}+\sum_{k=0}^{a-1}\gamma_n^{[k]+}\bigg)\bigg(\chi_{n,0}+F_{n-a+1}(z_1,1)\bigg)\\
					+(1-\delta_{p})\sum_{\ell=0}^{a-1}\sum_{k=0}^{\ell}\gamma_{\ell}^{[k]+}\bigg(\sum_{j=\ell}^{a-1}e_{j,l}g_{n-j}\bigg)-\bigg(\sum_{\ell=a}^{b}\alpha_{n,\ell}^{+}\chi_{\ell,0}+\sum_{j=1}^{b}\beta_{n,j}^{+}+\sum_{k=0}^{a-1}\gamma_n^{[k]+}\bigg)z_1^{n}\bigg]
			\end{aligned}}{z_1^{b}-\big(\chi_{b,0}+F_{b-a+1}(z_1,1)\big)K^{b}(z_1)} \nonumber\\ \label{ch1.eq42}
		\end{eqnarray}
	\end{small}
	where
	\begin{eqnarray*}
		F_i(z_1,1)&=&\sum_{y=1}^{a+i-1}\chi_{a+i-1,y}~T^{y}(z_1); \ \ \  1 \le i \le b-a+1
	\end{eqnarray*}
	After using (\ref{ch1.eq41}) and (\ref{ch1.eq42}) in (\ref{ch1.eq37}) and (\ref{ch1.eq39}), and after some algebraic manipulation, we obtain the bivariate pgfs of queue and server size distribution at FES and SOS, which are exhibited below:
	\begin{small}
		\begin{eqnarray}
			\Pi^{+}(z_1,z_2)&=&\frac{\begin{aligned}
					\sum_{n=0}^{a-1}\bigg[(1-\delta_{p})\sum_{k=0}^{n}\gamma_n^{[k]+}\bigg(\sum_{j=n}^{a-1}e_{j,n}\bigg(\sum_{i=a}^{b}g_{i-j}K^{i}{(z_1)}z_2^{i}z_1^{b}+\sum_{i=b+1-j}^{\infty}g_{i}z_{1}^{i+j}K^{b}(z_1)z_2^{b}\\
					-\sum_{i=a}^{b-1}g_{i-j}K^{i}(z_1)z_2^{i}(\chi_{b,0}+F_{b-a+1}(z_1,1))K^b(z_1)\bigg)\bigg)+\bigg(\sum_{\ell=a}^{b}\alpha_{n,\ell}^{+}\chi_{\ell,0}+\sum_{j=1}^{b}\beta_{n,j}^{+}+\delta_{p}\sum_{k=0}^{a-1}\gamma_{n}^{[k]+}\bigg)\\
					z_1^{n}H_{n}(z_1)K^{b}(z_1)z_{2}^b-\bigg(\sum_{\ell=a}^{b}\alpha_{n,\ell}^{+}\chi_{\ell,0}+\sum_{j=1}^{b}\beta_{n,j}^{+}+\sum_{k=0}^{a-1}\gamma_n^{[k]+}\bigg)z_1^{n}K^{b}(z_1)z_{2}^b\bigg]
					\\+K^{b}(z_1)z_{2}^b\sum_{n=a}^{b-1}\bigg[\bigg(\sum_{\ell=a}^{b}\alpha_{n,\ell}^{+}\chi_{\ell,0}+\sum_{j=1}^{b}\beta_{n,j}^{+}+\sum_{k=0}^{a-1}\gamma_n^{[k]+}\bigg)z_2^{n}K^n(z_1)\bigg(z_1^{b}-\big(\chi_{b,0}+F_{b-a+1}(z_1,1)\big)\bigg)\\
					+K^{n}(z_1)\biggl\{\bigg(\sum_{\ell=a}^{b}\alpha_{n,\ell}^{+}\chi_{\ell,0}+\sum_{j=1}^{b}\beta_{n,j}^{+}+\sum_{k=0}^{a-1}\gamma_n^{[k]+}\bigg)\bigg(\chi_{n,0}+F_{n-a+1}(z_1,1)\bigg)\\
					+(1-\delta_{p})\sum_{\ell=0}^{a-1}\sum_{k=0}^{\ell}\gamma_{\ell}^{[k]+}\bigg(\sum_{j=\ell}^{a-1}e_{j,l}g_{n-j}\bigg)-\bigg(\sum_{\ell=a}^{b}\alpha_{n,\ell}^{+}\chi_{\ell,0}+\sum_{j=1}^{b}\beta_{n,j}^{+}+\sum_{k=0}^{a-1}\gamma_n^{[k]+}\bigg)z_1^{n}\biggr\}\bigg]
			\end{aligned}}{z_1^{b}-\big(\chi_{b,0}+F_{b-a+1}(z_1,1)\big)K^{b}(z_1)} \nonumber\\ \label{ch1.eq43}
		\end{eqnarray}
	\end{small}
	where
	\begin{eqnarray*}
		F_i(z_1,z_2)&=&\sum_{y=1}^{a+i-1}\chi_{a+i-1,y}~z_{2}^{y}~T^{y}(z_1); \ \ \  1 \le i \le b-a+1 
	\end{eqnarray*}
	\begin{small}
		\begin{eqnarray}
			\Delta^{+}(z_1,z_2)&=&\frac{\begin{aligned}
					\sum_{n=0}^{a-1}F_{b-a+1}(z_1,z_2)K^{b}(z_1)\bigg[(1-\delta_{p})\sum_{k=0}^{n}\gamma_n^{[k]+}\bigg(\sum_{j=n}^{a-1}e_{j,n}\bigg(g_{b-j}z_{1}^{b}+\sum_{i=b+1-j}^{\infty}g_{i}z_{1}^{i+j}\bigg)\bigg)\\
					+\bigg(\sum_{\ell=a}^{b}\alpha_{n,\ell}^{+}\chi_{\ell,0}+\sum_{j=1}^{b}\beta_{n,j}^{+}+\delta_{p}\sum_{k=0}^{a-1}\gamma_{n}^{[k]+}\bigg)z_1^{n}
					H_{n}(z_1)-\bigg(\sum_{\ell=a}^{b}\alpha_{n,\ell}^{+}\chi_{\ell,0}+\sum_{j=1}^{b}\beta_{n,j}^{+}+\sum_{k=0}^{a-1}\gamma_n^{[k]+}\bigg)z_1^{n}\bigg]
					\\+\sum_{n=a}^{b-1}\bigg[\bigg(\sum_{\ell=a}^{b}\alpha_{n,\ell}^{+}\chi_{\ell,0}+\sum_{j=1}^{b}\beta_{n,j}^{+}+\delta_{p}\sum_{k=0}^{a-1}\gamma_n^{[k]+}+(1-\delta_{p})\sum_{\ell=0}^{a-1}\sum_{k=0}^{\ell}\gamma_{\ell}^{[k]+}\sum_{j=\ell}^{a-1}e_{j,l}g_{n-j}\bigg)\\
					K^{n}(z_1)\bigg(F_{n-a+1}(z_1,z_2)\big(z_1^{b}-\big(\chi_{b,0}+F_{b-a+1}(z_1,1)\big)K^{b}(z_1)\big)\\+\big(\chi_{n,0}+F_{n-a+1}(z_1,1)\big)F_{b-a+1}(z_1,z_2)K^{b}(z_1)\bigg)\\
					-F_{b-a+1}(z_1,z_2)K^{b}(z_1)\bigg(\sum_{\ell=a}^{b}\alpha_{n,\ell}^{+}\chi_{\ell,0}+\sum_{j=1}^{b}\beta_{n,j}^{+}+\sum_{k=0}^{a-1}\gamma_n^{[k]+}\bigg)z_1^{n}\bigg]
			\end{aligned}}{z_1^{b}-\big(\chi_{b,0}+F_{b-a+1}(z_1,1)\big)K^{b}(z_1)} \nn \\
			\label{ch1.eq44}
		\end{eqnarray}
	\end{small}
	\noindent Equations (\ref{ch1.eq43}) and (\ref{ch1.eq44}) provide the bivariate pgfs for the joint queue and server size distributions under FES and SOS, respectively. These constitute the principal contributions of this work, and to the best of the authors’ knowledge, this introduce new results to the queueing literature. The joint probabilities can be derived from these pgfs using partial-fraction decomposition or by expressing them in terms of the roots of the characteristic equation. The pgfs in (\ref{ch1.eq41}) and (\ref{ch1.eq42}) designate for the queue length distribution at service/vacation completion epochs, and also represent new findings.
	\subsection{Extraction of joint probabilities}
	Before continuing with the derivation of the joint probabilities, it is necessary to evaluate the unknown probabilities present in the numerator of the pgf in (\ref{ch1.eq43}). This can be accomplished using the standard methodology adopted in earlier works, such as Karan and Pradhan (2025) \cite{karan2025analysis}. Using this approach, the $2b$ number of unknowns can be reduced to $b$ unknowns through a relation linking $\alpha_n^+$ and $\sum_{k=0}^{\mbox{min}(n,a-1)}\gamma_{n}^{[k]+}$ as follows:
	\begin{eqnarray}
		\sum_{k=0}^{\mbox{min}(n,a-1)}\gamma_{n}^{[k]+}&=& \sum_{j=0}^{n} \zeta_{j}^{n-j}(\alpha^{+}_{j}+ \beta^{+}_{j}),~~0 \leq n \leq a-1, \nn\\
		\sum_{k=0}^{\mbox{min}(n,a-1)}\gamma_{n}^{[k]+} &=& \pi_{n},~~a \leq n \leq b-1,\label{ch1.eq45}
	\end{eqnarray}
	\mbox{where} $\zeta_{0}^{(j)}=\dfrac{h_{0}^{(j)}}{1-\delta_{p} h_{0}^{(j)}} ~~\mbox{and}~~ \zeta_{n}^{(j)}=\dfrac{h_{n}^{(j)}+ \delta_{p} \sum_{i=1}^{n} h_{i}^{(n+j-i)} \zeta_{n-i}^{(j)}}{1-\delta_{p} h_{0}^{(n+j)}}, ~~ 1 \leq n \leq a-1, ~ 0 \le j \le a-1 ~\\ \mbox{such that } n+j\le a-1 ~~
	\mbox{and}~~ \pi_{n}=\sum_{i=0}^{a-1}\bigg(h_{n-i}^{(j)} + \delta_{p} \sum_{j=0}^{a-1-i} \zeta_{j}^{(n-i)} h_{n-i-j}^{(i)}\bigg) (\alpha^{+}_{i}+\beta^{+}_{i}),~a \leq n \leq b-1$.\\
	Furthermore, the stability of the system may be examined using the concept introduced by Abolnikov and Dukhovny (2021) \cite{abolnikov1991markov}. According to this criterion, the corresponding Markov chain is ergodic provided that
	\begin{eqnarray*}
		\dfrac{d}{dz_1}\{z_1^{b}-\big(\chi_{b,0}+F_{b-a+1}(z_1,1)\big)K^{b}(z_1)\}<b
	\end{eqnarray*}
	which leads to
	\begin{eqnarray*}
		\rho=\dfrac{\lambda {\bar{g }} \left(\mu_b+\displaystyle\sum_{y=1}^{b}\chi_{b,y}\mu^O_y\right)}{b}<1.
	\end{eqnarray*}
	Once the service and vacation time distributions are specified, the corresponding pgfs can be expressed as rational functions. The joint distribution is then obtained by collecting the appropriate coefficients of $z_2^i,$ those for $~a\leq i \leq b$ from equation (\ref{ch1.eq43}) and $z_2^i,~1\leq i \leq b$  from equation (\ref{ch1.eq44}). These coefficients appear first below.\\
	We now extract the coefficients of $z_2^{r},$ for $\ (a\le r \le b)$ from the bivariate pgf $\Pi^{+}(z_1,z_2),$ which are listed below.\\
	Coefficient of $z_{2}^{r}, ~a\le r\le b-1$:
	\begin{small}
		\begin{equation}
			\sum_{n=0}^{\infty}\alpha_{n,r}^{+}z_1^{n}=K^{r}(z_1)\bigg(\sum_{\ell=a}^{b}\alpha_{r,\ell}^{+}\chi_{\ell,0}+\sum_{j=1}^{b}\beta_{r,j}^{+}
			+(1-\delta_{p})\sum_{n=0}^{a-1}\sum_{k=0}^{n}\gamma_{r}^{[k]+}\big(\sum_{j=n}^{a-1}e_{j,n}g_{n-j}\big)+\delta_{p}\sum_{k=0}^{a-1}\gamma_r^{[k]+}\bigg) ~~~\label{ch1.eq46}
		\end{equation}
	\end{small}
	Coefficient of $z_{2}^{b}$:
	\begin{small}
		\begin{eqnarray}
			\sum_{n=0}^{\infty}\alpha_{n,b}^{+}z_1^{n}=\frac{\begin{aligned} K^{b}(z_1)\bigg[\sum_{n=0}^{a-1}\bigg(\sum_{\ell=a}^{b}\alpha_{n,\ell}^{+}\chi_{\ell,0}+\sum_{j=1}^{b}\beta_{n,j}^{+}+\delta_{p}\sum_{k=0}^{a-1}\gamma_n^{[k]+}\bigg)z_1^{n}H_{n}(z_1)
					-\bigg(\sum_{\ell=a}^{b}\alpha_{n,\ell}^{+}\chi_{\ell,0}\\
					+\sum_{j=1}^{b}\beta_{n,j}^{+}	+\sum_{k=0}^{a-1}\gamma_n^{[k]+}\bigg)z_1^{n}+(1-\delta_{p})\sum_{k=0}^{n}\gamma_n^{[k]+}\bigg(\sum_{j=n}^{a-1}e_{j,n}\bigg(g_{b-j}z_1^{b}+\sum_{i=b+1-j}^{\infty}g_{i}z_{1}^{i+j}\bigg)\bigg)
					\\
					+\sum_{n=a}^{b-1}\biggl\{K^{n}(z_1)\bigg(\sum_{\ell=a}^{b}\alpha_{n,\ell}^{+}\chi_{\ell,0}+\sum_{j=1}^{b}\beta_{n,j}^{+}+\sum_{k=0}^{a-1}\gamma_n^{[k]+}\bigg)\bigg(\chi_{n,0}+F_{n-a+1}(z_1,1)\bigg)\\
					+(1-\delta_{p})\sum_{\ell=0}^{a-1}\sum_{k=0}^{\ell}\gamma_{\ell}^{[k]+}\bigg(\sum_{j=\ell}^{a-1}e_{j,l}g_{n-j}\bigg)-\bigg(\sum_{\ell=a}^{b}\alpha_{n,\ell}^{+}\chi_{\ell,0}+\sum_{j=1}^{b}\beta_{n,j}^{+}+\sum_{k=0}^{a-1}\gamma_n^{[k]+}\bigg)z_1^{n}\biggr\}
					\bigg]
			\end{aligned}}{z_1^{b}-\big(\chi_{b,0}+F_{b-a+1}(z_1,1)\big)K^{b}(z_1)} \label{ch1.eq47}
		\end{eqnarray}
	\end{small}
	Now, we find the coefficients of $z_2^{r}, \ (1\le r \le b)$ from the bivariate pgf $\Delta^{+}(z_1,z_2)$ and are given below.\\
	Coefficient of $z_{2}^{r}, \ 1\le r \le a $:
	\begin{small}
		\begin{eqnarray}
			\sum_{n=0}^{\infty}\beta_{n,r}^{+}z_1^{n}=\frac{\begin{aligned}
					\sum_{n=0}^{a-1}\chi_{b,r}T^{r}(z_1)K^{b}(z_1)\bigg[(1-\delta_{p})\sum_{k=0}^{n}\gamma_n^{[k]+}\bigg(\sum_{j=n}^{a-1}e_{j,n}\bigg(g_{b-j}z_{1}^{b}+\sum_{i=b+1-j}^{\infty}g_{i}z_{1}^{i+j}\bigg)\bigg)\\
					+\bigg(\sum_{\ell=a}^{b}\alpha_{n,\ell}^{+}\chi_{\ell,0}+\sum_{j=1}^{b}\beta_{n,j}^{+}+\delta_{p}\sum_{k=0}^{a-1}\gamma_{n}^{[k]+}\bigg)z_1^{n}
					H_{n}(z_1)-\bigg(\sum_{\ell=a}^{b}\alpha_{n,\ell}^{+}\chi_{\ell,0}+\sum_{j=1}^{b}\beta_{n,j}^{+}+\sum_{k=0}^{a-1}\gamma_n^{[k]+}\bigg)z_1^{n}\bigg]
					\\+\sum_{n=a}^{b-1}\bigg[\bigg(\sum_{\ell=a}^{b}\alpha_{n,\ell}^{+}\chi_{\ell,0}+\sum_{j=1}^{b}\beta_{n,j}^{+}+\delta_{p}\sum_{k=0}^{a-1}\gamma_n^{[k]+}+(1-\delta_{p})\sum_{\ell=0}^{a-1}\sum_{k=0}^{\ell}\gamma_{\ell}^{[k]+}\sum_{j=\ell}^{a-1}e_{j,l}g_{n-j}\bigg)\\
					K^{n}(z_1)\bigg(\chi_{n,r}T^{r}(z_1)\big(z_1^{b}-\big(\chi_{b,0}+F_{b-a+1}(z_1,1)\big)K^{b}(z_1)\big)\\+\big(\chi_{n,0}+F_{n-a+1}(z_1,1)\big)\chi_{b,r}T^{r}(z_1)K^{b}(z_1)\bigg)\\
					-\chi_{b,r}T^{r}(z_1)K^{b}(z_1)\bigg(\sum_{\ell=a}^{b}\alpha_{n,\ell}^{+}\chi_{\ell,0}+\sum_{j=1}^{b}\beta_{n,j}^{+}+\sum_{k=0}^{a-1}\gamma_n^{[k]+}\bigg)z_1^{n}\bigg]
			\end{aligned}}{z_1^{b}-\big(\chi_{b,0}+F_{b-a+1}(z_1,1)\big)K^{b}(z_1)} \label{ch1.eq48}
		\end{eqnarray}
	\end{small}
	Coefficient of $z_{2}^{r}, \ a+1\le r \le b-1 $:
	\begin{small}
		\begin{eqnarray}
			\sum_{n=0}^{\infty}\beta_{n,r}^{+}z_1^{n}=\frac{\begin{aligned}
					\sum_{n=0}^{a-1}\chi_{b,r}T^{r}(z_1)K^{b}(z_1)\bigg[(1-\delta_{p})\sum_{k=0}^{n}\gamma_n^{[k]+}\bigg(\sum_{j=n}^{a-1}e_{j,n}\bigg(g_{b-j}z_{1}^{b}+\sum_{i=b+1-j}^{\infty}g_{i}z_{1}^{i+j}\bigg)\bigg)\\
					+\bigg(\sum_{\ell=a}^{b}\alpha_{n,\ell}^{+}\chi_{\ell,0}+\sum_{j=1}^{b}\beta_{n,j}^{+}+\delta_{p}\sum_{k=0}^{a-1}\gamma_{n}^{[k]+}\bigg)z_1^{n}
					H_{n}(z_1)-\bigg(\sum_{\ell=a}^{b}\alpha_{n,\ell}^{+}\chi_{\ell,0}+\sum_{j=1}^{b}\beta_{n,j}^{+}+\sum_{k=0}^{a-1}\gamma_n^{[k]+}\bigg)z_1^{n}\bigg]
					\\+\sum_{n=a}^{b-1}\bigg[\bigg(\sum_{\ell=a}^{b}\alpha_{n,\ell}^{+}\chi_{\ell,0}+\sum_{j=1}^{b}\beta_{n,j}^{+}+\delta_{p}\sum_{k=0}^{a-1}\gamma_n^{[k]+}+(1-\delta_{p})\sum_{\ell=0}^{a-1}\sum_{k=0}^{\ell}\gamma_{\ell}^{[k]+}\sum_{j=\ell}^{a-1}e_{j,l}g_{n-j}\bigg)\\
					K^{n}(z_1)\bigg(\big(\chi_{n,0}+F_{n-a+1}(z_1,1)\big)\chi_{b,r}T^{r}(z_1)K^{b}(z_1)\bigg)\\
					-\chi_{b,r}T^{r}(z_1)K^{b}(z_1)\bigg(\sum_{\ell=a}^{b}\alpha_{n,\ell}^{+}\chi_{\ell,0}+\sum_{j=1}^{b}\beta_{n,j}^{+}+\sum_{k=0}^{a-1}\gamma_n^{[k]+}\bigg)z_1^{n}\bigg]\\
					+\sum_{n=a+1}^{b-1}\bigg[\bigg(\sum_{\ell=a}^{b}\alpha_{n,\ell}^{+}\chi_{\ell,0}+\sum_{j=1}^{b}\beta_{n,j}^{+}+\delta_{p}\sum_{k=0}^{a-1}\gamma_n^{[k]+}+(1-\delta_{p})\sum_{\ell=0}^{a-1}\sum_{k=0}^{\ell}\gamma_{\ell}^{[k]+}\sum_{j=\ell}^{a-1}e_{j,l}g_{n-j}\bigg) \\ K^{n}(z_1)\chi_{n,r}T^{r}(z_1)\big(z_1^{b}-\big(\chi_{b,0}+F_{b-a+1}(z_1,1)\big)K^{b}(z_1)\big)\bigg]
			\end{aligned}}{z_1^{b}-\big(\chi_{b,0}+F_{b-a+1}(z_1,1)\big)K^{b}(z_1)} \label{ch1.eq49}
		\end{eqnarray}
	\end{small}
	Coefficient of $z_{2}^{b}$:
	\begin{small}
		\begin{eqnarray}
			\sum_{n=0}^{\infty}\beta_{n,b}^{+}z_1^{n}=\frac{\begin{aligned}
					\sum_{n=0}^{a-1}\chi_{b,b}T^{b}(z_1)K^{b}(z_1)\bigg[(1-\delta_{p})\sum_{k=0}^{n}\gamma_n^{[k]+}\bigg(\sum_{j=n}^{a-1}e_{j,n}\bigg(g_{b-j}z_{1}^{b}+\sum_{i=b+1-j}^{\infty}g_{i}z_{1}^{i+j}\bigg)\bigg)\\
					+\bigg(\sum_{\ell=a}^{b}\alpha_{n,\ell}^{+}\chi_{\ell,0}+\sum_{j=1}^{b}\beta_{n,j}^{+}+\delta_{p}\sum_{k=0}^{a-1}\gamma_{n}^{[k]+}\bigg)z_1^{n}
					H_{n}(z_1)-\bigg(\sum_{\ell=a}^{b}\alpha_{n,\ell}^{+}\chi_{\ell,0}+\sum_{j=1}^{b}\beta_{n,j}^{+}+\sum_{k=0}^{a-1}\gamma_n^{[k]+}\bigg)z_1^{n}\bigg]
					\\+\sum_{n=a}^{b-1}\bigg[\bigg(\sum_{\ell=a}^{b}\alpha_{n,\ell}^{+}\chi_{\ell,0}+\sum_{j=1}^{b}\beta_{n,j}^{+}+\delta_{p}\sum_{k=0}^{a-1}\gamma_n^{[k]+}+(1-\delta_{p})\sum_{\ell=0}^{a-1}\sum_{k=0}^{\ell}\gamma_{\ell}^{[k]+}\sum_{j=\ell}^{a-1}e_{j,l}g_{n-j}\bigg)\\
					K^{n}(z_1)\bigg(\big(\chi_{n,0}+F_{n-a+1}(z_1,1)\big)\chi_{b,b}T^{b}(z_1)K^{b}(z_1)\bigg)\\
					-\chi_{b,b}T^{b}(z_1)K^{b}(z_1)\bigg(\sum_{\ell=a}^{b}\alpha_{n,\ell}^{+}\chi_{\ell,0}+\sum_{j=1}^{b}\beta_{n,j}^{+}+\sum_{k=0}^{a-1}\gamma_n^{[k]+}\bigg)z_1^{n}\bigg]
			\end{aligned}}{z_1^{b}-\big(\chi_{b,0}+F_{b-a+1}(z_1,1)\big)K^{b}(z_1)} \label{ch1.eq50}
		\end{eqnarray}
	\end{small}
	\noindent Based on the expressions above, the joint probabilities may be obtained through partial-fraction expansion by using the roots of the associated denominator and considering both simple and repeated roots, following the approach in Nandy and Pradhan (2023) \cite{nandy2023stationary}. For example, we provide the joint probabilities $\alpha^+_{n,r}$ below:
	\begin{eqnarray}
		\alpha^+_{n,r}&=& \bigg(\sum_{\ell=a}^{b}\alpha_{r,\ell}^{+}\chi_{\ell,0}+\sum_{j=1}^{b}\beta_{r,j}^{+}
		+(1-\delta_{p})\sum_{n=0}^{a-1}\sum_{k=0}^{n}\gamma_{r}^{[k]+}\big(\sum_{j=n}^{a-1}e_{j,n}g_{r-j}\big)+\delta_{p}\sum_{k=0}^{a-1}\gamma_r^{[k]_+}\bigg)\nn \\
		&&~~k^{r}_n, n\geq 0,~a\leq r \leq b-1\label{ch1.eq51}\\
		\alpha^+_{n,b}&=&  \left\{\begin{array}{r@{\mskip\thickmuskip}l}
			&\tau_{n}+\sum\limits_{k=1}^{M_1-b}\frac{c_{k}}{\xi^{n+1}_{k}},\quad 0\leq n\leq L_1-M_1,\\
			&\\
			& \sum\limits_{k=1}^{M_1-b}\frac{c_{k}}{\xi^{n+1}_{k}},\quad n> L_1-M_1~\end{array}\right. \label{ch1.eq52}
	\end{eqnarray}
	where, $\tau_i$ are the resulting coefficient after division of numerator $(\mbox{say,}~N_1(z_1))$ with degree $L_1$, by denominator $(\mbox{say,}~D_1(z_1))$ with degree $M_1$ of (\ref{ch1.eq47}), and $c_{k}=-\dfrac{N_1'(\xi_{k})}{D_1^{\prime}(\xi_{k})},\quad k=1,2,\ldots,M_1-b$, $\xi_k$ being the roots of $D_1(z_1)=0$ lying in $|z_1|>1$.\\
	Likewise, the joint probabilities $\beta^+_{n,r}$ at the completion of SOS can be obtained using equations (\ref{ch1.eq48})–(\ref{ch1.eq50}). In the same manner, the queue length probabilities alone can be extracted from equation (\ref{ch1.eq42}), and these are particularly useful for evaluating tail distributions. From an application standpoint, a key performance measure play a noteworthy role in the efficient utilization of network resources, which is often realized through the exploitation of traffic burstiness using statistical multiplexing and buffering mechanisms. In this regard, the packet loss ratio is a crucial metric, generally determined through tail probabilities, which is estimated using the smallest single root located outside the unit circle.
	\section{Particular cases and matching with existing results}\lb{PC}
	\noindent From the obtained results in Section \ref{SC}, some existing results can be achieved as special cases which are listed below.
	\begin{itemize}
		\item \textbf{$Geo^{X}/G_n^{(a,b)}/1$ queue without SOS:}\\
		In our model, if we consider $p=0$, (i.e., the probability of joining of a single customer to SOS is zero, which leads $\chi_{r,y}=0,~a\leq r \leq b, ~1\leq y \leq b$), and $g_1=1$ and $g_{i}=0, \forall i \ge 2,$ then it reduces to $Geo/G_n^{(a,b)}/1$ queue without SOS. Hence, the reduced corresponding bivariate pgf of queue and server content exactly matches with the one given by Nandy and Pradhan (2021) \cite{nandy2021joint} [P.443, eq. 33].
		
		\item \textbf{$Geo^{X}/G_n^{(a,b)}/1$ queue with SOS with/without joining of whole batch:}\\
		If we substitute the following in our recent model under study, then it reduces to $Geo^{X}/G_n^{(a,b)}/1$ queue with SOS in which the entire batch of size $r,~a\leq r \leq b$ already served in FES join SOS with probability $1-p$ and may not join SOS with probability $p$. All the corresponding results are still not available in the literature, to the best of authors' knowledge.
		\begin{align*}
			\chi_{r,y} =\left\{\begin{array}{r@{\mskip\thickmuskip}l}
				p, ~~& \text{if}~ y=0 \\
				0, ~~& \text{if} ~1 \leq y \leq r-1\\
				1-p, ~~& \text{if} ~y=r
			\end{array}\right.
		\end{align*}
		
		\item \textbf{$Geo^{X}/G_n^{(a,b)}/1$ queue without vacation:}\\
		If we consider the vacation time as zero i.e., $H_{n}(z_1)=H(z_1)=V^{[k]*}(\bar{\lambda}+\lambda G(z_1))=1$, where $n$ and $k$ lies between $0$ to $a-1$, then the joint queue and server size distribution for $Geo/G_n^{(a,b)}/1$ queue without vacation can be
		availed as a special case from the model under investigation. In this case, the system consists of both the FES and SOS. The corresponding theoretical and numerical results are still not available in the literature which can be perceived from our analysis as a special case.
		
		\item \textbf{$Geo/G/1$ queue without vacation and without SOS}\\
		In our current model under investigation, if we substitute $a=b=1$, $K^r(z_1)=K(z_1),~a\leq r \leq b$, $H_{n}(z_1)=H(z_1)=V^{[k]*}(\bar{\lambda}+\lambda G(z_1))=1,~ 0\le n,k \le a-1$, and $\chi_{r,y}=0,~a\leq r \leq b, ~1\leq y \leq b$, then our model reduces to the most basic discrete-time queueing system and the pgf exactly matches with  Alfa (2016) \cite{alfa2016applied} [P.157].
	\end{itemize}
	\section{Stationary joint distribution at arbitrary slot}\lb{RS}
	While the joint distribution at an arbitrary slot may be obtained via its bivariate pgf, the availability of joint distributions at service complete slot enables a more streamlined approach. Using these known results, we derive the arbitrary-slot probabilities, which in turn facilitate the evaluation of several marginal distributions and queueing characteristics in an elegant manner. Hence, we formulate a relationship linking the probabilities at these slots.
	\begin{thm}
		The joint probabilities at service completion and at an arbitrary slot i.e., $(\alpha_{n,r},\alpha_{n,r}^{+},\alpha_n^{+}),$ $(\beta_{n,r}, \beta_{n,r}^{+}, \beta_n^{+})$ and $(\gamma_{n}^{[k]},\gamma_{n}^{[k]+}) $are connected through the following relation: \\
		\begin{eqnarray}
			\vartheta_{0,0}&=&\frac{\gamma_{0}^{[0]+}}{E^{*}},   ~~~\label{ch1.eq53} \\
			\vartheta_{n,0}&=&\frac{\sum_{m=0}^{n}\sum_{k=0}^{m}e_{n,m}\gamma_{m}^{[k]+}}{E^{*}}, \ \ \ \ 1 \le n \le a-1  ~~~\label{ch1.eq54} \\
			\alpha_{0,r}&=&(1-\delta_{p})\sum_{n=0}^{a-1}g_{r-n}\vartheta_{n,0}+\frac{\sum_{m=a}^{b}\alpha_{r,m}^{+}\chi_{m,0}+\sum_{j=1}^{b}\beta_{r,j}^{+}+\sum_{k=0}^{a-1}\gamma_{r}^{[k]+}-\alpha_{0,r}^{+}}{E^{*}}, \nonumber \\
			&&~a \le r \le b \label{ch1.eq55}  \\
			\alpha_{n,r}&=& \sum_{i=1}^{n}g_{i}\alpha_{n-i,r}-\frac{\alpha_{n,r}^{+}}{E^{*}},\ \ \ \ n\ge 1, \  a \le r \le b-1 ~~~\label{ch1.eq56} \\
			\alpha_{n,b}&=&\sum_{i=1}^{n}g_{i}\alpha_{n-i,b}+\frac{\sum_{m=a}^{b}\alpha_{n+b,m}^{+}\chi_{m,0}+\sum_{j=1}^{b}\beta_{n+b,j}^{+}+\sum_{k=0}^{a-1}\gamma_{n+b}^{[k]+}-\alpha_{n,b}^{+}}{E^{*}} \nonumber \\
			&&~+(1-\delta_{p})\sum_{i=0}^{a-1}g_{n+b-i}\vartheta_{i,0} , \ \ \ \ n \ge 1 ~~~\label{ch1.eq57} \\
			\beta_{0,j}&=&\frac{\sum_{m=max(a,j)}^{b}\alpha_{0,m}^{+}\chi_{m,j}-\beta_{0,j}^{+}}{E^{*}}, \ \ \ \ 1 \le j \le b ~~~\label{ch1.eq58} \\
			\beta_{n,j}&=&\sum_{i=1}^{n}g_{i}\beta_{n-i,j}+\frac{\sum_{m=max(a,j)}^{b}\alpha_{n,m}^{+}\chi_{m,j}-\beta_{n,j}^{+}}{E^{*}},\ \ \ \ n\ge 1, \  1 \le j \le b ~~~\label{ch1.eq59} \\
			\gamma_{k}^{[k]}&=&\frac{\sum_{m=a}^{b}\alpha_{k,m}^{+}\chi_{m,0}+\sum_{j=1}^{b}\beta_{k,j}^{+}-\gamma_{k}^{[k]+}+\delta_{p}\sum_{m=0}^{k}\gamma_{k}^{[m]+}}{E^{*}}, \ \ \ \ 0 \le k \le a-1 ~~~\label{ch1.eq60}\\
			\gamma_{n}^{[k]}&=&\sum_{i=1}^{n}g_{i}\gamma_{n-i}^{[k]}-\frac{\gamma_n^{[k]+}}{E^{*}}, \ \ \ \ n \ge a~~~\label{ch1.eq61}
		\end{eqnarray}
		where
		\begin{eqnarray}
			E^{*}&=&\lambda\Lambda+(1-\delta_{p})\sum_{j=0}^{a-1}\sum_{m=0}^{j}\sum_{k=0}^{m}e_{j,m}\gamma_{m}^{[k]+},\label{ch1.eq62}\\
			\Lambda&=&\sum_{n=0}^{a-1}\bigg[(1-\delta_{p})\sum_{k=0}^{n}\gamma_n^{[k]+}\bigg(\sum_{j=n}^{a-1}e_{j,n}\big(\sum_{\ell=a}^{b}g_{\ell-j}\bar{S_\ell}+\sum_{\ell=b+1}^{\infty}g_{\ell-j}\bar{S_b}\big)\bigg)+\sum_{\ell=a}^{b}\alpha_{n,\ell}^{+}\chi_{\ell,0}\bar{V}^{[n]} \nn \\
			&&~+\sum_{j=1}^{b}\beta_{n,j}^{+}\bar{V}^{[n]}+\sum_{\ell=a}^{b}\sum_{m=1}^{\ell}\alpha_{n,\ell}^{+}\chi_{\ell,m}\bar S_{m}^O+\delta_{p}\sum_{k=0}^{a-1}\gamma_{n}^{[k]+}\bar{V}^{[n]}\bigg] \nn\\
			&&~+\sum_{n=a}^{b}\bigg[\sum_{k=0}^{a-1}\gamma_n^{[k]+}\bar S_{n}+\sum_{r=a}^{b}\bigg(\chi_{r,0}\bar S_{n}+\sum_{m=1}^{r}\chi_{r,m}\bar S_{m}^O\bigg)\alpha_{n,r}^{+}
			+\sum_{j=1}^{b}\beta_{n,j}^{+}\bar S_{n}\bigg]  \nn\\
			&&~+\sum_{n=b+1}^{\infty}\bigg[\sum_{k=0}^{a-1}\gamma_n^{[k]+}\bar S_{b}+\sum_{r=a}^{b}\bigg(\chi_{r,0}\bar S_{b}+\sum_{m=1}^{r}\chi_{r,m}\bar S_{m}^O\bigg)\alpha_{n,r}^{+}
			+\sum_{j=1}^{b}\beta_{n,j}^{+}\bar S_{b}\bigg] \label{ch1.eq63}\end{eqnarray}
	\end{thm}
	\begin{proof}
		Use of (\ref{ch1.eq29}) in (\ref{ch1.eq2}) - (\ref{ch1.eq3}) and then division by $\tau$ and use of (\ref{ch1.eq20}) - (\ref{ch1.eq25}) leads to
		\begin{eqnarray}
			\vartheta_{0,0}&=&\frac{1-(1-\delta_{p})\sum_{n=0}^{a-1}\vartheta_{n,0}}{\lambda\Lambda}\gamma_{0}^{[0]+}~~~\label{ch1.eq64} \\
			\vartheta_{n,0}&=&\frac{1-(1-\delta_{p})\sum_{n=0}^{a-1}\vartheta_{n,0}}{\lambda\Lambda}\sum_{m=0}^{n}\sum_{k=0}^{m}e_{n,m}\gamma_{m}^{[k]+}, \ \ \ 1 \le n \le a-1.~~~\label{ch1.eq65}
		\end{eqnarray}
		Now, dividing (\ref{ch1.eq64}) by (\ref{ch1.eq65}), we get
		\begin{eqnarray}
			\vartheta_{n,0}=\frac{\vartheta_{0,0}}{\gamma_{0}^+}\sum_{m=0}^{n}\sum_{k=0}^{m}e_{n,m}\gamma_{m}^{[k]+},\ \ \ 0 \le n \le a-1.~~~\label{ch1.eq66}
		\end{eqnarray}
		Use of (\ref{ch1.eq66}) in (\ref{ch1.eq64}) and (\ref{ch1.eq65}) yields
		\begin{eqnarray}
			\vartheta_{n,0}=\frac{\sum_{m=0}^{n}\sum_{k=0}^{m}e_{n,m}\gamma_{m}^{[k]+}}{\lambda\Lambda+(1-\delta_{p})\sum_{j=0}^{a-1}\sum_{m=0}^{j}\sum_{k=0}^{m}e_{j,m}\gamma_{m}^{[k]+}},\ \ \ 0 \le n \le a-1.~~~\label{ch1.eq67}
		\end{eqnarray}
		which is the desired result of (\ref{ch1.eq54}).\\
		Putting $z=1$ in (\ref{ch1.eq12}), use of (\ref{ch1.eq29}), and division by $\tau$ gives
		\begin{eqnarray}
			\alpha_{0,r}&=&\frac{1-(1-\delta_{p})\sum_{n=0}^{a-1}\vartheta_{n,0}}{{\lambda\Lambda}}\bigg(\sum_{m=a}^{b}\alpha_{r,m}^{+}\chi_{m,0}+\sum_{j=1}^{b}\beta_{r,j}^{+}
			+\sum_{k=0}^{a-1}\gamma_{r}^{[k]+}-\alpha_{0,r}^{+}\bigg)\nonumber \\
			&&+(1-\delta_{p})\sum_{n=0}^{a-1}g_{r-n}\vartheta_{n,0}, ~~~~~~~a \le r \le b  ~~~\label{ch1.eq68}
		\end{eqnarray}
		As
		\begin{eqnarray}
			1-(1-\delta_{p})\sum_{n=0}^{a-1}\vartheta_{n,0}=\frac{\lambda\Lambda}{\lambda\Lambda+(1-\delta_{p})\sum_{j=0}^{a-1}\sum_{m=0}^{j}\sum_{k=0}^{m}e_{j,m}\gamma_{m}^{[k]+}}, ~~~\label{ch1.eq69}
		\end{eqnarray}
		Using (\ref{ch1.eq69}) in (\ref{ch1.eq67}) and (\ref{ch1.eq68}), we get the required results of (\ref{ch1.eq55}).
		Substituting $z=1$ in (\ref{ch1.eq13}) - (\ref{ch1.eq19}) and then dividing by $\tau,$ using the similar approach, we get our desired results of (\ref{ch1.eq56}) - (\ref{ch1.eq61}).
	\end{proof}
	\section{Performance indices}\lb{PI}
	Performance metrics play an indispensable role from a practical standpoint. Using the steady-state probabilities obtained, we evaluate several significant performance indicators for the proposed model.
	\begin{itemize}
		\item[(i)] \emph{\textbf{Probability distribution of system and queue count:}} To determine the average number of customers in the system or queue, we first need to find the corresponding probability distributions. Let $~\Psi_n^{sys}$ and $\Psi_n^{queue}$
		denote the probability distribution of system and queue count, respectively, which are obtained as:\\
		\begin{align*}
			\Psi_n^{sys} = \left\{\begin{array}{r@{\mskip\thickmuskip}l}
				&\left(1-\delta_{p}\right)\vartheta_{n,0} +\displaystyle\sum_{j=1}^{\mbox{min}(n,b)}\beta_{n-j,j}+\sum_{k=0}^{\mbox{min}(n,a-1)}\gamma_{n}^{[k]} ~~,\hspace{1.0cm}0\leq n \leq a-1\\
				& \displaystyle\sum_{r=a}^{\mbox{min}(n,b)}\alpha_{n-r,r}+\displaystyle\sum_{j=1}^{\mbox{min}(n,b)}\beta_{n-j,j}+\sum_{k=0}^{\mbox{min}(n,a-1)}\gamma_{n}^{[k]}~~,\hspace{1.3cm} n\geq a,\vspace{0.1cm}
			\end{array}\right.
		\end{align*}
		
		\begin{align*}
			\Psi_n^{queue} =\left\{\begin{array}{r@{\mskip\thickmuskip}l}
				& \left(1-\delta_{p}\right)\vartheta_{n,0}+\displaystyle\sum_{r=a}^{b}\alpha_{n,r}+\displaystyle\sum_{j=1}^{b}\beta_{n,j}+\sum_{k=0}^{\mbox{min}(n,a-1)}\gamma_{n}^{[k]}~,\hspace{0.3cm}0\leq n \leq a-1\\
				&\displaystyle\sum_{r=a}^{b}\alpha_{n,r}+\displaystyle\sum_{j=1}^{b}\beta_{n,j}+\sum_{k=0}^{\mbox{min}(n,a-1)}\gamma_{n}^{[k]}~,\hspace{4.0cm} n\geq a.
			\end{array}\right.
		\end{align*}
		\item[(ii)] \emph{\textbf{Expected number of customers present with the server during the FES and SOS busy periods:}} The expected number of customers in the queue and being served in the FES and SOS busy periods are respectively obtained as
		\begin{eqnarray*}
			L_{q}=\sum_{n=0}^{\infty}n\Psi_{n}^{queue},~~~~~~	L_{S}^{FES}=\sum_{r=a}^{b}r \Psi_{r}^{ser},~~~~~~L_{S}^{SOS}=\sum_{l=1}^{b}y \eta_{l}^{ser},
		\end{eqnarray*} 
		where $\Psi_{r}^{ser}=\displaystyle\sum_{n=0}^{\infty}\alpha_{n,r},~a\leq r \leq b,~ ~\eta_{l}^{ser}=\displaystyle\sum_{n=0}^{\infty}\beta_{n,l},~1\leq l \leq b~ $.
		
		\item[(iii)] \emph{\textbf{Expected waiting time in the queue and in the system:}} The expected waiting time of the customers in the queue as well as in the system are calculated as:
		\begin{eqnarray*}
			W_q=\frac{L_q}{\lambda\bar{g}}, ~~~~~~W_s^{FES}=\frac{L_{S}^{FES}}{\lambda\bar{g}}~~~~\mbox{and}~~~W_s^{SOS}=\frac{L_{S}^{SOS}}{\lambda\bar{g}}.
		\end{eqnarray*}	
		\item[(iv)] \emph{\textbf{Energy saving factor:}}
		From an application perspective, the energy saving factor (ESF) denotes the percentage of time the server stays in inactive or in vacation periods (single or multiple), and is mathematically written as:
		\begin{align*}
			ESF = \left\{\begin{array}{r@{\mskip\thickmuskip}l}
				&(\Phi_S+\phi_D)\times100\% ~~,\hspace{0.8cm}\mbox{for single vacation}\\
				& \varphi_M \times 100\%~~,\hspace{1.8cm} \mbox{for multiple vacations},\vspace{0.1cm}\\
			\end{array}\right.
		\end{align*}
		where $\Phi_S=\sum_{k=0}^{a-1}\sum_{n=0}^{\infty}\gamma_{n}^{[k]},~\varphi_M=\sum_{k=0}^{a-1}\sum_{n=0}^{\infty}\gamma_n^{[k]},~\phi_D=\sum_{n=0}^{a-1}\vartheta_{n,0}$ represents the steady-state probabilities for single vacation, multiple vacations, and dormant state, respectively. \\
		\item[(v)] \emph{\textbf{Expected idle period:}}
			Let $\mathcal{I}_1$ and $\mathcal{I}_2$
		be the dormant periods, and $\mathcal{V}_{1}$ and $\mathcal{V}_{2}$
		the vacation periods, triggered after FES and SOS, respectively.\\
		(a) For single vacation: The expected idle period for the single vaction is calculated by E[$\mathcal{I}$]=E[$\mathcal{I}_{1}$] + 	E[$\mathcal{I}_{2}$] + E[$\mathcal{V}_{1}$] + E[$\mathcal{V}_{2}$], where E[$\mathcal{I}_{1}$], E[$\mathcal{I}_{2}$], E[$\mathcal{V}_{1}$] and E[$\mathcal{V}_{2}$]  are defined as \\
		\begin{small}
			\begin{eqnarray}
				E[\mathcal{I}_{1}] &=& \dfrac{\sum_{n=0}^{a-1} \sum_{r=a}^{b} \alpha^{+}_{n,r} \chi_{r,0}\bigg[\sum_{i=0}^{a-1-n}h_{i}^{(n)}(a-n-i)\bigg]}{\lambda \bar{g} \sum_{n=0}^{a-1}\sum_{r=a}^{b} \alpha^{+}_{n,r}} \nn\\
				E[\mathcal{I}_{2}] &=& \dfrac{\sum_{n=0}^{a-1} \sum_{r=a}^{b} \sum_{j=1}^{r} \alpha^{+}_{n,r} \chi_{r,j}\bigg[ \sum_{k=0}^{a-1-n} \sum_{i=0}^{a-1-n-k} t_{k}^{j}h_{i}^{(n)}(a-n-k-i)\bigg]}{\lambda\bar{g} \sum_{n=0}^{a-1}\sum_{r=a}^{b}\alpha^{+}_{n,r}} \nn \\
				E[\mathcal{V}_{1}] &=&\dfrac{\sum_{n=0}^{a-1}\sum_{r=a}^{b}\alpha^{+}_{n,r} \chi_{r,0}\bar{V}^{[n]}}{\sum_{n=0}^{a-1}\sum_{r=a}^{b}\alpha^{+}_{n,r}} \nn\\
				E[\mathcal{V}_{2}] &=& \dfrac{\sum_{n=0}^{a-1}\sum_{r=a}^{b} \sum_{j=1}^{r}\alpha^{+}_{n,r} \chi_{r,j}\bigg[\sum_{i=0}^{a-1-n}t^{j}_{i}\bar{V}^{[n]}\bigg]}{\sum_{n=0}^{a-1}\sum_{r=a}^{b}\alpha^{+}_{n,r}} \nn
			\end{eqnarray}
		\end{small}
		(b) For multiple vacations: The expected idle period for the multiple vacation is calculated by E[$\mathcal{I}$]= E[$\mathcal{V}_{1}$] + E[$\mathcal{V}_{2}$], where E[$\mathcal{V}_{1}$] and E[$\mathcal{V}_{2}$] are defined as 
			\begin{eqnarray}
			E[\mathcal{V}_{1}]&=& \dfrac{\dfrac{\sum_{n=0}^{a-1}\sum_{r=a}^{b}\alpha^{+}_{n,r} \chi_{r,0}\bar{V}^{[n]}}{\sum_{n=0}^{a-1}\sum_{r=a}^{b}\alpha^{+}_{n,r}}}{1 - \sum_{n=0}^{a-1} \sum_{i=0}^{n}\sum_{j=a}^{b}\alpha^{+}_{i,j} \chi_{j,0} h_{n-i}^{(n)}} \nn\eea
			\bea
			E[\mathcal{V}_{2}] &=& \dfrac{ \dfrac{\sum_{n=0}^{a-1}\sum_{r=a}^{b} \sum_{j=1}^{r}\alpha^{+}_{n,r} \chi_{r,j}\bigg[\sum_{i=0}^{a-1-n}t^{j}_{i}\bar{V}^{[n]}\bigg]}{\sum_{n=0}^{a-1}\sum_{r=a}^{b}\alpha^{+}_{n,r}}}{1 - \sum_{n=0}^{a-1} \sum_{j=a}^{b} \sum_{i=1}^{j}\alpha^{+}_{n,j} \chi_{j,i}\bigg[\sum_{k=0}^{a-1-n} \sum_{r=0}^{a-1-n-k}t^{i}_{k}h_{r}^{(n)}\bigg]} \nn
		\end{eqnarray}
		\item[(vi)] \emph{\textbf{Utility of the system:}} Using the concept of alternating renewal process, we can easily calculate the system utility for both FES and SOS and it is given by
		\begin{eqnarray}
		U[\mathcal{B}_{1}] &=& E[\mathcal{I}] \biggl(\dfrac{1- P_{idle}- P^{SOS}_{B}}{P_{idle}}\biggr) \nn\\
		U[\mathcal{B}_{2}] &=& E[\mathcal{I}] \biggl(\dfrac{1- P_{idle}- P^{FES}_{B}}{P_{idle}}\biggr) \nn
		\end{eqnarray}
		where $P_{idle}$  is calculated by $(1-\delta_p)\left[\sum_{n=0}^{a-1}\vartheta_{n,0}+\Phi_S\right]+\delta_p \varphi_M$, $\mathcal{B}_{1}$ and $\mathcal{B}_{2}$ are the time duration when the server is busy in providing the FES and SOS, respectively and $P_B^{FES}=\displaystyle\sum_{r=a}^{b}\Psi_{r}^{ser}$, $P_{B}^{SOS}=\displaystyle\sum_{l=1}^{b}\eta_{l}^{ser}$ designate the busy probability during FES and SOS, respectively. Also, U[$\mathcal{B}$]= U[$\mathcal{B}_{1}$] + U[$\mathcal{B}_{2}$] which is derived by Karan and Pradhan (2025) \cite{karan2025analysis}, recently.\\
		\item[(vii)]  \emph{\textbf{Expected cycle length:}} The length of the cycle is denoted as E[$\mathcal{C}$] and therefore E[$\mathcal{C}$] is the sum of the dormant period, vacation period and system utility for both FES and SOS in single vacation and the sum of the vacation period and system utility for both FES and SOS in multiple vacations.\\
		(a) For single vacation: E[$\mathcal{C}$]=U[$\mathcal{B}_{1}$] + U[$\mathcal{B}_{2}$] + E[$\mathcal{I}_{1}$] + E[$\mathcal{I}_{2}$] + E[$\mathcal{V}_{1}$] + E[$\mathcal{V}_{2}$].\\
		(b) For multiple vacation: E[$\mathcal{C}$]=U[$\mathcal{B}_{1}$] + U[$\mathcal{B}_{2}$] + E[$\mathcal{V}_{1}$] + E[$\mathcal{V}_{2}$].\\
		\item[(viii)] \emph{\textbf{Throughput of the system:}} The maximum throughput in an infinite-buffer queue is calculated by multiplying the effective service rate by the percentage of time the server is accessible to customers.
		\begin{eqnarray*}
			\mbox{Throughput of the system}=\left(P_B^{FES} \times \mu^{FES}\right)+\left(P_B^{SOS} \times \mu^{SOS}\right),
		\end{eqnarray*} 
		where, $\mu^{FES}=\displaystyle\sum_{r=a}^{b}r\mu_r\big/\displaystyle\sum_{r=a}^{b}r$ and $\mu^{SOS}=\displaystyle\sum_{y=1}^{b}y\mu_y^{O}\big/\displaystyle\sum_{y=1}^{b}y$, represent the average service rates during FES and SOS, respectively.
		\end{itemize}
			\section{Numerical illustrations}\lb{NI}
		The aim of this section is to demonstrate the numerical performance of the proposed model and its associated analytical results. All the computed results are displayed either in tabular or graphical form along with detailed discussions or interpretation. We illustrate the numerical example using a discrete phase type (DPH) service time distribution, owing to its ability to encompass a wide class of service distributions as special cases.
\subsection{Example with DPH service time distribution for both phases and negative binomial vacation time distribution}
In IoT cloud stations, batch service stations can be utilized to process a large set of workloads together in order to maximize the throughput. Most of the empirical studies related to wireless telecommunication systems have revealed that the general DPH distribution can be used more adequately to model the real life scenarios in a tractable  analytical frame.
Phase type distributions play a significant role in modeling complex stochastic systems, and are widely applied in areas such as insurance risk analysis for estimating ruin probabilities, as well as reliability and performance evaluation of technical systems. Over the past few decades, they have gained considerable prominence in science and engineering, primarily due to their analytical tractability and several desirable properties: (i) phase type distributions can approximate any non-negative distribution with arbitrary accuracy, making them a highly flexible and versatile class, (ii) they are closed under common operations such as addition, minimum, and maximum, facilitating exact, closed-form solutions in many analyses, (iii) in stochastic modeling, they can be explicitly represented through finite state Markov chains or processes. DPH distribution represents the distribution of the absorption time in a discrete time Markov chain consisting of $(m+1)$ states, where $m$ are transient and one is absorbing.
Keeping the above mentioned importance, in this example, we consider the service distribution to follow DPH distribution for both FES and SOS with parameters displayed in Tables [\ref{tbb1} - \ref{tbb2}]. Here, the FES time follows DPH distribution represented by $(\beta_{r}, \mathcal{T}_{r})$, where $\beta_{r}$ designates the initial probability distribution and it is represented in the form of a row vector with `$m$' components, and $\mathcal{T}_{r}$ is a square matrix of order `$m$' with all non-negative entries lies between $0$ to $1$, which one exhibited in Table [\ref{tbb1}] along with mean service time $(\bar{S_r})$. The SOS time is also considered as DPH and represented by $(\beta^{'}_y, \mathcal{T}^{'}_{y})$ and its value is given in Table [\ref{tbb2}]. Here, we consider the minimum threshold $a=3$, maximum capacity $b=8$, arrival rate $\lambda=0.50$ and the probability of an individual customer opt the SOS time, $`p$'$=0.50$. The vacation time is considered to follow a negative binomial distribution $(\mbox{i.e., NB}(\eta_k,r))$ where  $\eta_k=0.35$, $r=2$ and so that $H_{k}(z_1)=\bigg(\dfrac{2 \eta_{k} (\bar \lambda + \lambda z_{1})}{ 1 - (1 - 2 \eta_{k})(\bar\lambda + \lambda z_1)}\bigg)^2$, where $0 \le k \le a-1$. The joint distributions are shown in Tables [\ref{tb3} - \ref{tb6}].		 
\begin{table}[h!]
	\centering	
	\caption{DPH distribution for different batch-size for FES }\label{tbb1}
	\begin{tiny}
		\begin{tabular}{|c|c|c|c|c|c|c|c|}\hline
			$r$ & $\beta_r$ & $\mathcal{T}_r$ & $\bar{S}_{r}$&$r$&$\beta_r$&$\mathcal{T}_r$&$\bar{S}_{r}$\\\hline
			3 & (0.30,~0.30,~0.40)
			& $\left(\begin{array}{ccc}
				0.40 & 0.20 & 0.20 \\
				0.30 & 0.30 & 0.10 \\
				0.20 & 0.30 & 0.40\\
			\end{array}\right)$
			& 4.912162&6&(0.2, 0.5, 0.3)& $\left(\begin{array}{ccc}
				0.20 & 0.40 & 0.10 \\
				0.10 & 0.60 & 0.20 \\
				0.10 & 0.20 & 0.50 \\	\end{array}\right)$&6.138298\\
			&&&&&&&\\
			4 & (0.40, 0.20, 0.40)
			& $\left(\begin{array}{ccc}
				0.30 & 0.40 & 0.20 \\
				0.30 & 0.20 & 0.30\\
				0.20 & 0.30 & 0.20\\
			\end{array}\right)$
			& 5.083721&7&(0.25, 0.35, 0.40)&$\left(\begin{array}{ccc}
				0.40 & 0.30 & 0.20 \\
				0.20 & 0.40 & 0.20 \\
				0.20 & 0.40 & 0.30 \\
			\end{array}\right)$&7.245454\\
			&&&&&&&\\
			5 & (0.5, 0.1, 0.40)
			& $\left(\begin{array}{ccc}
				0.50 & 0.10 & 0.30 \\
				0.30 & 0.40 & 0.10 \\
				0.10 & 0.20 & 0.50\\
			\end{array}\right)$
			& 5.988636 &8&(0.40, 0.30, 0.30)&$\left(\begin{array}{ccc}
				0.30 & 0.20 & 0.40 \\
				0.30 & 0.40 & 0.20 \\
				0.30 & 0.50 & 0.10 \\
			\end{array}\right)$&10.000000\\\hline	\end{tabular}
	\end{tiny}
\end{table}	

\begin{table}[h!]
	\centering
	\caption{DPH distribution for different batch-size for SOS }\label{tbb2}
	\begin{tiny}	\begin{tabular}{|c|c|c|c|c|c|c|c|}\hline
			$y$ & $\beta^{'}_y$ & $\mathcal{T}^{'}_y$ & $\bar{S}^{O}_{y}$&$y$ & $\beta^{'}_y$ & $\mathcal{T}^{'}_y$ & $\bar{S}^{O}_{y}$\\\hline
			1 & (0.30  0.30  0.40)
			& $\left(\begin{array}{ccc}
				0.10 & 0.10 & 0.20 \\
				0.10 & 0.20 & 0.10 \\
				0.20 & 0.10 & 0.10\\
			\end{array}\right)$
			& 1.666667 &5& (0.40  0.40  0.20)
			&$\left(\begin{array}{ccc}
				0.10 & 0.20 & 0.20 \\
				0.30 & 0.30 & 0.20 \\
				0.20 & 0.20 & 0.10 \\
			\end{array}\right)$&2.615385\\
			&&&&&&&\\
			2 & (0.40 0.50  0.10)
			& $\left(\begin{array}{ccc}
				0.20 & 0.20 & 0.20 \\
				0.20 & 0.10 & 0.10\\
				0.10 & 0.20 & 0.10\\
			\end{array}\right)$
			& 1.915493 &6& (0.40 0.10  0.50)
			&$\left(\begin{array}{ccc}
				0.10 & 0.30 & 0.30 \\
				0.10 & 0.40 & 0.20 \\
				0.10 & 0.20 & 0.30 \\
			\end{array}\right)$&2.907216\\
			&&&&&&&\\
			3 & (0.30  0.40  0.30)
			& $\left(\begin{array}{ccc}
				0.10 & 0.20 & 0.30 \\
				0.20 & 0.20 & 0.10 \\
				0.20 & 0.20 & 0.10\\
			\end{array}\right)$
			& 2.132075 &7& (0.30 0.60  0.10)
			&$\left(\begin{array}{ccc}
				0.30 & 0.40 & 0.20 \\
				0.20 & 0.20 & 0.20 \\
				0.20 & 0.10 & 0.30 \\
			\end{array}\right)$&3.300000\\
			&&&&&&&\\
			4 & (0.20  0.50  0.30)
			& $\left(\begin{array}{ccc}
				0.30 & 0.20 & 0.20 \\
				0.20 & 0.30 & 0.10 \\
				0.10 & 0.20 & 0.20 \\
			\end{array}\right)$
			& 2.503106 &8& (0.40  0.20  0.40)
			&$\left(\begin{array}{ccc}
				0.10 & 0.20 & 0.30 \\
				0.40 & 0.30 & 0.10 \\
				0.20 & 0.20 & 0.50 \\
			\end{array}\right)$&4.534759\\\hline
		\end{tabular}
	\end{tiny}
\end{table}
\begin{table}[h!]
	\centering
	\begin{tiny}
		\caption{ Probability distribution at departure epoch for single  vacation}\lb{tb3}$\vspace{0.03cm}$
		\begin{tabular}{|c|c|c|c|c|c|c|c|c|c|c|} \hline
			$n$&$\alpha^{+}_{n,3}$&$\alpha^{+}_{n,4}$ & $\alpha^{+}_{n,5}$& $\alpha^{+}_{n,6}$& $\alpha^{+}_{n,7}$& $\alpha^{+}_{n,8}$& $\alpha^{+}_{n}$& $\beta^{+}_{n,1}$&$\beta^{+}_{n,2}$&$\beta^{+}_{n,3}$\\\hline
			0&0.029708&0.023698&0.008285&0.010191&0.007123&0.007407&0.086412&0.011629&0.014807&0.009920\\
			1&0.008681&0.008235&0.002856&0.003125&0.002163&0.007744&0.032804&0.005818&0.007899&0.005843\\
			2&0.021684&0.021018&0.007241&0.007858&0.005433&0.011883&0.075113&0.014906&0.020091&0.014432\\
			3&0.006908&0.008808&0.002806&0.002741&0.001868&0.010997&0.034128&0.007202&0.010480&0.008436\\
			4&0.009479&0.012664&0.003946&0.003819&0.002598&0.011474&0.043980&0.010111&0.014672&0.011573\\
			5&0.004238&0.007271&0.002027&0.001867&\textbf{0.001257}&0.010182&0.026842&0.005499&0.008539&0.007464\\
			10&0.001010&0.003583&0.000715&0.000572&0.000376&0.005294&0.011550&0.002075&0.003534&0.003483\\
			20&0.000031&0.000555&0.000053&0.000031&0.000019&0.000877&0.001566&0.000233&0.000433&0.000475\\
			50&0.000000&0.000002&0.000000&0.000000&0.000000&0.000003&0.000005&0.000000&0.000001&0.000002\\
			70&0.000000&0.000000&0.000000&0.000000&0.000000&0.000000&0.000000&0.000000&0.000000&0.000000\\
			$\geq100$&0.000000&0.000000&0.000000&0.000000&0.000000&0.000000&0.000000&0.0000000&0.000000&0.000000\\\hline
		\end{tabular}
	\end{tiny}
\end{table}	
\noindent Here, \textbf{0.001257} denotes the probability that immediately after a departure, there are $5$ customers in the system and $7$ number of customers will be served by FES, under single vacation policy at departure epoch. 
\newpage
\begin{table}[h!]
	\centering
	\begin{tiny}
		\begin{tabular}{|c|c|c|c|c|c|c|c|c|c|c|} \hline
			$n$&$\beta^{+}_{n,4}$&$\beta^{+}_{n,5}$ & $\beta^{+}_{n,6}$& $\beta^{+}_{n,7}$& $\beta^{+}_{n,8}$& $\beta^{+}_{n}$& $\gamma^{[0]_+}_{n}$&$\gamma^{[1]_+}_{n}$&$\gamma^{[2]_+}_{n}$&$\gamma^{+}_{n}$\\\hline
			0&0.004427&0.001886&0.000607&0.000119&0.000009&0.043404&0.018091&0.000000&0.000000&0.018091\\
			1&0.003105&0.001573&0.000597&0.000137&0.000013&0.024985&0.004897&0.014415&0.000000&0.019312\\
			2&0.007199&0.003399&0.001226&0.000272&0.000026&0.061551&0.012122&0.002851&0.042007&0.056980\\
			3&0.005021&0.002748&0.001119&0.000275&0.000028&0.035309&0.003346&0.006850&0.006547&0.016743\\
			4&0.006617&0.003473&0.001389&0.000339&0.000035&0.048209&0.004495&0.000935&0.015454&0.020884\\
			5&0.004924&0.002868&0.001239&0.000320&0.000035&0.030888&0.001748&0.001171&0.000838&0.003757\\
			10&0.002636&0.001663&0.000782&0.000218&0.000028&0.014419&0.000283&0.000011&0.000004&0.000298\\
			20&0.000402&0.000274&0.000139&0.000041&0.000006&0.004496&0.000003&0.000000&0.000000&0.000003\\
			50&0.000001&0.000000&0.000000&0.000000&0.000000&0.000004&0.000000&0.000000&0.000000&0.000000\\
			70&0.000000&0.000000&0.000000&0.000000&0.000000&0.000000&0.000000&0.000000&0.000000&0.000000\\
			$\geq100$&0.000000&0.000000&0.000000&0.000000&0.000000&0.000000&0.000000&0.000000&0.000000&0.000000\\\hline
		\end{tabular}	
	\end{tiny}
\end{table}
\begin{table}[h!]
	\centering
	\begin{tiny}
		\caption{ Probability distribution at arbitrary  epoch for single  vacation}\lb{tb4}$\vspace{0.03cm}$		
		\begin{tabular}{|c|c|c|c|c|c|c|c|c|c|c|} \hline
			$n$&$\vartheta_{n,0}$&$\alpha_{n,3}$ & $\alpha_{n,4}$& $\alpha_{n,5}$& $\alpha_{n,6}$& $\alpha_{n,7}$& $\alpha_{n,8}$& $\beta_{n,1}$&$\beta_{n,2}$&$\beta_{n,3}$\\\hline
			0&0.013126&0.046714&0.073691&0.019944&0.019211&0.130501&0.011381&0.006027&0.008835&0.006573\\
			1&0.017950&0.007716&0.016132&0.003911&0.003496&0.002345&0.010416&0.002240&0.003573&0.003028\\
			2&0.055916&0.019281&0.041174&0.009880&0.008795&0.005897&0.013441&0.005629&0.008866&0.007205\\
			3&&0.006174&0.017254&0.003666&0.003097&0.002056&0.011964&0.002337&0.004030&0.003758\\
			4&&0.008472&0.024809&0.005153&0.004314&0.002859&0.011955&0.003251&0.005559&0.005028\\
			5&&0.003788&0.014245&0.002641&0.002107&0.001385&0.010418&0.001731&0.003159&0.003175\\
			10&&0.000903&0.007020&0.000931&0.000646&0.000414&0.005158&0.000655&0.0001295&0.001443\\
			20&&0.000028&0.001086&0.000069&0.000035&0.000021&0.000840&0.000075&0.000161&0.000197\\
			50&&0.000000&0.000004&0.000000&0.000000&0.000000&0.000003&0.000000&0.000000&0.000000\\
			70&&0.000000&0.000000&0.000000&0.000000&0.000000&0.000006&0.000000&0.000000&0.000000\\
			$\geq100$&&0.000000&0.000000&0.000000&0.000000&0.000000&0.000000&0.0000000&0.000000&0.000000\\\hline
		\end{tabular}
	\end{tiny}
\end{table}	
\begin{table}[h!]
	\centering
	\begin{tiny}
		\begin{tabular}{|c|c|c|c|c|c|c|c|c|c|c|} \hline
			$n$& $\beta_{n,4}$&$\beta_{n,5}$ & $\beta_{n,6}$& $\beta_{n,7}$& $\beta_{n,8}$& $\gamma^{[0]}_{n}$&$\gamma^{[1]}_{n}$&$\gamma^{[2]}_{n}$&$\Psi_{n}^{queue}$& $\Psi_{n}^{sys}$\\\hline
			0&0.003418&0.001536&0.000545&0.000122&0.000014&0.022502&0.000000&0.000000&0.246688&0.062698\\
			1&0.001966&0.001083&0.000466&0.000125&0.000017&0.003198&0.008965&0.000000&0.086627&0.042611\\
			2&0.004305&0.002166&0.000880&0.000225&0.000030&0.007915&0.000621&0.017416&0.209642&0.111174\\
			3&0.002760&0.001632&0.000756&0.000215&0.000031&0.002185&0.001491&0.000475&0.063882&0.066385\\
			4&0.003503&0.001965&0.000892&0.000252&0.000037&0.002935&0.000203&0.001121&0.082311&0.106154\\
			5&0.002557&0.001590&0.000778&0.000232&0.000036&0.001142&0.000255&0.000061&0.049300&0.070072\\
			10&0.001306&0.000866&0.000458&0.000146&0.000026&0.000185&0.000002&0.000000&0.021455&0.046736\\
			20&0.000196&0.000140&0.000079&0.000027&0.000005&0.000002&0.000000&0.000000&0.002962&0.007492\\
			50&0.000000&0.000000&0.000000&0.000000&0.000000&0.000000&0.000000&0.000000&0.000011&0.000026\\
			70&0.000000&0.000000&0.000000&0.000000&0.000000&0.000000&0.000000&0.000000&0.000000&0.000000\\
			$\geq100$&0.000000&0.000000&0.000000&0.000000&0.000000&0.000000&0.000000&0.000000&0.000000&0.000000\\\hline
		\end{tabular}	
	\end{tiny}
\end{table}
\begin{table}[h!]
	\centering
	\begin{tiny}
		\caption{ Probability distribution at departure epoch for multiple vacations}\lb{tb5}$\vspace{0.03cm}$
		\begin{tabular}{|c|c|c|c|c|c|c|c|c|c|c|} \hline
			$n$&$\alpha^{+}_{n,3}$&$\alpha^{+}_{n,4}$ & $\alpha^{+}_{n,5}$& $\alpha^{+}_{n,6}$& $\alpha^{+}_{n,7}$& $\alpha^{+}_{n,8}$& $\alpha^{+}_{n}$& $\beta^{+}_{n,1}$&$\beta^{+}_{n,2}$&$\beta^{+}_{n,3}$\\\hline
			0&0.021629&0.017939&0.007577&0.009022&0.005888&0.006040&0.062659&0.008829&0.011494&0.008000\\
			1&0.006320&0.006233&0.002612&0.002767&0.001788&0.006258&0.025978&0.004429&0.006149&0.004713\\
			2&0.015787&0.015909&0.006622&0.006956&0.004491&0.009604&0.059369&0.011343&0.015630&0.011639\\
			3&0.005029&0.006667&0.002566&0.002426&0.001544&0.008834&0.027066&0.005499&0.008177&0.006797\\
			4&0.006901&0.009586&0.003608&0.003381&0.002148&0.009203&0.034827&0.007719&0.011442&0.009322\\
			5&0.003085&0.005504&0.001853&0.001653&0.001039&0.008136&0.021270&0.004211&0.006674&0.006000\\
			10&0.000735&0.002713&0.000653&0.000507&0.000311&0.004185&0.009140&0.001596&0.002761&0.002776\\
			20&0.000023&0.000419&0.000048&0.000027&0.000016&0.000682&0.001215&0.000179&0.000335&0.000371\\
			50&0.000000&0.000002&0.000000&0.000000&0.000000&0.000002&0.000004&0.000000&0.000001&0.000001\\
			70&0.000000&0.000000&0.000000&0.000000&0.000000&0.000000&0.000002&0.000000&0.000000&0.000000\\
			$\geq100$&0.000000&0.000000&0.000000&0.000000&0.000000&0.000000&0.000000&0.0000000&0.000000&0.000000\\\hline
		\end{tabular}
	\end{tiny}
\end{table}	
\begin{table}[h!]
	\centering
	\begin{tiny}
		\begin{tabular}{|c|c|c|c|c|c|c|c|c|c|c|} \hline
			$n$& $\beta^{+}_{n,4}$&$\beta^{+}_{n,5}$ & $\beta^{+}_{n,6}$& $\beta^{+}_{n,7}$& $\beta^{+}_{n,8}$& $\beta^{+}_{n}$& $\gamma^{[0]_+}_{n}$&$\gamma^{[1]_+}_{n}$&$\gamma^{[2]_+}_{n}$&$\gamma^{+}_{n}$\\\hline
			0&0.003721&0.001587&0.000502&0.000098&0.000008&0.034230&0.022466&0.000000&0.000000&0.022466\\
			1&0.002584&0.001304&0.000487&0.000111&0.000011&0.019788&0.006081&0.031747&0.000000&0.037828\\
			2&0.006012&0.002832&0.001004&0.000221&0.000021&0.048702&0.015053&0.006279&0.128430&0.149762\\
			3&0.004149&0.002259&0.000908&0.000221&0.000023&0.028033&0.004155&0.015087&0.020015&0.039257\\
			4&0.005477&0.002861&0.001128&0.000273&0.000028&0.038250&0.005582&0.002059&0.047248&0.054892\\
			5&0.004038&0.002339&0.001000&0.000257&0.000028&0.024547&0.002171&0.002579&0.002562&0.007312\\
			10&0.002127&0.001335&0.000623&0.000173&0.000022&0.011413&0.000352&0.000024&0.000013&0.000389\\
			20&0.000316&0.000215&0.000108&0.000032&0.000004&0.001560&0.000004&0.000000&0.000000&0.000004\\
			50&0.000001&0.000000&0.000000&0.000000&0.000000&0.000003&0.000000&0.000000&0.000000&0.000000\\
			70&0.000000&0.000000&0.000000&0.000000&0.000000&0.000000&0.000000&0.000000&0.000000&0.000000\\
			$\geq100$&0.000000&0.000000&0.000000&0.000000&0.000000&0.000000&0.000000&0.000000&0.000000&0.000000\\\hline
		\end{tabular}	
	\end{tiny}
\end{table}
\newpage
\begin{table}[h!]
	\centering
	\begin{tiny}
		\caption{ Probability distribution at arbitrary  epoch for multiple  vacation}\lb{tb6}$\vspace{0.03cm}$		
		\begin{tabular}{|c|c|c|c|c|c|c|c|c|c|c|} \hline
			$n$&$\alpha_{n,3}$&$\alpha_{n,4}$ & $\alpha_{n,5}$& $\alpha_{n,6}$& $\alpha_{n,7}$& $\alpha_{n,8}$& $\beta_{n,1}$&$\beta_{n,2}$&$\beta_{n,3}$& $\beta_{n,4}$\\\hline
			0&0.042530&0.069756&0.022808&0.021266&0.013490&0.011606&0.005722&0.008577&0.006629&0.003593\\
			1&0.007025&0.015271&0.004472&0.003870&0.002424&0.010512&0.002134&0.003481&0.003054&0.002041\\
			2&0.017554&0.038975&0.011298&0.009736&0.006096&0.013553&0.005360&0.008630&0.007266&0.004490\\
			3&0.005621&0.016333&0.004192&0.003428&0.002125&0.011999&0.002234&0.003936&0.003785&0.002844\\
			4&0.007713&0.023484&0.005893&0.004776&0.002956&0.011971&0.003107&0.005425&0.005062&0.003617\\
			5&0.003449&0.013484&0.003021&\textbf{0.002333}&0.001431&0.010394&0.001659&0.003089&0.003189&0.002616\\
			10&0.000822&0.006645&0.001065&0.000715&0.000428&0.005091&0.000630&0.001265&0.001436&0.001315\\
			20&0.000025&0.001028&0.000079&0.000039&0.000022&0.000816&0.000072&0.000155&0.000192&0.000193\\
			50&0.000000&0.000004&0.000000&0.000000&0.000000&0.000003&0.000000&0.000000&0.000000&0.000000\\
			70&0.000000&0.000000&0.000000&0.000000&0.000000&0.000000&0.000000&0.000000&0.000000&0.000000\\
			$\geq100$&0.000000&0.000000&0.000000&0.000000&0.000000&0.000000&0.0000000&0.000000&0.000000&0.000000\\\hline
		\end{tabular}
	\end{tiny}
\end{table}	
\noindent Here, \textbf{0.002333} denotes the probability that immediately after a departure, there are $5$ customers in the system and $6$ number of customers will be served by FES, under single vacation policy at arbitrary epoch.
\begin{table}[h!]
	\centering
	\begin{tiny}
		\begin{tabular}{|c|c|c|c|c|c|c|c|c|c|} \hline
			$n$&$\beta_{n,5}$ & $\beta_{n,6}$& $\beta_{n,7}$& $\beta_{n,8}$&  $\gamma^{[0]}_{n}$&$\gamma^{[1]}_{n}$&$\gamma^{[2]}_{n}$&$\Psi_{n}^{queue}$&$\Psi_{n}^{sys}$\\\hline
			0&0.001617&0.000564&0.000124&0.000014&0.034944&0.000000&0.000000&0.243240&0.061784\\
			1&0.001120&0.000476&0.000126&0.000017&0.004966&0.024690&0.000000&0.085680&0.042105\\
			2&0.002253&0.000901&0.000229&0.000030&0.012292&0.001709&0.066588&0.206961&0.109748\\
			3&0.001673&0.000766&0.000217&0.000031&0.003393&0.004107&0.001816&0.068500&0.067333\\
			4&0.002019&0.000905&0.000253&0.000037&0.004558&0.000560&0.004287&0.086625&0.106613\\
			5&0.001617&0.000783&0.000232&0.000036&0.001773&0.000702&0.000232&0.050042&0.071562\\
			10&0.000868&0.000456&0.000145&0.000026&0.000287&0.000006&0.000001&0.021203&0.046780\\
			20&0.000137&0.000077&0.000026&0.000005&0.000004&0.000000&0.000000&0.002870&0.007343\\
			50&0.000000&0.000000&0.000000&0.000000&0.000003&0.000000&0.000000&0.000010&0.000025\\
			70&0.000000&0.000000&0.000000&0.000000&0.000000&0.000000&0.000000&0.000000&0.000000\\
			$\geq100$&0.000000&0.000000&0.000000&0.000000&0.000000&0.000000&0.000000&0.000000&0.000000\\\hline
		\end{tabular}	
	\end{tiny}
\end{table}
\section{Sensitivity analysis}
This section is dedicated to showcase the impact of different key system parameters, including arrival rate, FES and SOS rate, vacation rate on significant performance metrics such as system throughput, energy saving factor, expected length of queue, system utility, expected cycle length and expected idle period. We present a graphical depiction of the trade-offs among these performance indices. Every parameters have been meticulously changed to evaluate its impact, ensuring the system remains in a steady state condition throughout the experiment. This method enables a thorough analysis of parameter sensitivities and their individual contributions to the system's overall performance, so that decision on managerial implications can be taken in a user-friendly manner.
\subsection{Impact of service rates on the trade-off between system throughput and energy saving factor} \lb{AA}
\noindent \hspace*{0.3cm} In this case, the values of parameters such as the minimum threshold value $(a)$ and maximum threshold $(b)$ are taken as $4$ and $10$, respectively, and the arrival rate ($\lambda$) varies from $[0.15, 0.175, 0.20, 0.225, 0.25, 0.275, 0.30]$. The FES time is considered as follow deterministic distribution with rate of service of batch of size  $`r$' where $4 \leq r \leq 10$ and mean service times are $\frac{1}{\mu_r}=(r+1)$ units i.e.,  $K^{r}(z_1)=(\bar{\lambda}+\lambda G(z_1))^{r+1}$, respectively where $G(z_1)=(0.4z_1+0.3z_1^2+0.3z_1^3) $. The vacation time follows geometric distribution with rate $\eta_{k}=(k+1)0.25$ where $k=0$ to $a-1$, and also the SOS time is considered as negative binomial distribution with rate of service of a batch of size $`i$' is $\mu_i^O=\frac{\mu}{i+2}$, where $\mu$ varies from $0.95(0.05)1.45$ and $~1\le i \le 10.$ The trade-off between system throughput and energy saving factor is displayed in figures [\ref{fig:subfig2A} - \ref{fig:subfig2F}] for various service rates chosen in such a way that the system stability is preserved throughout the experiment. \\
From Figures [\ref{fig:subfig2A} - \ref{fig:subfig2F}], we observe the following:
\begin{itemize}
	\item Under both single and multiple vacation rules, the system throughput increases while the energy saving factor decreases as the value of $\mu$ increases. This phenomenon occurs because a higher service rate reduces the mean service time, which enables the server to finish the tasks faster and be accessible to new clients sooner. This reduces the amount of time that energy saving modes can be maintained because the server stays active for a greater percentage of the time. At the same time, more clients may be served in a given amount of time due to the faster service rates, which improves system throughput. Moreover, from figures [\ref{fig:subfig2A}–\ref{fig:subfig2F}], we conclude the optimal value of $\mu$ (approximately) such that it optimized the throughput performance and energy efficiency across various arrival rates for both single and multiple vacation policy.
	\item The system becomes overcrowded as the arrival rate increases because more customers arrive in a shorter amount of time. Because of this, the system will interact with users for a longer period of time, which will increase system throughput and decrease the energy saving factor.
\end{itemize}
\begin{figure}[ht!]
	\centering
	\begin{minipage}{0.35\textwidth}
		\hspace*{-0.2cm}\includegraphics[width=1.45\textwidth]{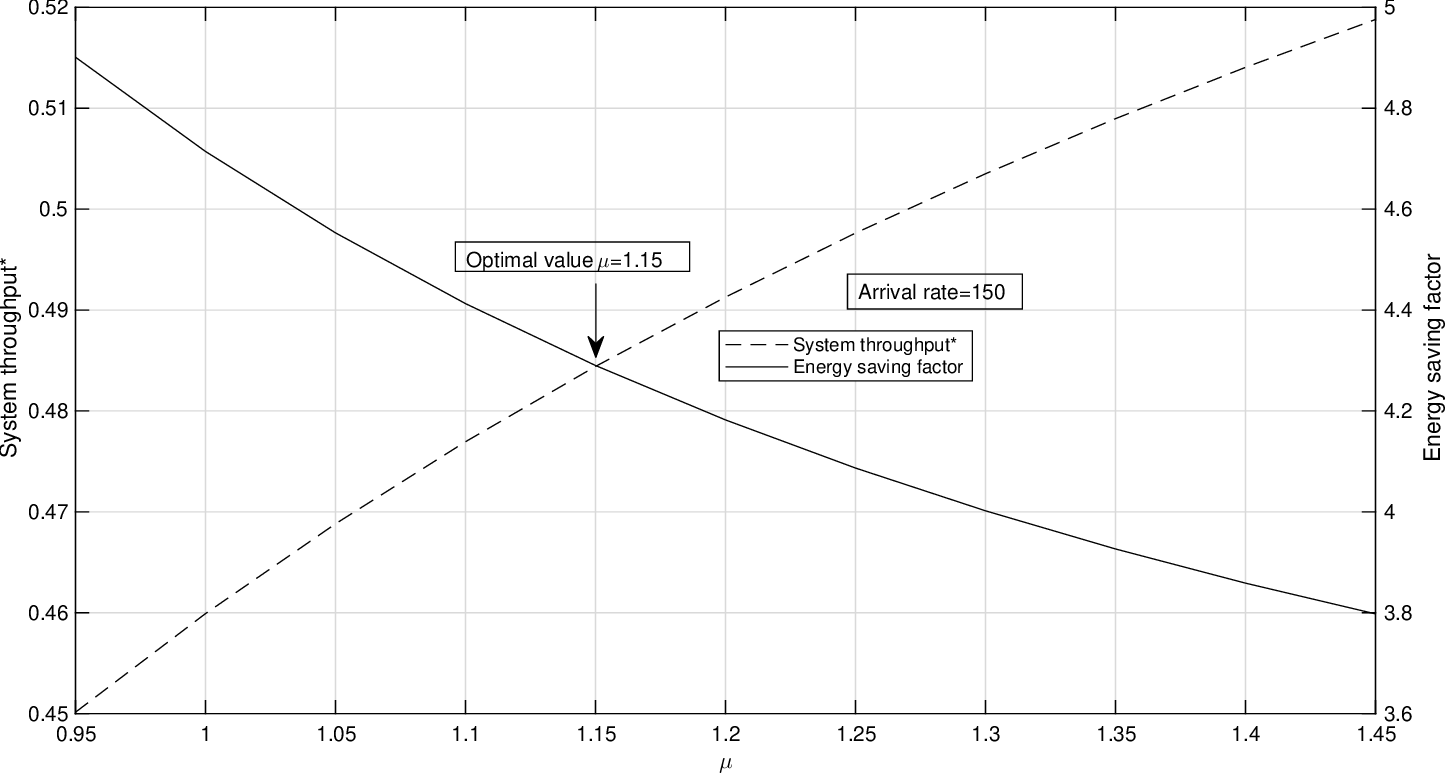}
		\subcaption{For single vacation}
		\label{fig:subfig2A}
	\end{minipage}	
	\hfill
	\begin{minipage}{0.35\textwidth}
		\hspace*{-2.0cm}\includegraphics[width=1.45\textwidth]{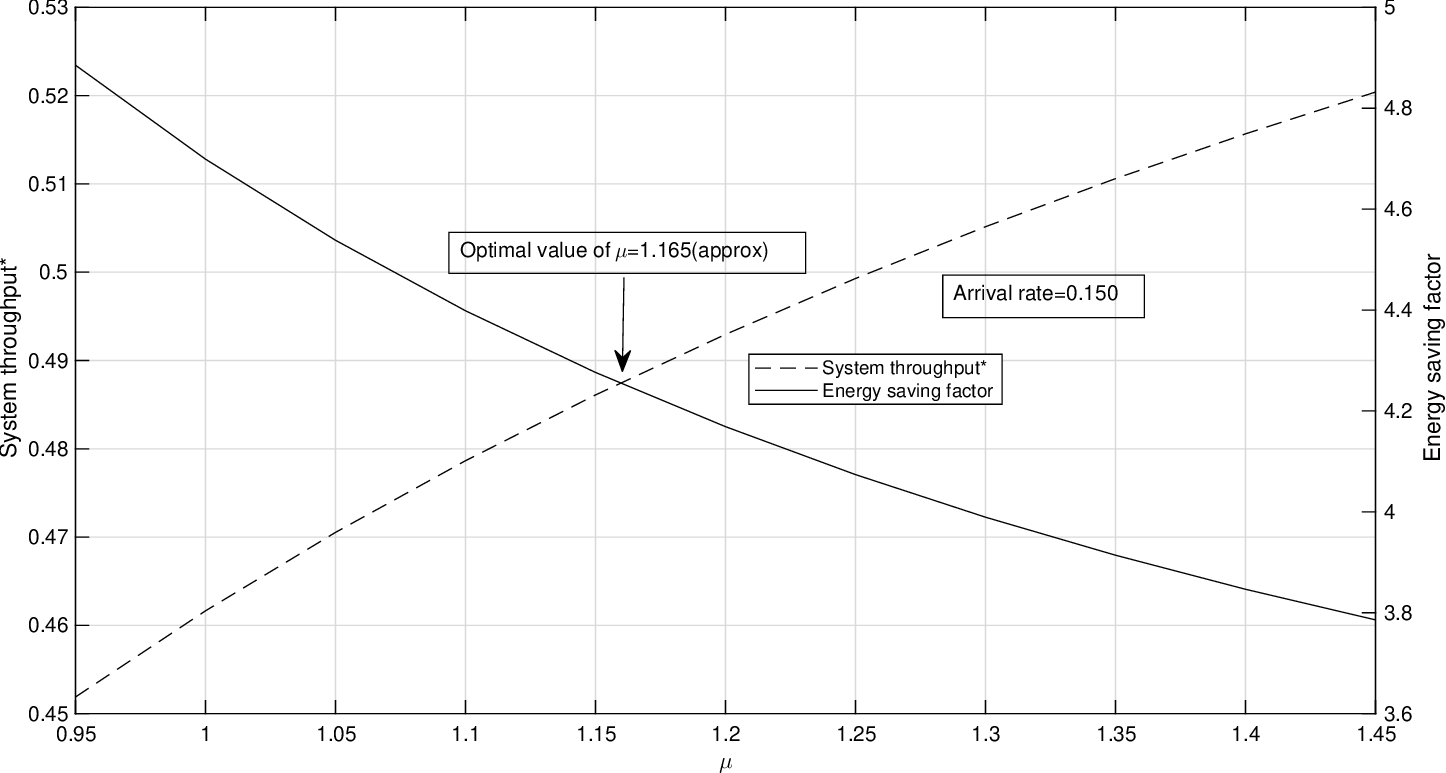}
		\subcaption{For multiple vacations}
		\label{fig:subfig2B}
	\end{minipage}
	\begin{minipage}{0.35\textwidth}
		\hspace*{-0.2cm}\includegraphics[width=1.45\textwidth]{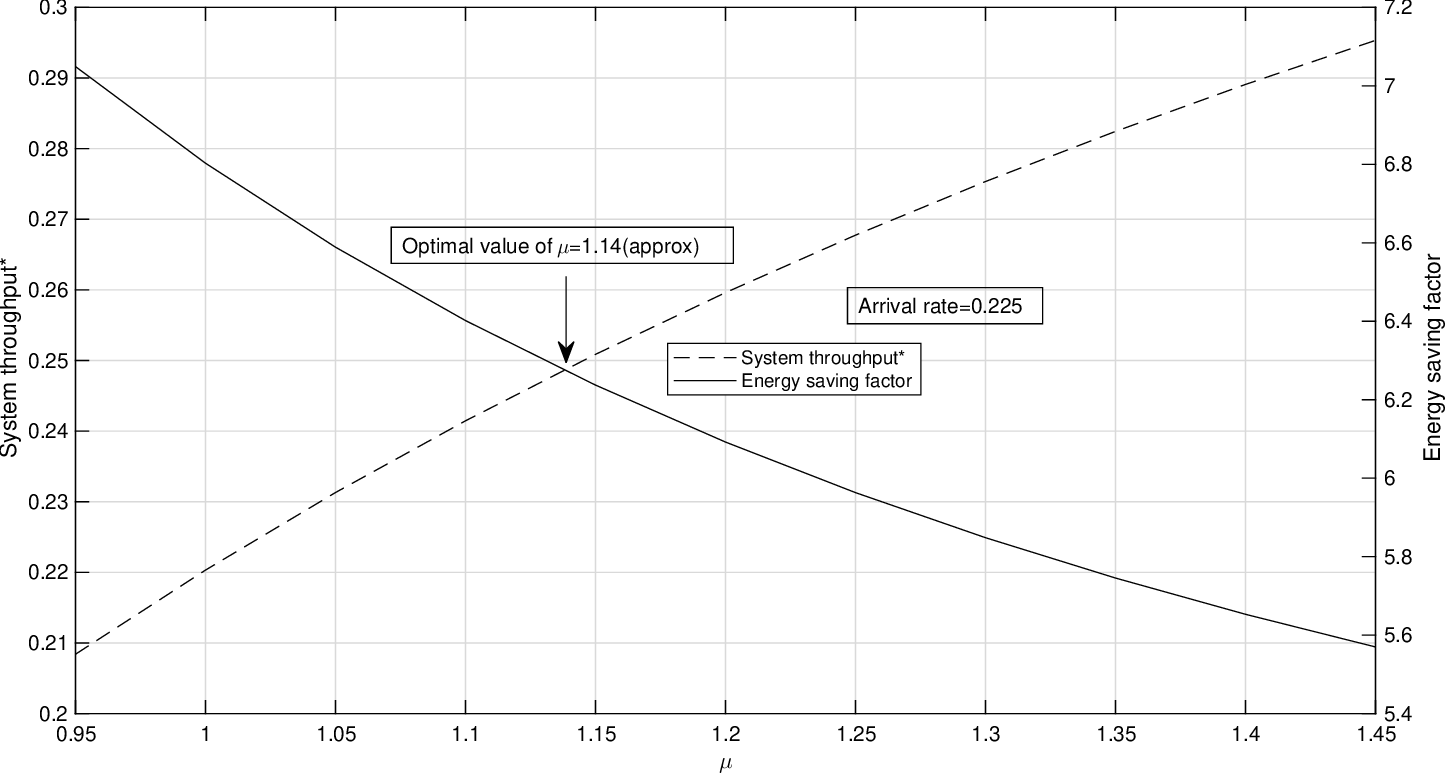}
		\subcaption{For single vacation}
		\label{fig:subfig2C}
	\end{minipage}	
	\hfill
	\begin{minipage}{0.35\textwidth}
		\hspace*{-2.0cm}\includegraphics[width=1.45\textwidth]{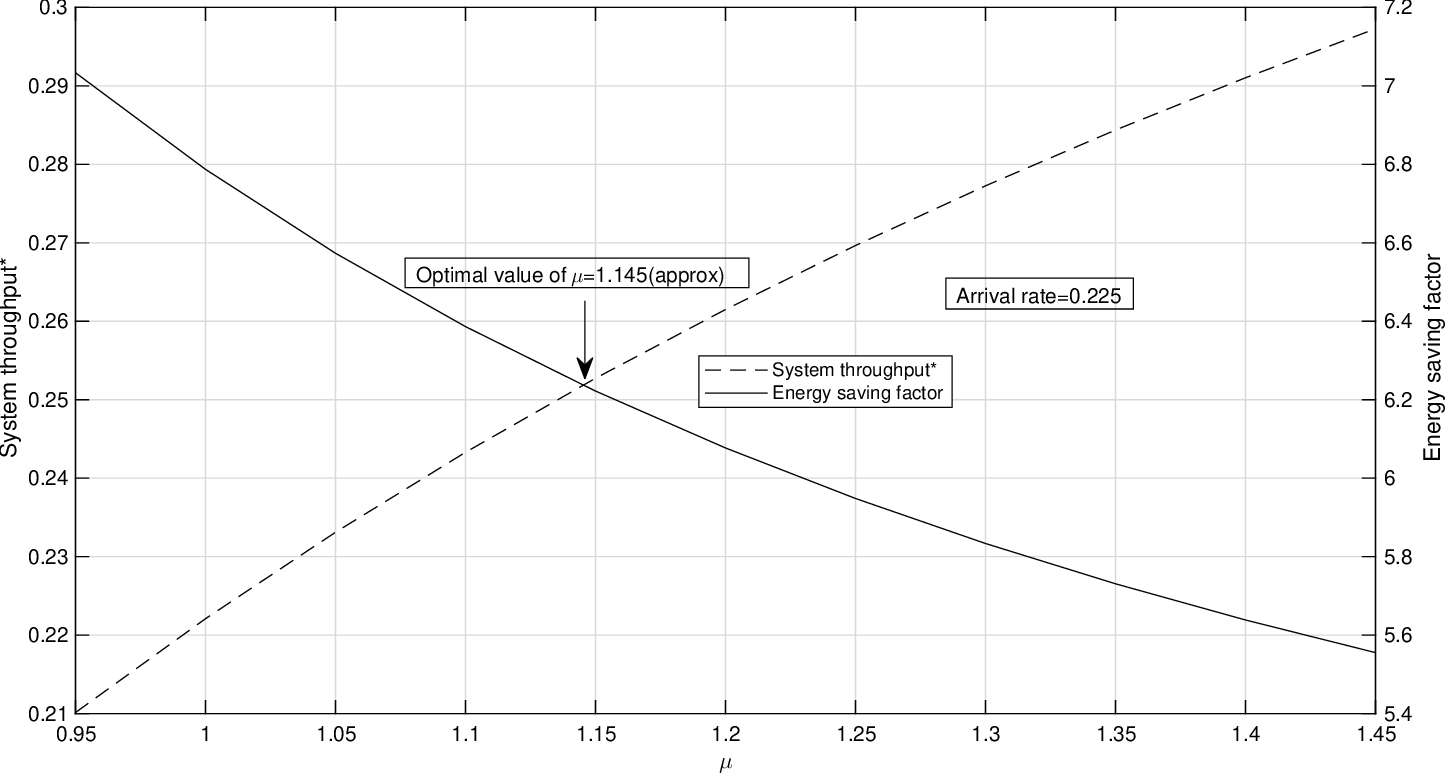}
		\subcaption{For multiple vacations}
		\label{fig:subfig2D}
	\end{minipage}
	\begin{minipage}{0.35\textwidth}
		\hspace*{-0.2cm}\includegraphics[width=1.45\textwidth]{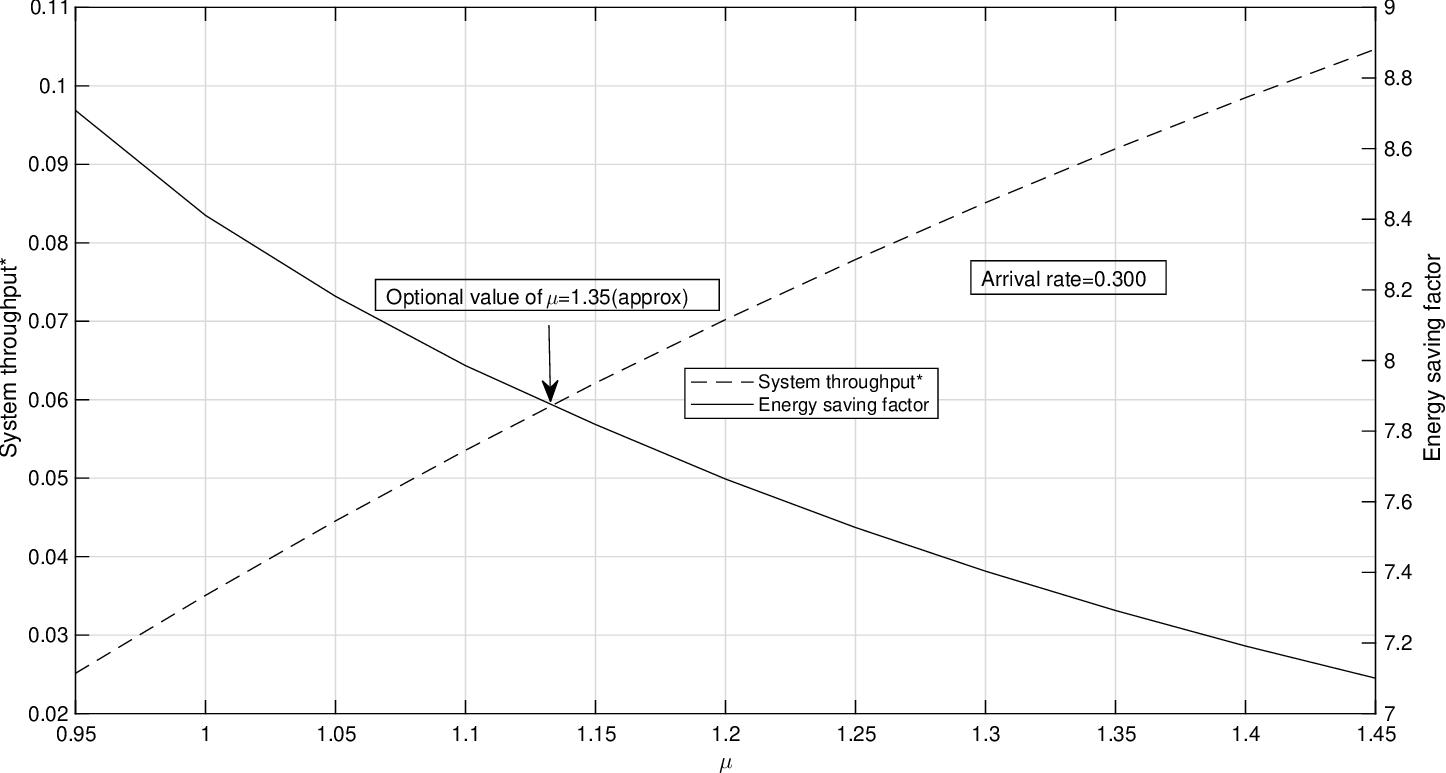}
		\subcaption{For single vacation}
		\label{fig:subfig2E}
	\end{minipage}	
	\hfill
	\begin{minipage}{0.35\textwidth}
		\hspace*{-2.0cm}\includegraphics[width=1.45\textwidth]{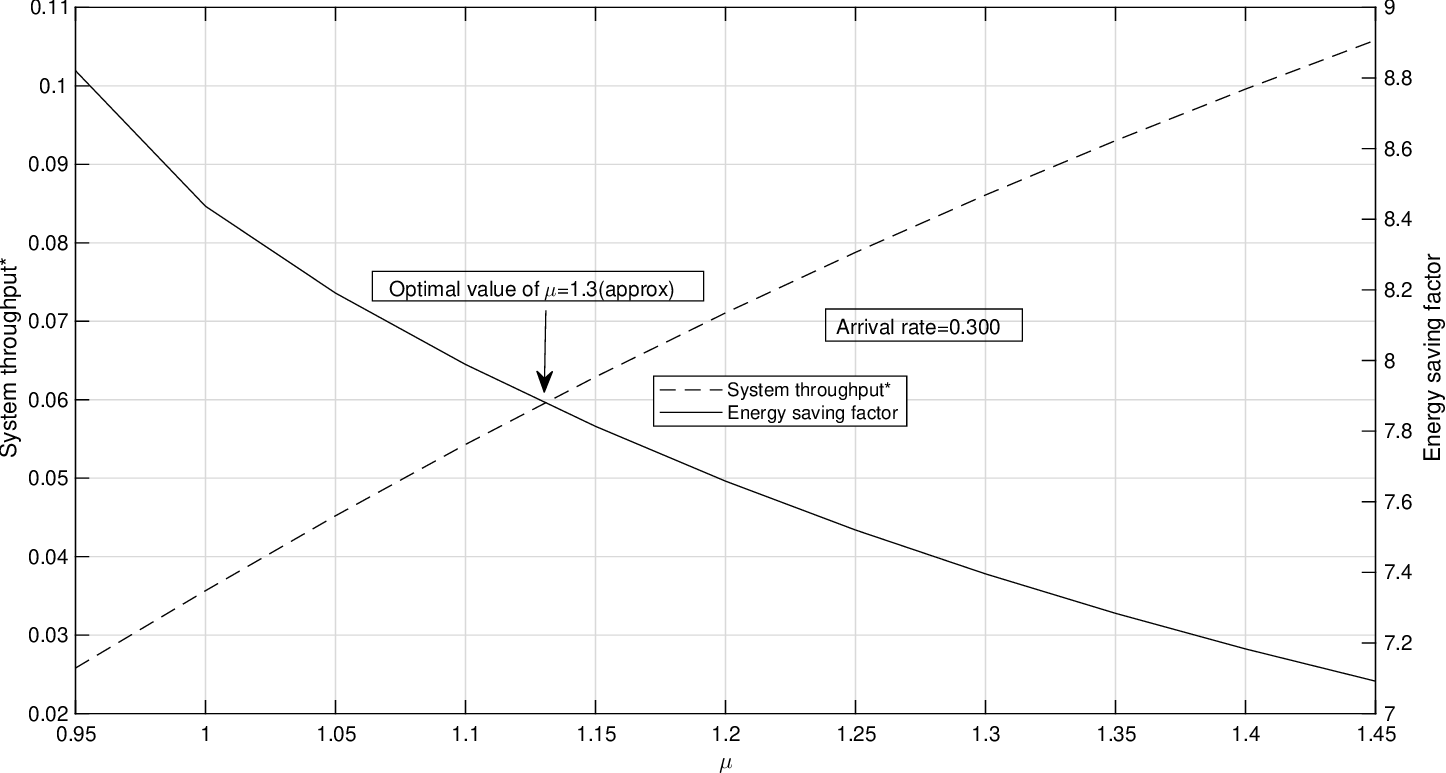}
		\subcaption{For multiple vacations}
		\label{fig:subfig2F}
	\end{minipage}
	\caption{Impact of service rates on the trade-off between system throughput and energy saving factor.}
	\label{fig:2}
\end{figure}
\subsection{Impact of vacation rates on the expected queue length and system utility for different FES time distribution}
The effect of the vacation rates $\eta_{k}$  varies from $0.025(0.025)0.225$ on the expected queue length is examined for different FES taken as deterministic(Det), geometric(Geo), negative binomial (NB) and general DPH distribution while keeping the SOS time distribution fixed as geometric distribution with rate of service of a batch of size $`i$' is $\mu^{O}_{i}=\dfrac{1.85}{i+2}$,  $1\leq i \leq 10$.\\
\begin{itemize}
	\item From figures [\ref{fig:subfig3A} - \ref{fig:subfig3B}], we observe that the expected queue length decreases consistently as the vacation rate increases, which is indicative of a reduction in the average vacation duration and subsequent enhancement of server's availability. The Det FES time distribution consistently produces the smallest expected queue length, while Geo, NB, and DPH FES time distribution having progressively larger queue length. From this, we conclude that if the service provider choose the service time as Det distribution then there is a possibility to reduce the length of queue.
	\item From figures [\ref{fig:subfig4A} - \ref{fig:subfig4B}], we observe that the system utility increases consistently as the vacation rate increases. DPH is taken as FES time distribution and it is consistently produces the smallest system utility and NB produces the highest system utility.
\end{itemize}
\begin{figure}[ht!]
	\centering
	\begin{minipage}{0.45\textwidth}
		\hspace*{0.0cm}\includegraphics[width=1.65\textwidth]{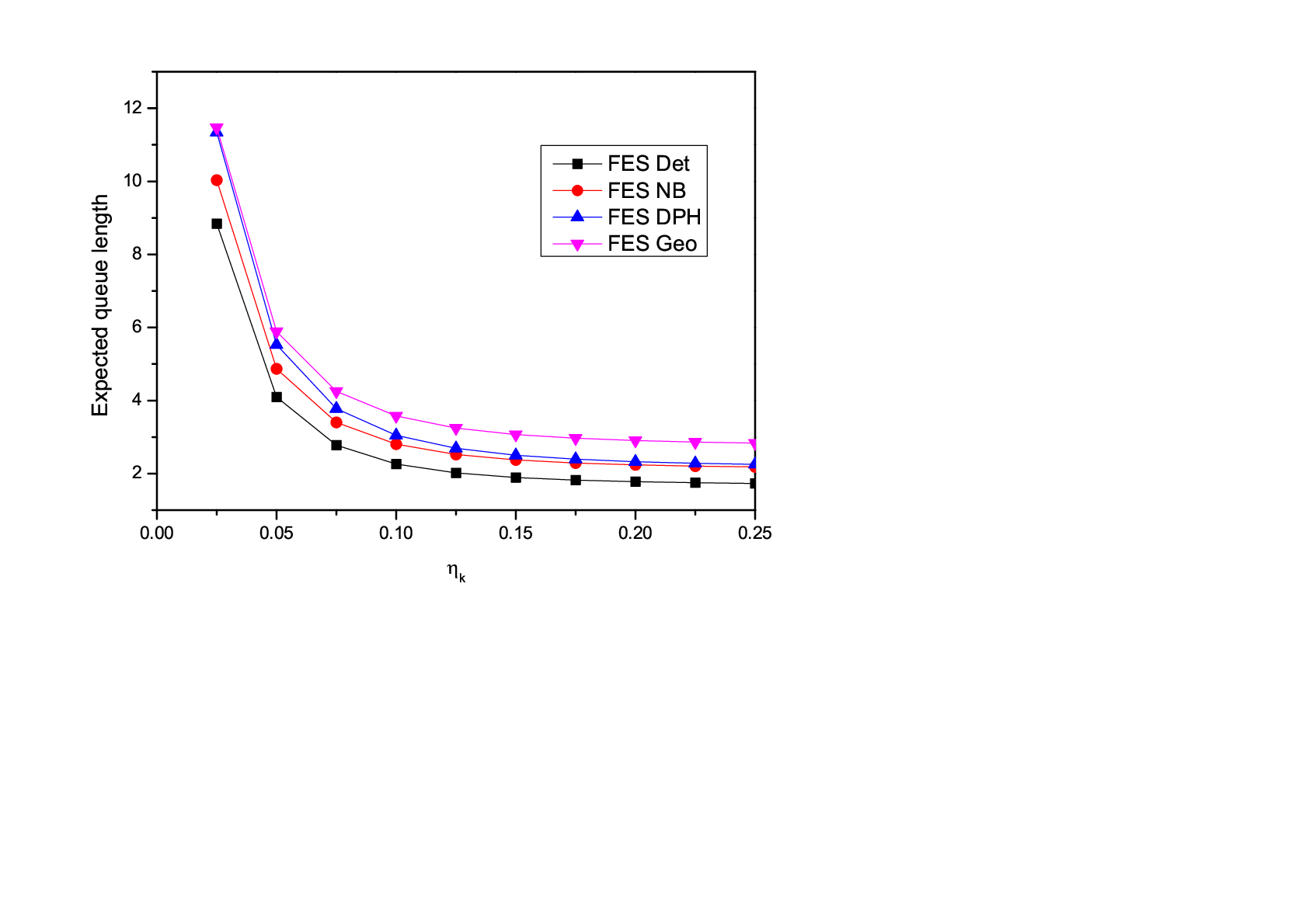}\vspace{-2.7cm}
		\subcaption{For single vacation}
		\label{fig:subfig3A}
	\end{minipage}	
	\hfill
	\begin{minipage}{0.45\textwidth}
		\hspace*{-0.5cm}\includegraphics[width=1.65\textwidth]{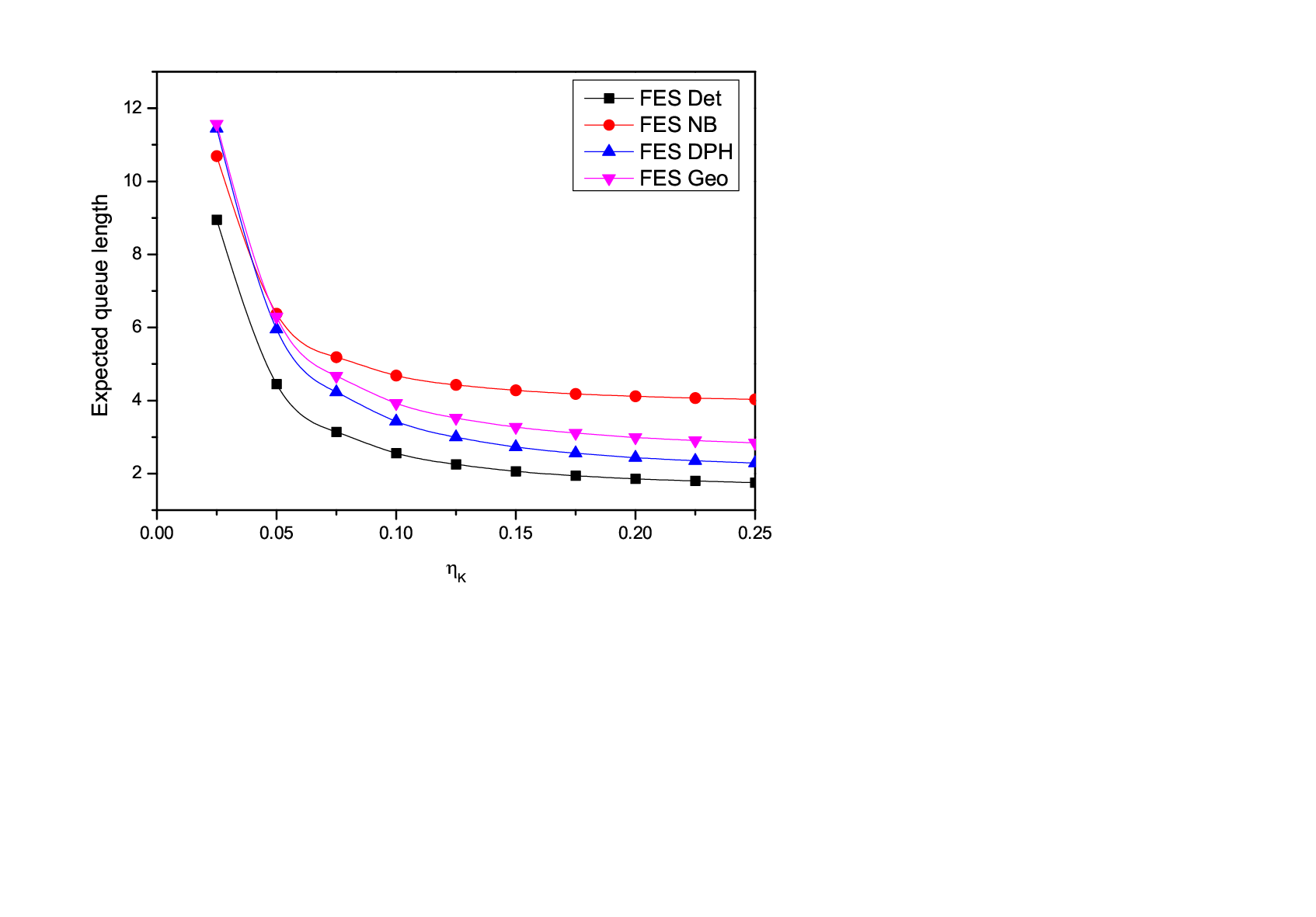}\vspace{-2.7cm}
		\subcaption{For multiple vacations~~~~~~~~~~~~~~}
		\label{fig:subfig3B}
	\end{minipage}
	\caption{Impact of vacation rates on the expected queue length and system utility for different FES time distribution.}
	\label{fig:3}
\end{figure} 
\begin{figure}[ht!]
	\centering
	\begin{minipage}{0.45\textwidth}
		\hspace*{0.0cm}\includegraphics[width=1.65\textwidth]{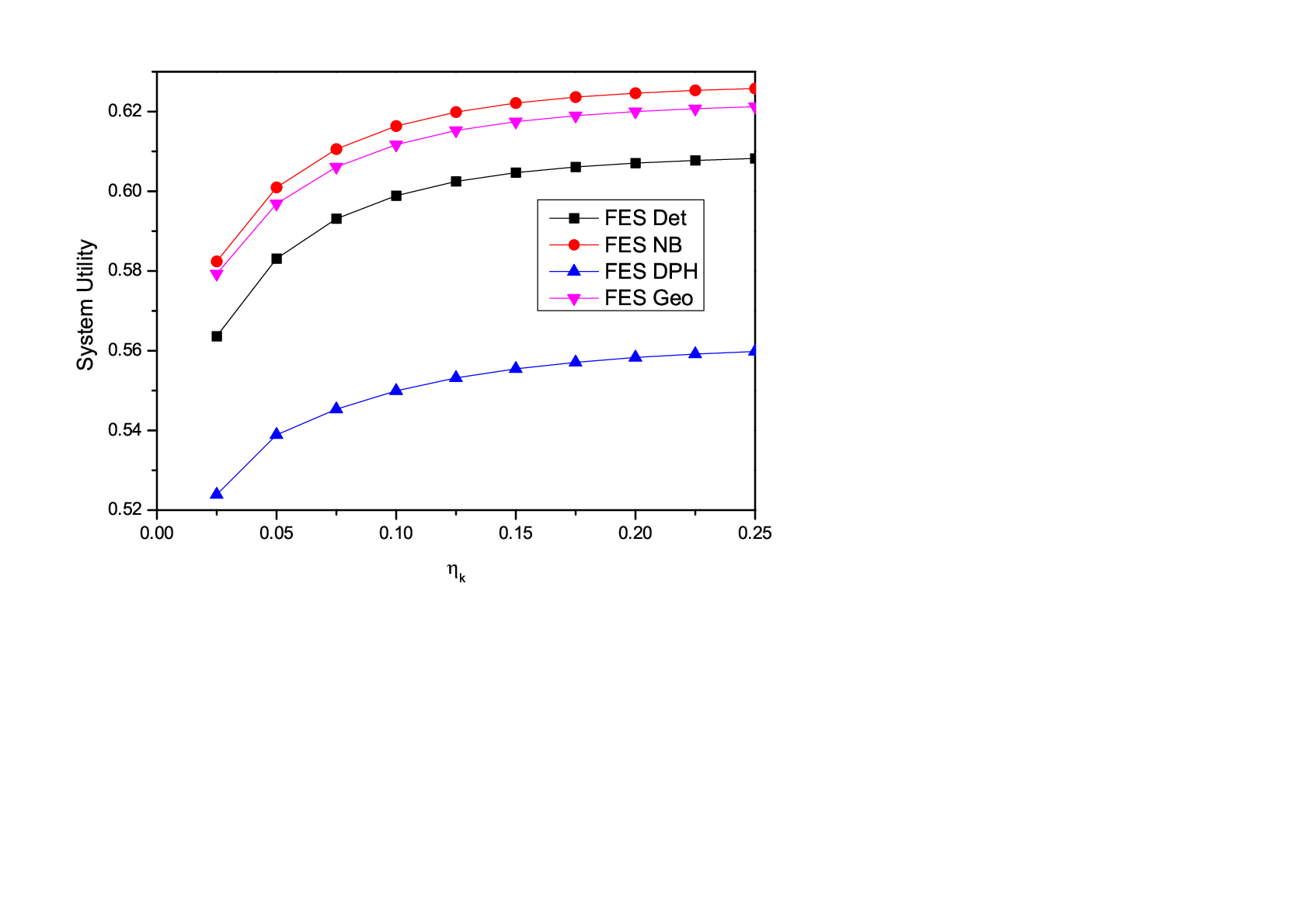}\vspace{-2.7cm}
		\subcaption{For single vacation}
		\label{fig:subfig4A}
	\end{minipage}	
	\hfill
	\begin{minipage}{0.45\textwidth}
		\hspace*{-0.5cm}\includegraphics[width=1.65\textwidth]{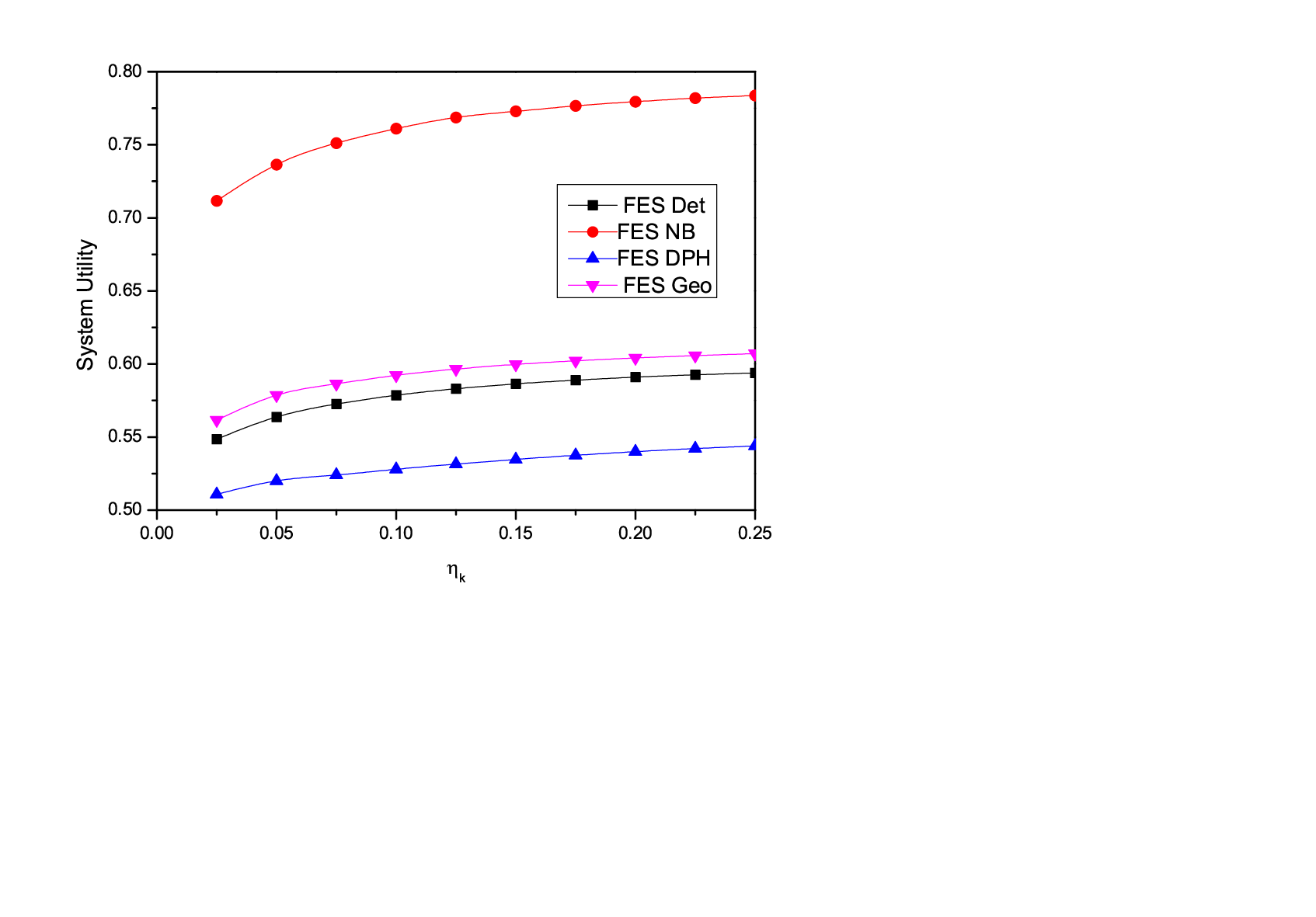}\vspace{-2.7cm}
		\subcaption{For multiple vacations}
		\label{fig:subfig4B}
	\end{minipage}
	\caption{Impact of vacation rates on the system utility for different FES time distribution.}
	\label{fig:4}
\end{figure}
\subsection{Impact of vacation rates on the expected queue length and system utility for different SOS time distribution}
The effect of the vacation rate $\eta_{k}$ from $0.025(0.025)0.225$ on the expected queue length is examined for different SOS time distribution taken as Det, Geo, NB and DPH distributions while keeping the FES time distribution fixed which follows  geometric distribution with the rate of service of a batch of size $`i$' is $\mu_{i}=\dfrac{0.95}{i+1}$,  $4\leq i \leq 10$.\\
\begin{itemize}
	\item From figures [\ref{fig:subfig5A} - \ref{fig:subfig5B}], we observe that the expected queue length decreases consistently as the vacation rate increases. The Geo and DPH distribution for SOS time distribution produces the smallest expected queue length in the case of single and multiple vacations, respectively.
	\item From figures [\ref{fig:subfig6A} - \ref{fig:subfig6B}], we observe that the system utility increases consistently as the vacation rate increases. The deterministic distribution taken as SOS consistently produces the highest system utility among these four different distribution as the vacation rate $\eta_{k}$ increases. 
\end{itemize}
\begin{figure}[ht!]
	\centering
	\begin{minipage}{0.45\textwidth}
		\hspace*{0.0cm}\includegraphics[width=1.60\textwidth]{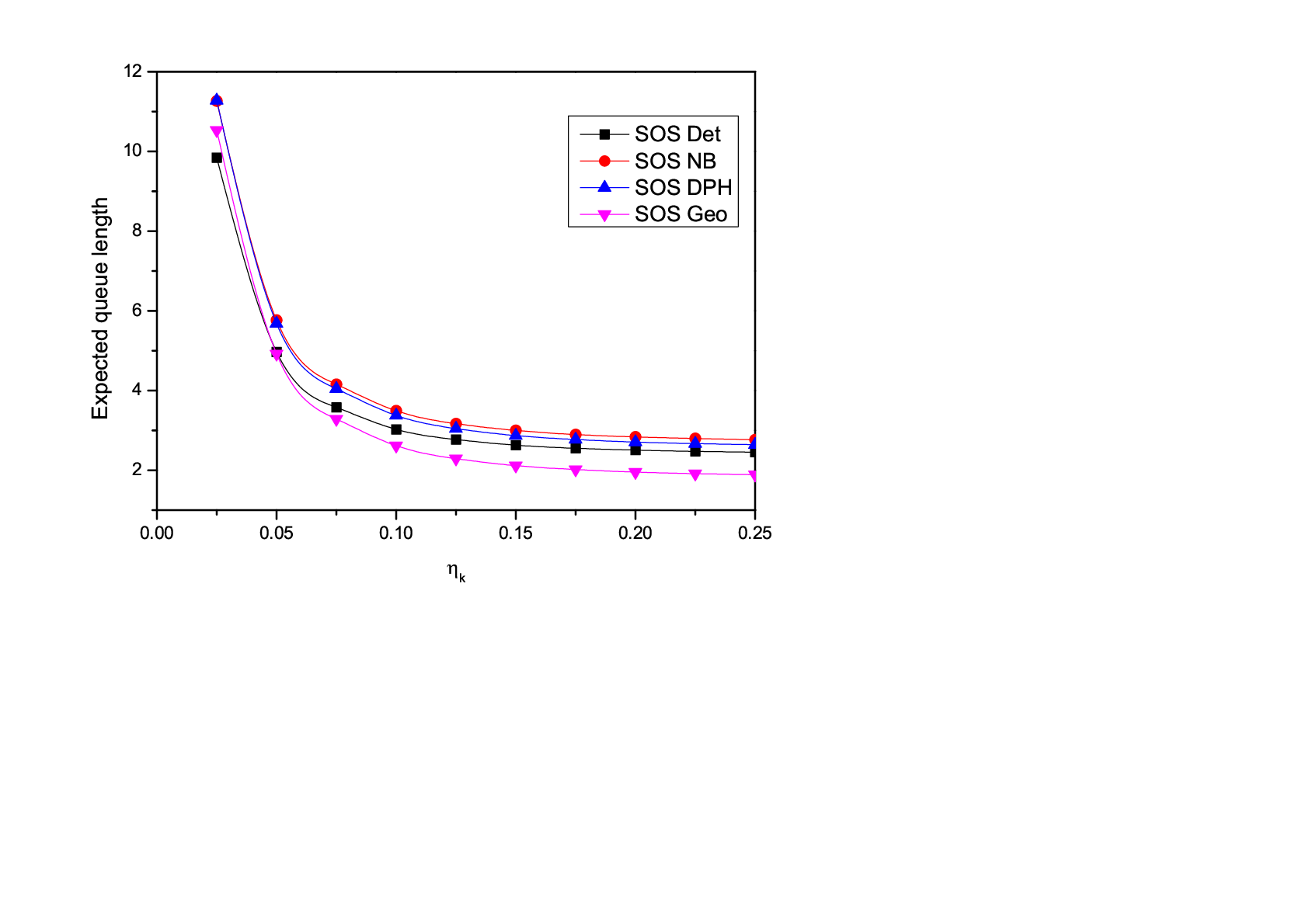}\vspace{-2.7cm}
		\subcaption{For single vacation}
		\label{fig:subfig5A}
	\end{minipage}	
	\hfill
	\begin{minipage}{0.45\textwidth}
		\hspace*{-0.5cm}\includegraphics[width=1.60\textwidth]{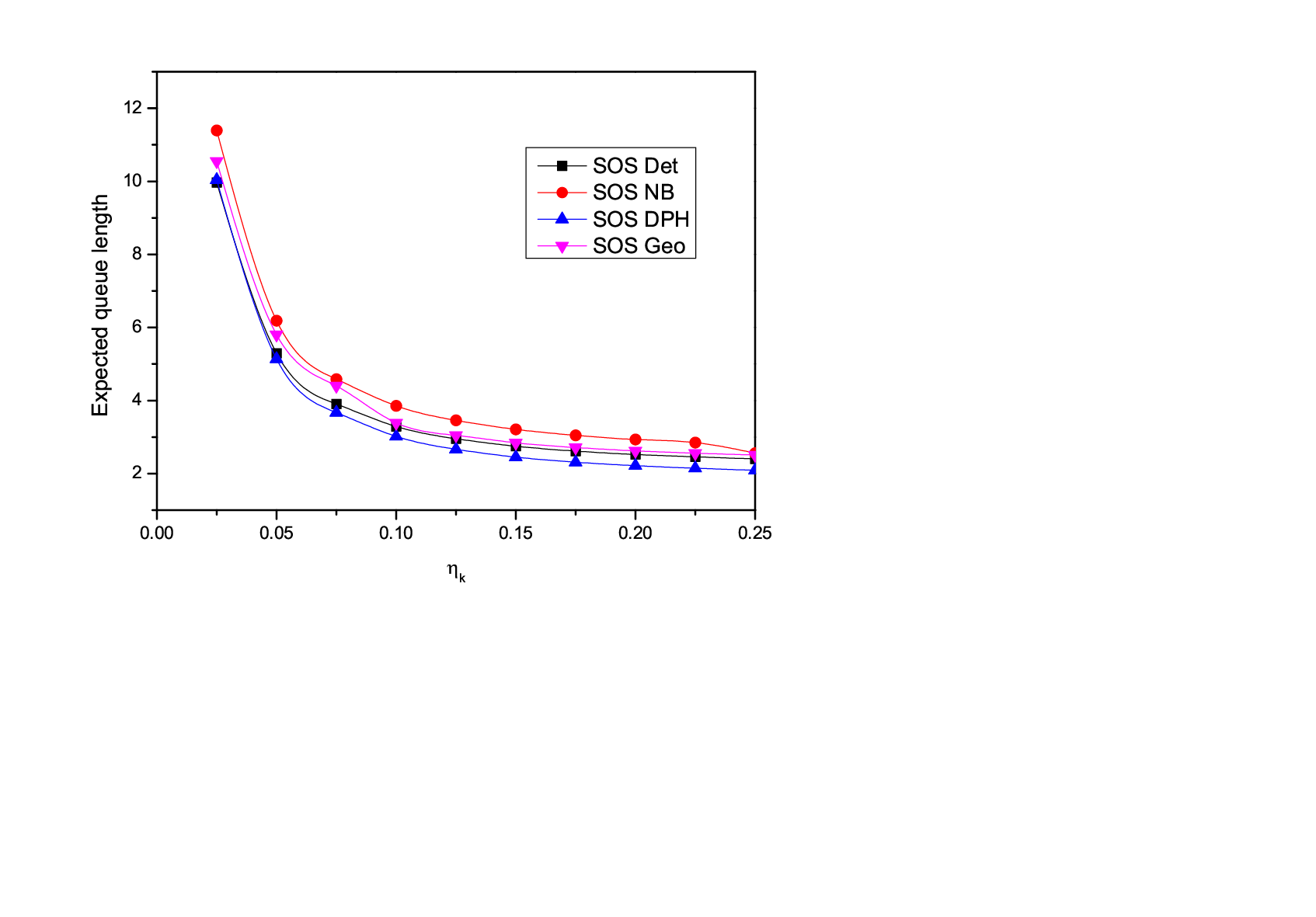}\vspace{-2.7cm}
		\subcaption{For multiple vacations}
		\label{fig:subfig5B}
	\end{minipage}
	\caption{Impact of vacation rate on the expected  length of queue for different SOS and with fixed FES.}
	\label{fig:5}
\end{figure}
\begin{figure}[h!]
	\centering
	\begin{minipage}{0.45\textwidth}
		\hspace*{0.0cm}\includegraphics[width=1.65\textwidth]{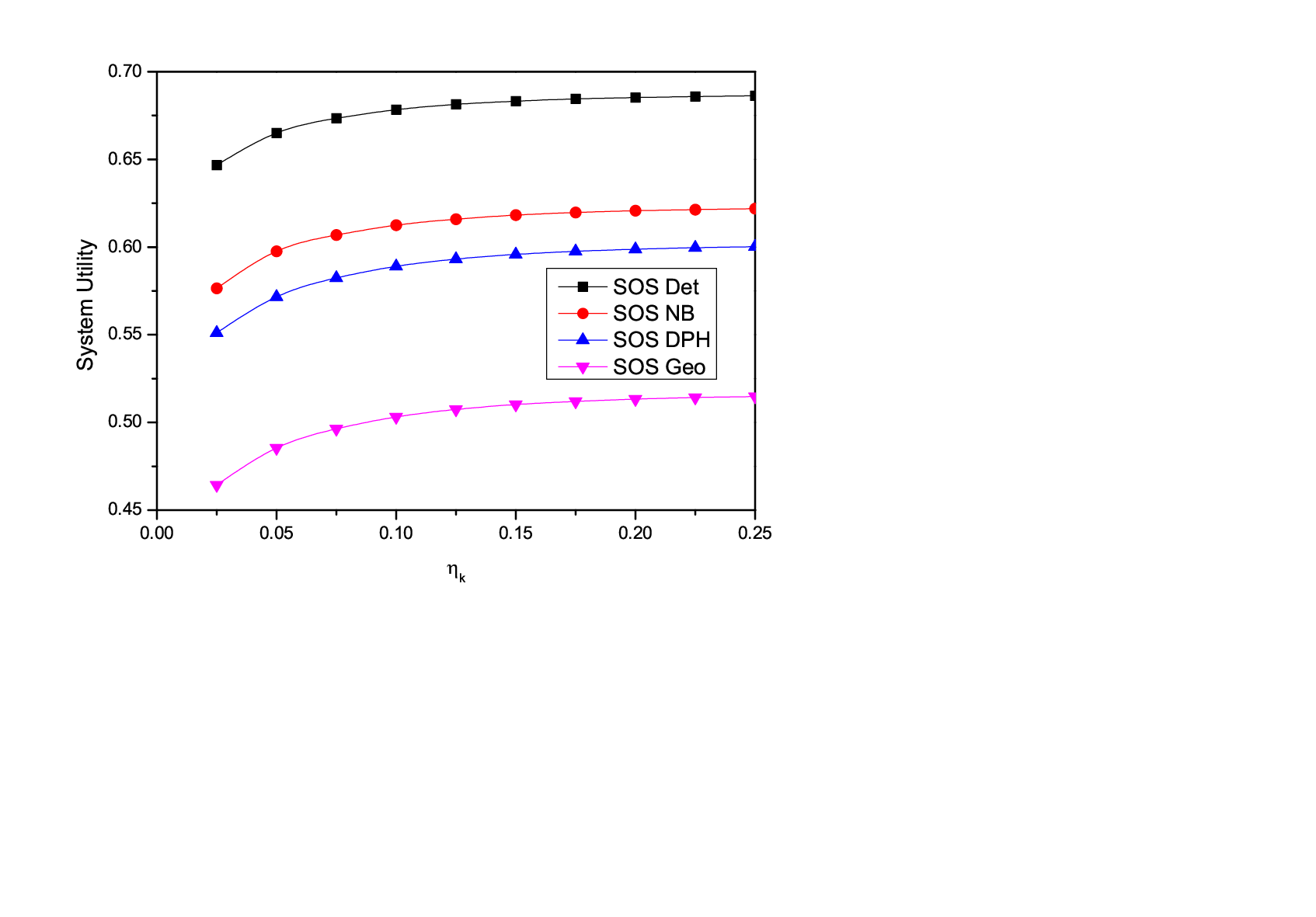}\vspace{-2.7cm}
		\subcaption{For single vacation}
		\label{fig:subfig6A}
	\end{minipage}	
	\hfill
	\begin{minipage}{0.45\textwidth}
		\hspace*{-0.5cm}\includegraphics[width=1.65\textwidth]{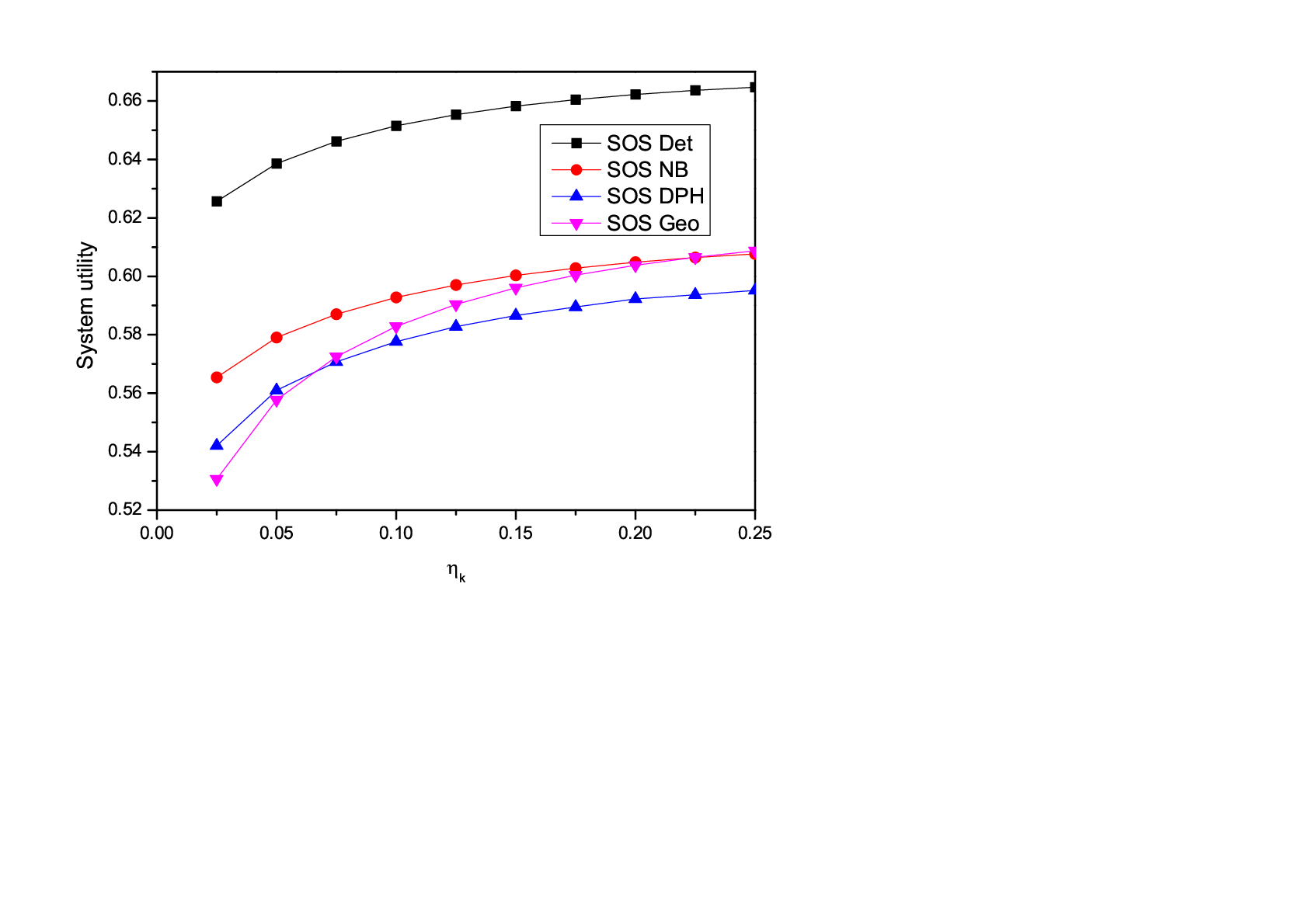}\vspace{-2.7cm}
		\subcaption{For multiple vacations}
		\label{fig:subfig6B}
	\end{minipage}
	\caption{Impact of vacation rate on the system utility for different SOS and with fixed FES.}
	\label{fig:6}
\end{figure}
\subsection{Impact of service rates on expected idle period and expected cycle length as arrival rate varies}
In this subsection, our main objective is to investigate the impact of service rates on the two magnificent performance measures like expected idle period and expected cycle length with the dynamical change of arrival rates. Here FES time, SOS time and vacation time, all follow the same distribution as considered in Subsection [\ref{AA}]. All the values of systems parameter like $`a$', $`b$' and $`\lambda$' are also same as mention in Subsection [\ref{AA}].
Figures [\ref{fig:7} - \ref{fig:8}], respectively, display the dynamic change of the values of expected idle period and expected cycle length for the variation of the values of $\mu$ and $\lambda$. From Figures [\ref{fig:7} - \ref{fig:8}], we conclude the following observations:  
\begin{itemize}
	\item When the service rates are fixed, a higher arrival rate consistently leads to the decreased expected idle period in the system. This results matches with the real life scenarios, as a higher arrival rate produces more customers entering into the system, keeping the server busy for more duration i.e., with less idle time. 
	\item On the other hand, the opposite scenario is observed when the service rates vary while the arrival rate remains fixed. In this case, the expected idle period increases. 
	\item  Also, it is interesting to note that the expected idle period is lower when the system operates under multiple vacations policy as compared to single vacation policy. This observation also lies within expectation, as in multiple vacations scenario, the server frequently goes for vacation if the required number of customers is not present in the queue at the end of each vacation duration. During this period, only the sensor component remains active, monitoring the queue length at each vacation's conclusion. Conversely, under a single vacation policy, at the end of vacation period, the server is still present and waiting for minimum threshold to be accumulated. The similar observations are made while we vary arrival rates and investigate the impact of service rates on expected cycle length, which is exhibited in [\ref{fig:8}]. 
\end{itemize}
\begin{figure}[!htb]
	\noindent\begin{minipage}{.5\textwidth}
		\centering
		\hspace*{-1.2cm}\includegraphics[scale=0.43]{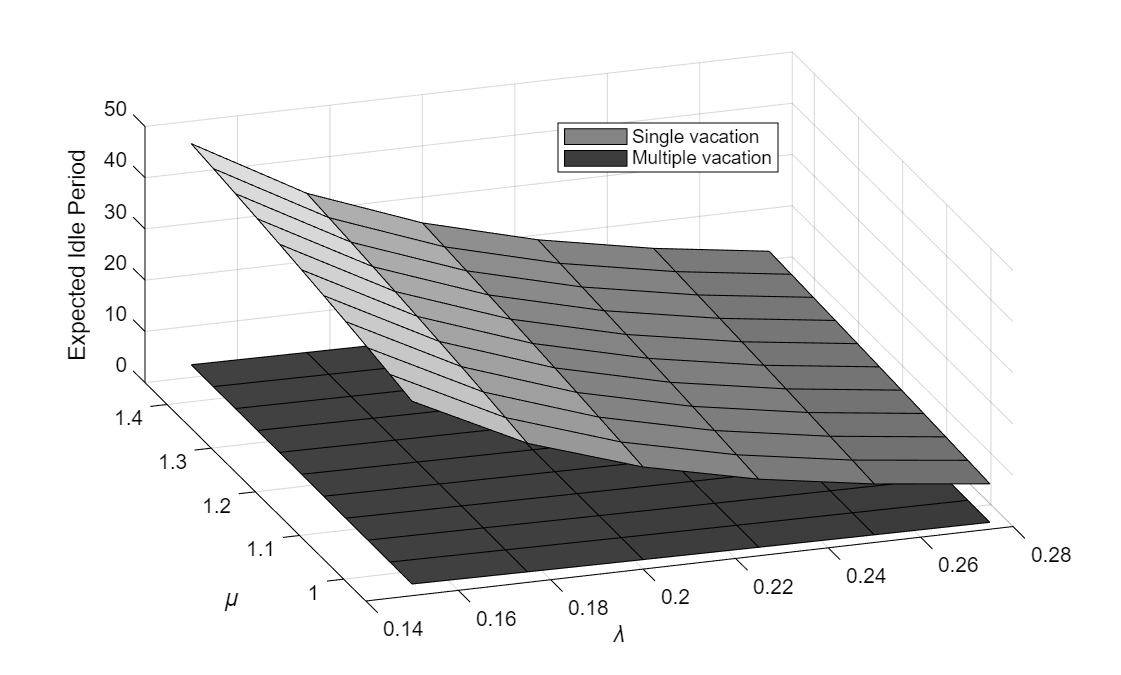}
		\captionof{figure}{\\ Expected idle period versus arrival rate \\ versus service rate.}\label{fig:7}
	\end{minipage}%
	\begin{minipage}{.5\textwidth}
		\centering
		\hspace*{-1.0cm}\includegraphics[scale=0.45]{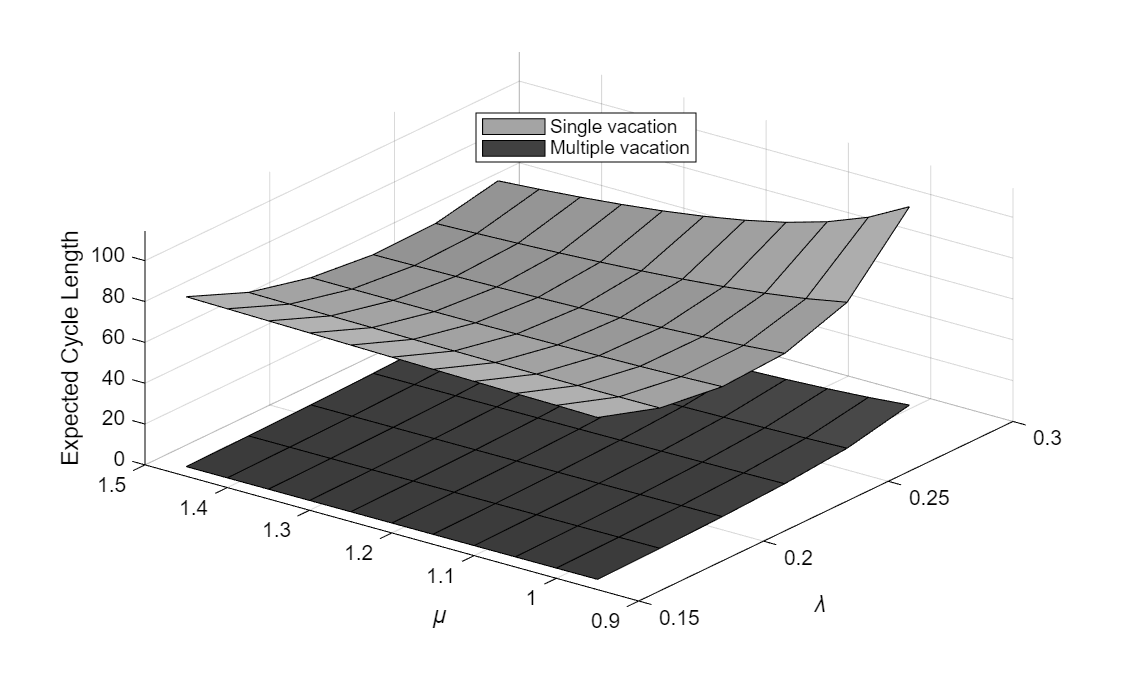}
		\captionof{figure}{\\ Expected cycle length versus arrival rate versus service rate.} \label{fig:8}
	\end{minipage}
\end{figure}
\section{Conclusion and future scopes}
In this article, we have modeled an infinite capacity discrete time vacation queue with the incorporation of group arrival, queue length dependent vacation, first essential and second optional service. The bivariate probability generating functions for stationary joint probability at both phases have been derived in a very simple and systematic way which avoids the construction of complex transition probability matrix. The developed queueing models, associated results and performance indices are very much beneficial in the concept of modeling of cloud computing, wireless telecommunication networks, health care systems, banking service systems etc. It also demonstrate the computational procedure for the readers and users. We have also included a vast amount of graphical observation in order to carry out the sensitivity check of the model parameters to the performance indices. In future, it will be a interesting research work to perform the aspects of modeling analysis and managerial implications of the queueing system with lag correlations taking places among inter arrivals which we may observe in various practical circumstances.  
\pagebreak


\end{document}